\documentclass[11pt]{article}
\usepackage[
    top=1in,
    bottom=1in,
    left=0.75in,
    right=0.75in
]{geometry}

\usepackage[T1]{fontenc}
\usepackage[ttscale=.875]{libertine}
\usepackage{microtype}
\usepackage{setspace}
\usepackage[title]{appendix}

\usepackage{amsfonts}
\usepackage{amsmath}
\usepackage{amssymb}
\usepackage{amsthm}
\usepackage{mathabx}
\usepackage{thmtools, thm-restate}
\usepackage{bbm}
\usepackage{bm}
\usepackage{dsfont}
\usepackage{fullpage}
\usepackage{tikz}
\usetikzlibrary{tikzmark,calc,decorations.pathreplacing}
\usepackage[most]{tcolorbox}
\usepackage{xcolor}
\usepackage{ytableau}
\usepackage{multicol}

\newcommand*{\myfont}{\fontfamily{bch}\selectfont}
\DeclareTextFontCommand{\textmyfont}{\myfont}

\usepackage[linesnumbered,ruled,vlined]{algorithm2e} 

\SetKwComment{Comment}{/* }{ */}

\SetCommentSty{mycommfont}

\let\oldnl\nl
\newcommand{\nonl}{\renewcommand{\nl}{\let\nl\oldnl}}

\definecolor{algocitecolor}{RGB}{242, 118, 19}

\newcommand{\hpol}{\F_{2}[\mathbf{Z}]_{\Hom}^{\leq 1}}
\newcommand{\zero}[1]{{{#1}_{0}}}
\newcommand{\lines}[1]{{{#1}_{1}}}

\newcommand{\setml}{\Psi}
\newcommand{\rand}[1]{{\color{red}{#1}}}
\definecolor{darknavy}{HTML}{1F4E79}
\newcommand{\randsecond}[1]{{\color{darknavy}{#1}}}
\newcommand{\had}[2]{\mathrm{Had}^{(#1)}\left(#2\right)}
\newcommand{\hadfamily}[2]{\mathrm{Had}^{(#1)}_{#2}}

\newcommand{\corr}{\textcolor{algocitecolor}{\mathsf{NewLC}}}
\newcommand{\stdcorr}{\textcolor{algocitecolor}{\mathsf{StdLC}}}
\newcommand{\tabcorr}{\textcolor{algocitecolor}{\mathsf{TabLC}}}
\newcommand{\tabldt}{\textcolor{algocitecolor}{\mathsf{TabLDT}}}
\newcommand{\vldt}{\textcolor{algocitecolor}{\mathsf{NewLDT}}}
\newcommand{\stdldt}{\textcolor{algocitecolor}{\mathsf{StdLDT}}}
\newcommand{\vvanish}{\textcolor{algocitecolor}{\mathsf{ZERO}}}
\newcommand{\hadtest}{\textcolor{algocitecolor}{\mathsf{HadTest}}}

\newcommand{\LDE}{\mathrm{LDE}}

\usepackage[hidelinks,hypertexnames=false]{hyperref}
\usepackage[nameinlink]{cleveref}

\crefname{appendix}{Appendix}{Appendices}
\Crefname{appendix}{Appendix}{Appendices}
\usepackage{url}            
\usepackage{footnotebackref}

\newtheorem{theorem}{Theorem}[section]
\newtheorem{corollary}[theorem]{Corollary}

\newtheorem{lemma}[theorem]{Lemma}
\newtheorem{observation}[theorem]{Observation}

\newtheorem{definition}[theorem]{Definition}
\newtheorem{claim}[theorem]{Claim}
\newtheorem{fact}[theorem]{Fact}
\newtheorem{remark}[theorem]{Remark}

\newcommand{\F}{\mathbb{F}}

\newcommand{\N}{\mathbb{N}}

\newcommand{\Boo}{\{0,1 \}}

\newcommand{\bigO}{\mathcal{O}}

\newcommand{\cV}{\mathcal{V}}

\newcommand{\cP}{\mathcal{P}}

\newcommand{\veca}{\mathbf{a}}
\newcommand{\vecb}{\mathbf{b}}

\newcommand{\ba}{\mathbf{a}}
\newcommand{\bu}{\mathbf{u}}
\newcommand{\bv}{\mathbf{v}}
\newcommand{\bb}{\mathbf{b}}
\newcommand{\bx}{\mathbf{x}}

\newcommand{\bZ}{\mathbf{Z}}

\newcommand{\by}{\mathbf{y}}

\newcommand{\rej}{\mathrm{Rej}}

\DeclareMathOperator{\Hom}{hom}
\DeclareMathOperator{\poly}{poly}

\newcommand{\paren}[1]{\left( #1 \right)}

\newcommand{\brac}[1]{\left[ #1 \right]}

\newcommand{\set}[1]{\left\{ #1 \right\}}

\newcommand{\setcond}[2]{\left\{ #1 \;\middle\vert\; #2 \right\}}

\newcommand{\accept}{\textsc{Accept}}
\newcommand{\reject}{\textsc{Reject}}

\newcommand{\eqq}{\overset{?}{=}}
\newcommand{\etail}[1][i]{\Lambda_{\eta_{#1}, L}}
\DeclareMathOperator*{\E}{\mathbb{E}}

\definecolor{thmcolor}{RGB}{173, 255, 222}
\definecolor{citecolor}{RGB}{1, 210, 56}
\definecolor{lemmacolor}{RGB}{130, 169, 252}
\definecolor{emphcolor}{RGB}{98, 17, 245}

\newtcolorbox{algobox}{colback=lightgray!5!white,colframe=gray!75!gray}
\newtcolorbox{thmbox}{colback=thmcolor!5!white,colframe=teal!75!teal}
\newtcolorbox{lemmabox}{colback=lemmacolor!5!white,colframe=blue!75!blue}

\usepackage{enumerate}
\usepackage[parfill]{parskip}
\usepackage{changepage}
\usepackage{framed}

\allowdisplaybreaks
\hypersetup{
	colorlinks,
	linkcolor={blue},
	citecolor={citecolor},
	urlcolor={blue}
}

\usepackage[
	backend=biber,
	style=alphabetic,
	sorting=nyt,
	backref=true,
	maxcitenames = 8,
	mincitenames = 5,
	maxalphanames = 8,
	minalphanames = 5,
	maxnames = 10,
	minnames = 5
]{biblatex}

\date{\today}

\begin{document}

\title{A Simple Algebraic Proof of the PCP Theorem}


\author{Prashanth Amireddy\thanks{School of Engineering and Applied Sciences, Harvard University, Cambridge, Massachusetts, USA. Supported in part by a Simons Investigator Award and NSF Award CCF 2152413 to Madhu Sudan and a Simons Investigator Award to Salil Vadhan. Part of this work was done during a visit to the University of Copenhagen supported by the European Research Council (ERC) under grant agreement no. 101125652 (ALBA). Email: \texttt{pamireddy@g.harvard.edu}} \and
		Amik Raj Behera\thanks{Department of Computer Science, University of Copenhagen, Denmark. Supported by Srikanth Srinivasan's start-up grant from the University of Copenhagen. Email: \texttt{ambe@di.ku.dk} } \and
		Srikanth Srinivasan \thanks{Department of Computer Science, University of Copenhagen, Denmark. Supported by the European Research Council (ERC) under grant agreement no. 101125652 (ALBA). Email: \texttt{srsr@di.ku.dk} } \and
		Madhu Sudan\thanks{School of Engineering and Applied Sciences, Harvard University, Cambridge, Massachusetts, USA. Supported in part by a Simons Investigator Award, NSF Award CCF 2152413 and AFOSR award FA9550-25-1-0112. The U.S.~Government retains a nonexclusive, royalty-free license to publish or reproduce the published form of this contribution, or allow others to do so, for U.S.~Government purposes. Copyright is owned by SIAM to the extent not limited by these rights. Email: \texttt{madhu@cs.harvard.edu}} \and Sophus Valentin Willumsgaard \thanks{Department of Computer Science, University of Copenhagen, Denmark. Supported by the European Research Council (ERC) under grant agreement no. 101125652 (ALBA). Email: \texttt{sophus.willumsgaard@di.ku.dk} }   }

\maketitle

\begin{abstract}
	We give the simplest known algebraic proof of the PCP theorem, involving only ingredients like code concatenation, polynomial interpolation, and polynomial multiplication. Specifically, we prove that graph 3-coloring has a polynomial-sized proof that can be verified by a verifier tossing logarithmically many coins and querying a constant number of bits in the proof. In particular, our proof does not involve any PCP compositions; notably, it does not invoke the NP-completeness of any fixed problem, such as SAT or 3-coloring, in the construction of the verifier. The main innovation in our work is a clean, coding theoretic, way to encode univariate polynomials that allows us to implement ``low-degree testing'' using just a constant number of bits of queries. Insights from recent attempts to simplify the PCP proof by the authors (STOC 2026) and Goldreich (ECCC 2025) allow us to observe that low-degree testing was the key bottleneck in converting previous algebraic constructions of the PCP verifier into a constant query PCP. Thus, by overcoming this bottleneck, we get the full PCP verifier using elementary and self-contained steps. As concrete support for the claimed simplicity, we include the full pseudocode of the PCP verifier, assuming finite field arithmetic, and a full description of the completeness (aka ``honest'') prover, assuming multivariate polynomial arithmetic including interpolation and evaluation, that fit in about a page each.
\end{abstract}
\thispagestyle{empty}

\newpage


\tableofcontents

\thispagestyle{empty}

\newpage

\setcounter{page}{1}
\section{Introduction}

Ever since the discovery of the PCP theorem in the early 90s~\cite{AroraS,ALMSS}, a ``simple proof'' of the PCP theorem has been an object of desire in the CS community. Indeed, this quest has led to many improvements to the original algebraic proof, while also launching radically novel directions as in the proof of the PCP theorem by gap amplification by Dinur~\cite{Dinur} (see also Dinur and Reingold~\cite{DinurReingold}). The algebraic proofs and gap amplification proofs are quite far from each other, and the relative simplicity is hard to compare.
In this work, we focus on the algebraic direction and further simplify it to be able to achieve two concrete milestones: Specifically, we give pseudocodes (assuming finite field and polynomial arithmetic) for the PCP verifier, as well as the completeness (aka ``honest'') prover, each of which fits in about a page (see \Cref{app:full-verifier}). Our constructions have also been implemented in Python~\cite{Github-repo-PCP-verifier}, and our proof has been formalized in Lean by Lee~\cite{lee2026_pcp}. One notable aspect of our proof is that there is no composition of proofs in our verifier, and in particular we do not rely on NP-completeness proofs of some canonical problem to recursively verify proofs. To elaborate further, we give some background on past algebraic proofs of the PCP theorem and the ingredients therein.

Before delving into the details, let us recall that PCPs are typically parameterized by two main parameters --- the {\em randomness} complexity of the verifier, ideally logarithmic in the length of the classical proof; and the {\em query}  complexity, ideally constant. Achieving both ideal limits yields the PCP theorem, which is the goal of this paper; however, the constructions involve ingredients that are weak in one or both parameters. Of course our objective is to achieve these ideal parameters with a ``simple'' proof. As indicated above, we use this term informally to capture two ingredients, namely (1) the PCP verifier who, given an instance of say graph 3-coloring, determines the length of the proof and the queries to be performed and the acceptance condition and (2) the ``honest'' prover (formally the completeness prover), who given a 3-coloring of a 3-colorable graph, computes the bits of the PCP proof to supply to the verifier (which the verifier should accept with probability one). Indeed, one could argue that the complexity of this honest prover is perhaps the most natural measure of complexity of a PCP proof in that it explains what the PCP proof looks like; and this has been a complex transformation in all previous proofs and can be viewed as the primary motivation of this work.

\paragraph{The Ingredients in an Algebraic PCP:}Algebraic PCPs, starting with the works of Arora and Safra~\cite{AroraS} have relied on three kinds of ingredients: (1) ``Atomic PCPs'', which are self-contained, simple constructions of PCPs for some natural problem (like SAT or Graph-Coloring), which are deficient in at least one of the two parameters. (2) ``Composition'' --- a technique that combines two PCP-like objects that we will call outer and inner verifiers to get a new object that has randomness complexity close to that of the outer verifier and query complexity closer to that of the inner verifier. Applying this operator judiciously can get both parameters closer to the ideal limit. (3) ``Composition-preparation'' --- we coin this term and use it to describe a variety of steps involved to convert the atomic PCPs into outer and inner verifiers suitable for composition. In some cases, like in the original proof of the PCP theorem in the work of Arora, Lund, Motwani, Sudan and Szegedy~\cite{ALMSS}, this can be a technically complex operation. But in all cases, the preparations are needed to make the atomic verifier into a ``robust'' one suitable for use as an outer verifier (involving operations termed ``parallelization'' (see \cite{ALMSS}) or ``bundling'' (see \cite{BGHSV}) or ``robustification'') as well as into a ``PCP of proximity'' or ``assignment tester'' (see \cite{BGHSV, DinurReingold}) which makes them suitable for use as an inner verifier. (Notably, the use of the NP completeness of the ``natural problem'' for which the PCP was originally designed to get a verifier for less natural and harder-to-specify problems.) Together, ``composition'' with ``composition-preparation'' contribute to the bulk of the conceptual complexity of algebraic PCPs, making them hard to describe and even harder to teach. In this work, our goal is to build a PCP whose complexity is not much more than that of the atomic PCP, foregoing composition-preparation altogether and simplifying ``composition'' to the level of ``code-concatenation'' which we expand on next.

\paragraph{Composition of PCPs vs. Concatenation of Codes (see also Goldreich~\cite[Section 4]{Goldreich-one-half-PCP-composition}):} Concatenation of codes is an elementary operation that converts codes over large alphabets into codes over smaller ones. Concatenation uses two ingredients - an outer code over a large alphabet, where the large alphabet makes it easy to achieve a large distance between codewords; and an inner code that shows how to encode elements of the large alphabet as, say, binary strings. Concatenation simply converts the codewords of the outer code into binary strings by writing each element of the alphabet in binary using the inner code. Notationally, concatenation is as simple as composition of functions (viewed appropriately).

Composition of PCPs is an extension of this concept to PCPs but becomes significantly more complex. Viewed in terms of what the honest prover would do, composition takes an outer prover who builds a PCP proof over a large alphabet (in which the verifier queries a constant number of symbols) and then converts this proof into a PCP proof on bits. But this translation involves a number of complexities beyond concatenation. (1) In the simplest form, the bits of the composed proof are encodings of all $c$-tuples of symbols of the outer code for some constant $c$. (2) In a slightly more complex form, a special family of $c$-tuples of the symbols of the outer codes is chosen to be re-encoded. (3) In the most complex form, the encoding of each $c$-tuple depends on the acceptance condition that the outer verifier intends to apply, and thus varies from $c$-tuple to $c$-tuple. (Typically this step involves use of the Cook-Levin theorem to express arbitrary outer verifiers in some canonical form using an NP-complete problem.)\footnote{In our estimate, all previous algebraic proofs including~\cite{ALMSS,BGHSV,BenSassonS,ABSSW-PCP-one-composition} use (3) though all but \cite{ALMSS} don't rely on (2). The PCP in \cite{Dinur} relies crucially on (2) but does not need (3).}

\paragraph{Our result and ideas:}The  PCP we present in this paper eliminates the need for compositions in the sense above.\footnote{One could argue that the use of the low-degree test, described below, which necessitates a suitable `proof oracle', is a non-trivial mechanism that might count for `0.5' compositions in the sense of Goldreich~\cite{Goldreich-one-half-PCP-composition}.} As a consequence the honest prover as well as the verifier become significantly simpler and can be described in a relatively straightforward manner using standard polynomial manipulations. Broadly, our PCP is obtained by taking the simplest known atomic outer PCP, namely the one from \cite{ABSSW-PCP-one-composition}; identifying the key bottleneck to a simple composition there, namely the ``low-degree test'', and then giving a new encoding and implementation of the low-degree test that allows for a simple composition. We expand on these steps by first giving a quick overview of the atomic test.

The atomic test we use goes back to the work of Babai, Fortnow, and Lund~\cite{BFL}, which was first converted to a ``large alphabet'' constant query test in \cite{ALMSS}. These works effectively show that verifying graph 3-coloring on $n$ vertex graphs ``reduces'' to verifying that a constant number of $m$-variate polynomials over a field $\F_q$ are of degree at most $d$ for $m,d,q = \poly(\log n)$ and verifying that these polynomials are zero on some subset of the domain of the form $H^m$ for $H \subseteq \F_q$.~\footnote{A similar reduction also works for SAT but the reduction is cleaner and more straightforward for $3$-coloring.} The latter task used to involve the famed ``sum-check protocol'' but this aspect (while elegant) is cleaned up significantly in the work of Ben-Sasson and Sudan \cite{BenSassonS}  who show that the final task reduces to $m$ low-degree tests on $m$ additional polynomials, and further cleaned up by
previous work of the authors \cite{ABSSW-PCP-one-composition}
who reduce this to one low-degree test on just one additional $2m$-variate polynomial. The complexity hidden by the ``reductions'' alluded to above involves querying a constant number of polynomials at a constant number of locations and verifying that the query responses are zeroes of a constant number of constant-degree polynomials.

Thus, apart from the ``low-degree test'' (still to be discussed), all aspects of the proof thus far require querying the proof in a constant number of places, where every element of the proof is from $\F_q$ and thus requires $\Theta(\log\log n)$ bits to describe. The query complexity here can be reduced to $\bigO(1)$ bits by simple concatenation with the ``low-degree long code'' (see \cite{BGHMRS})\footnote{A low-degree long code encodes a field element $\alpha \in \F_{q}$, where $\F_{q} \cong \F_{2}^{t}$, by evaluations $R(\alpha)$ of all $\bigO(1)$-degree $t$-variate polynomials $R$.} which involves writing down the evaluation of every $\bigO(1)$-degree polynomial of these $\bigO(\log\log n)$ bits. This encoding is known to be locally testable and sufficient to perform the tests used in the reductions mentioned above; and their encoding takes exponential length, but on strings of length $\log \log n$ we can afford this blow-up! This leads us to the key bottleneck, namely the low-degree test.

\paragraph{The low-degree test:} The low-degree testing problem is a simple one to specify: We are given oracle access to a function $f:\F_q^m \to \F_q$, a degree parameter $d$ and we would like a test that queries $f$ in $\bigO(1)$ locations and accepts degree $d$ polynomials while rejecting functions that are, say $0.01$-far (in relative Hamming distance) from degree $d$ polynomials with probability say $0.0001$. Unfortunately, as stated, this problem is not solvable, and so one allows an auxiliary ``proof'' oracle $f_1$. Now we would like that the test is allowed $\bigO(1)$-queries into $f$ and $f_1$, such that for every degree $d$ polynomial $f$ there exists $f_1$ such that the verifier always accepts $(f,f_1)$ while if $f$ is $0.01$-far, then for every $f_1$ the pair $(f,f_1)$ is rejected with probability at least $0.0001$. Such tests are well-known --- we use the line-point test of Rubinfeld and Sudan~\cite{RubSud} that suffices for our purpose. Here the function $f_1:\F_q^{2m} \to \F_q^{d+1}$ where the domain represents a pair of points in $\F_q^m$ and the $f_1(a,b)$ purportedly represents the degree $d$ univariate polynomial that describes $f$ restricted to a line containing $a$ with slope $b$. The test picks $a,b\in \F_q^m$ and a non-zero $\lambda \in \F_q$ randomly and verifies that $f_1[a,b](\lambda) \eqq f(a+\lambda b)$ which is an $\F_q$-linear test involving just one query each to $f$ and $f_1$.

Unfortunately, we cannot reduce the alphabet size of the prover in this test by using a ``low-degree long code''. In particular evaluations of $f_1$ are $d\log q = \poly(\log n)$ bit strings (specifically the strings are $\bigO(\log n)^{c_1}$ bits long for some $c_1 > 1$.) Low-degree long-code encodings map $k$-bit strings to $2^{k^{c_2}}$-bit strings for some constant $c_2>1$. Composing the two gives a quasi-polynomial-sized PCP  (equivalently, the resulting randomness complexity would be $\bigO(\log n)^{c_1 c_2}$) which does not meet our ideal goal. This leads us to our key problem and solution described next.

\paragraph{A new multivariate encoding:} Our goal is to come up with some encoding $E:\F_q^d \to \F_2^{N}$ for some $N$, such that this encoding is easy to test with $\bigO(1)$ queries and so that given $\lambda \in \F_q$, the encoding $E[f_1(a,b)]$ allows us to ``locally retrieve'' $f_1(a,b)(\lambda)$ (specifically we need a ``self-corrector'', one that, even if the encoding is slightly corrupted, for every $\lambda$ correctly reports $f_1(a,b)(\lambda)$ for the nearest codeword). For our application achieving $N = 2^{d^{\epsilon}}$ for arbitrarily small $\varepsilon > 0$ works and we can allow a locality of $\bigO_\varepsilon(1)$. Such a task has been considered in the literature in the past, specifically in the work of Dinur, Fischer, Kindler, Raz, and Safra~\cite{DFKRS11}, though possibly with additional requirements on the nature of the information that is locally retrievable from the information and with different requirements on the locality and $N$. 

We start by getting a non-binary encoding, specifically an $E_0:\F_q^d \to \F_q^{N'}$ with $N'=2^{O(d^\varepsilon)}$. Concatenation with a Hadamard code converts this to a binary code in a standard way, and we omit the details here. 
To get this encoding $E_0$, we represent a univariate degree $< d$ polynomial $P(X)$ by a multivariate polynomial of lower degree. We call this the \emph{Set-Multilinear encoding} of $P$ as it is essentially a ``set-multilinearization'' of the well-known Inverse Kronecker map that encodes a high-degree univariate polynomial as a low-degree multivariate polynomial. 

The class of multivariate polynomials we use in the encoding is set-multilinear\footnote{A multivariate polynomial in $c$ disjoint sets of variables $Y_1,\ldots,Y_c$ is  \emph{set-multilinear} if every monomial contains exactly one variable from each $Y_i.$} polynomials in $c$ blocks of variables (we will choose $c$ to be an absolute constant), with each block containing $d^{1/c}$ variables. Denoting these blocks of variables by $\{Y_{i,j}\}_{i \in [c], j \in [\ell]}$ for $\ell = \lceil d^{1/c}\rceil$, we think of the variable $Y_{ij}$ as representing the monomial $X^{(j-1)\ell^{i-1}}$.  Writing every exponent in base $\ell$ we get a way to represent every monomial $X^e$ as a set multilinear form in the $Y_{ij}$'s. Thus every polynomial $P(X)$ can be represented as a set-multilinear polynomial, and hence degree $\leq c$, polynomial $\widetilde{P}$ in the variables $Y_{ij}$. Our encoding of $P$ is simply the evaluation of $\widetilde{P}$ over $\F_q^{c\ell}$. We test this encoding loosely by simply testing that it is a degree $\leq c$ polynomial. Since $c$ is a constant, this only requires constantly many queries into the encoding. Every honest encoding of a univariate polynomial passes the test. But if a function passes the test, we only get that it is close to the evaluation of some $\widetilde{P}$ that has total degree $\leq c$. But this is sufficient to argue that $\widetilde{P}((X^{(j-1)\ell^{{i-1}}})_{i,j})$ is a (univariate) polynomial of degree $\bigO(cd)$ and furthermore this encoding can be locally self-corrected. The degree guarantee is weaker than strictly needed in a low-degree test\footnote{This could be rectified with a few more tests, but we do not do so in the interest of keeping the verifier simple} but this is completely fine for the soundness proof.
We note that the encoding proposed in \cite{DFKRS11} also seems to start with the inverse Kronecker map, but does not do the ``set multilinearization''. As a result their encoding is more efficient, with $N=\poly_\varepsilon(d)$, but achieves locality of $\bigO(d^\epsilon)$ for local information retrieval.

Combining this encoding of univariate low-degree polynomials with the other ingredients listed above (the atomic PCP from \cite{ABSSW-PCP-one-composition} and the low-degree long code) now suffices to give us our PCP.

Our new multivariate encoding (along with the replacement for the sum-check protocol based on~\cite{BenSassonS,ABSSW-PCP-one-composition}) allows us to extend this technique to get the randomness complexity down to $\bigO(\log n),$ at the expense of one more code-concatenation.

\paragraph{Comparison with the work of Goldreich~\cite{Goldreich-one-half-PCP-composition}:} This result is motivated by a result of Goldreich~\cite{Goldreich-one-half-PCP-composition}, who shows how to obtain a PCP theorem with randomness complexity $n^{\varepsilon}$ (for arbitrarily small $\varepsilon$) and $\bigO(1)$ queries using only code-concatenation. (\cite{Goldreich-one-half-PCP-composition} describes the use of code concatenation as a ``half-composition''. Using this PCP (or a variant thereof) as an inner PCP that can be composed with the outer PCP from \cite{ALMSS} yields the full PCP theorem with one (and a half) composition steps.)

Goldreich's main result, namely the $\bigO(1)$-query PCP with randomness complexity $n^{\varepsilon}$, is limited by the use of a version of the sum-check protocol, which has constant query complexity only for a restricted range of parameters. One can view our construction as modifying the construction from \cite{Goldreich-one-half-PCP-composition} with a replacement for the sum-check protocol based on~\cite{BenSassonS,ABSSW-PCP-one-composition}. This allows us to reduce the query complexity to a constant while keeping randomness $\bigO(\log n)$, while paying a price in the alphabet size of the PCP which is now $\bigO(\poly\log n)$ bits. This answer size could be reduced to $\bigO(1)$ by encoding answers using a Hadamard PCP and applying a recursive verification strategy as in \cite{Goldreich-one-half-PCP-composition}.
This would bring down the query complexity to constant, but would increase the randomness to $\bigO(\poly\log n)$ and thus fall short of our goal of reducing the randomness down to $\bigO(\log n).$ 
To bring down the randomness to $\bigO(\log n)$, we need our set-multilinear encoding of univariate polynomials, which allows us to reduce the answer size more randomness-efficiently and thus get the full PCP theorem at the expense of two levels of code-concatenation.

\subsection*{Organization of the paper}
In \Cref{sec:preliminaries}, we give the necessary definitions, background results, and state the PCP Theorem (\Cref{thm:pcp}). In \Cref{sec:encodings}, we describe the two encodings that we will use: set-multilinear encoding (\Cref{subsec:first-encoding}) and low-degree Hadamard encoding (\Cref{subsec:second-encoding}). We also discuss their syntactic test (i.e., a test to decide whether the oracle is close to the expected form) and self-correctors in \Cref{subsec:tests-set-ml} and \Cref{subsec:tests-hadamard}, respectively. We provide the proofs from \Cref{sec:encodings} in \Cref{app:encodings-app} and \Cref{app:syntactic-test-gen-Hadamard}. Using these two encodings, we give our low-degree test in \Cref{sec:ldt} and analyze it. In \Cref{sec:lc}, we give our local corrector over alphabet $\F_{2}$ and analyze it in \Cref{app:proof-local-correction}. In \Cref{sec:zero-test}, we give our zero-on-grid test over alphabet $\F_{2}$ and analyze it in \Cref{app:proof-zero-test}. Using these three algorithms (low-degree test, local corrector, and zero-on-grid test), we describe our PCP verifier in \Cref{sec:pcp}. In \Cref{app:full-verifier}, we give a complete standalone description of the PCP verifier and an honest prover. 

\section{Preliminaries}\label{sec:preliminaries}

\paragraph{Probabilistic Algorithms and Oracles}For an algorithm $\mathcal{V}$ and a string $\Pi$ over alphabet $\Sigma$, we will use $\mathcal{V}^{\Pi}$ to denote the algorithm $\mathcal{V}$ with oracle access to string $\Pi$. By oracle access, we mean $\mathcal{V}$ can query $\Pi[i]$ for any $1 \leq i \leq |\Pi|$. For a probabilistic algorithm $\mathcal{V}$, we will use the notation $\mathcal{V}^{\Pi}(x;R)$ to say that the algorithm $\mathcal{V}$ has oracle access to $\Pi$, has input $x$, and access to a random string $R$. In this notation, $\mathcal{V}^{\Pi}(x;R)$ is a deterministic algorithm, and the randomness is in the choice of $R$. For a string $\Pi$ and a subset $S \subseteq [|\Pi|]$, $\Pi|_{S}$ refers to the substring formed by entries of $\Pi$ corresponding to $S$.

\subsection{PCPs}
Throughout this article, whenever we mention a verifier, we mean a \textit{standard verifier}, which we define next. We will always use a standard verifier, both in PCP and in the subroutines such as low-degree test.\\




\begin{definition}[Standard Verifier]\label{defn:standard-verifier}
	For functions $r,\ell: \mathbb{Z}^{\geq 0} \to \mathbb{Z}^{\geq 0}$, define an $(r,\ell)$-standard verifier $\mathcal{V}$ as follows:
	On input $\mathbf{x} \in \Boo^{n}$ of length $n$, a string\footnote{This string $R$ is the random string fed into $\mathcal{V}$.} $R \in \Boo^{r(n)}$, and oracle access to a string\footnote{called a {\em proof}} $\Pi \in \{0,1\}^{\mathrm{size}(n)}$,
	\begin{itemize}
		\item $\mathcal{V}$ computes a subset $Q = Q(\mathbf{x},R) \subseteq [\mathrm{size}(n)]$ of cardinality $\ell(n)$. These are the \emph{queries} of $\mathcal{V}$ on input $\mathbf{x}$ and randomness $R$.
		\item $\mathcal{V}$ then either returns $\accept$ or $\reject$ depending on $\Pi|_{Q}$.
	\end{itemize}

	Furthermore, the running time of $\mathcal{V}$ has to be $\poly(n\cdot 2^{r(n)})$.
\end{definition}

\noindent
Observe that in the above definition, $\mathcal{V}$ makes $\ell(n)$ queries to $\Pi$ using the $r(n)$ coin tosses. In particular, $\mathcal{V}$ can only query a coordinate within range of $[0, \ell(n) \cdot 2^{r(n)}-1]$. So without loss of generality, we can  always assume that the proof size $|\Pi|$ is $\mathrm{size}(n) = \bigO(\ell(n) \cdot 2^{r(n)})$.\\


In this paper, we will focus on the language $3$-$\mathsf{COLOR}$ of $3$-colorable graphs, which is known to be $\mathsf{NP}$-complete. This is only because $3$-$\mathsf{COLOR}$ fits quite nicely with our PCP proof. We could always construct a PCP verifier similar to ours for other $\mathsf{NP}$ problems. In particular, we show that certifying membership in $3$-$\mathsf{COLOR}$ admits $(\bigO(\log n),\bigO(1))$-standard verifiers:

\begin{thmbox}
	\begin{restatable}[The PCP Theorem {\cite{ALMSS}}]{theorem}{pcpthm}\label{thm:pcp}
		There exists a randomized oracle algorithm $\cV$, a constant $\gamma>0$, and a polynomially growing function $\mathrm{size}(n)$ such that for every $n$ and for every input graph $G=(V,E)$ over $n$ vertices, we have the following.
		\begin{itemize}
			\item If $G \in 3$-$\mathsf{COLOR}$, then there exists $\Pi\in \{0,1\}^{\mathrm{size}(n)}$ such that
			      \begin{align*}
				      \Pr[\cV^{\Pi}(G) \text{ returns \accept}] = 1.
			      \end{align*}

			\item If $G \notin 3$-$\mathsf{COLOR}$, then for every $\widetilde{\Pi} \in \{0,1\}^{\mathrm{size}(n)}$, we have
			      \begin{align*}
				      \Pr[\cV^{\widetilde{\Pi}}(G) \text{ returns \reject}] \ge \gamma.
			      \end{align*}
		\end{itemize}
		Furthermore, for every $\Pi\in \{0,1\}^{\mathrm{size}(n)}$, $\cV^\Pi(G)$ runs in time $\poly(n)$, uses $\bigO(\log n)$ bits of randomness, and makes $\bigO(1)$ oracle queries.
	\end{restatable}
\end{thmbox}

For a prime power $q$, we denote the finite field of size $q$ by $\F_q$. Throughout this article, we will have $q = 2^{t}$ for some $t \in \mathbb{N}$. We will use $\mathcal{P}_{d}(m,\F)$ to denote the space of functions that can be represented by $m$-variate polynomials of degree at most $d$ (or simply {\em degree-$d$}) over the field $\F$. For a finite string $\Pi$ and an algorithm $\mathcal{A}$, we will use $\mathcal{A}^{\Pi}$ to denote that the algorithm $\mathcal{A}$ has oracle access to $\Pi$. In particular, $\mathcal{A}$ can query the string $\Pi$ at any coordinate/entry. For us, $\Pi$ will be evaluation of functions on some domain $D$, and so $\mathcal{A}^{\Pi}$ means that $\mathcal{A}$ can access the value of $\Pi[\mathbf{x}]$ for any point $\mathbf{x} \in D$. When $D=D_1\times D_2$, for $\mathbf{x} \in D_1$, we will use $\Pi[\mathbf{x},\cdot]$ to denote the restriction of $\Pi$ corresponding to fixing the first coordinate to $\mathbf{x}$. We define $\Pi[\mathbf{x},\cdot,\cdot]$ similarly when $D=D_1\times D_2\times D_3$. For a probabilistic algorithm $\mathcal{A}$, we will use $\mathcal{A}(\mathbf{x}; \mathbf{r})$ to denote the output of the algorithm on input $\mathbf{x}$ and random string $\mathbf{r}$.

\noindent
For two functions $f,g: S \to T$, we use $\delta(f,g)$ to denote their relative Hamming distance, i.e., the fraction of points in the domain $S$ where they differ. We say two functions $f,g:S\rightarrow T$ are $\delta$-close if $\delta(f,g) \leq \delta$; else, we say they are $\delta$-far. For a family $\mathcal{F}$ of functions from $S$ to $T$ and a function $f: S \to T$, we use $\delta(f,\mathcal{F})$ to denote the minimum relative distance of $f$ from a function in $\mathcal{F}$, i.e., $\delta(f,\mathcal{F}) = \min_{g \in \mathcal{F}} \delta(f,g)$.

\noindent
For a subset $S \subseteq \F_{q}^{m}$, we will use the notation $\mathbf{x} \sim S$ to denote that $\mathbf{x}$ is uniformly distributed over the subset $S$.
Let $q=2^{t}$ for an even integer $t$, and $\mathbf{Z}=(Z_1,\dots,Z_t)$ denote $t$ many variables. Let $\F_2[\mathbf{Z}]^{\le i}$ denote the set of multilinear polynomials in $\F_2[\mathbf{Z}]$ of degree at most $i$. Also, we will use $\hpol$ to denote the family of homogeneous degree $1$ polynomials. For a set of polynomials $\set{G_{1},\ldots,G_{k}}$, we will denote the ideal generated by them as $\mathbb{I}(G_{1},\ldots,G_{k})$. Let $\omega \in \F_q\setminus \{1\}$ be a cube root of the multiplicative identity. We also state a few standard facts about finite fields and low-degree extensions.\\

\begin{fact}[Construction of finite fields,~{see e.g.~\cite[Chapter 20]{Shoup}}]\label{fact:construction-field-generator} Let $q = 2^{t}$ for some $t \in \mathbb{N}$. A description of the finite field $\F_{q}$ along with a generator can be computed\footnote{i.e., given $t$ as input, output the ``truth tables'' of the field addition and multiplication. Consequently, we can also find an element of any given order (if it exists).} in time $\mathrm{poly}(q)$.
\end{fact}
\noindent
\begin{fact}[Low-Degree Extension]\label{fact:low-deg-extension}
	Let $H \subseteq \F_{q}$ be an arbitrary subset and $m \in \mathbb{N}$ denote the number of variables. For every function $f: H^{m} \to \F_{q}$, we denote by $\LDE(f) \in \F_{q}[X_{1},\ldots,X_{m}]$, the \emph{low-degree extension} of $f$
	which is the unique polynomial satisfying:
	\begin{itemize}
		\item Individual degree\footnote{Individual degree is defined to be the maximum degree of any variable appearing in the polynomial. Hence, the total degree of $\LDE(f)$ is at most $(|H|-1)\cdot m$.} of~$\LDE(f)$ is at most $|H|-1$.
		\item For every $\mathbf{a} \in H^{m}$, $\LDE(f)(\mathbf{a}) = f(\mathbf{a})$.
	\end{itemize}
	Furthermore, $\LDE(f)$ can be constructed in time $\mathrm{poly}(|H|^{m}, \log q)$ via solving a system of linear equations (coefficients of $\LDE(f)$ are the variables and evaluations of $f$ on $H^{m}$ are the constraints).
\end{fact}

We start by recalling polynomial distance lemmas, which are the crucial reason why low-degree polynomials are so useful in constructing PCPs, and they will be the backbone of our analysis.\\

\begin{theorem}[Polynomial Distance Lemmas,~{\cite{ore1922hohere, DL78, Schwartz80, Zippel79}}]\label{thm:odlsz}
	Fix a field $\F_{q}$, degree parameter $d \in \mathbb{N}$, and a non-zero polynomial $P \in \F_{q}[X_{1},\ldots,X_{m}]$ of degree at most $d$.
	\begin{itemize}
		\item If $d\le q$, we have \begin{align*}
			      \Pr_{\mathbf{a} \sim \F_{q}^{m}}[P(\mathbf{a}) \neq 0] \; \geq \; 1 - \dfrac{d}{q}.
		      \end{align*}
		\item  If $q=2$ and $P$ is multilinear, we have
		      \begin{align*}
			      \Pr_{\mathbf{a} \sim \F_{2}^{m}}[P(\mathbf{a}) \neq 0] \; \geq \; \dfrac{1}{2^{d}}.
		      \end{align*}
	\end{itemize}
\end{theorem}

\noindent
An immediate and useful corollary of \Cref{thm:odlsz} is the following: If two degree \(\leq d\) polynomials $P$ and $Q$ agree on strictly more than $d/q$-fraction of $\F_{q}^{m}$, then $P = Q$.

\paragraph*{Lines Table}For every $m \in \mathbb{N}$, field $\F_{q}$, and points ${\bf a},{\bf b}\in\F_q^m$,
let $\ell_{{\bf a},{\bf b}}:\F_q \to \F_q^m$ (read as ``line passing through ${\bf a}$ with slope ${\bf b}$'') be defined as $\ell_{{\bf a},{\bf b}}(\lambda ):={\bf a}+\lambda {\bf b}$.\\

\begin{definition}[Lines Table]\label{defn:lines-table}
	Fix a field $\F_{q}$.
	Let $d \in \N$ be the degree parameter and $m \in \N$ be the number of variables.
	For every degree \(\leq d\) polynomial $f: \F_{q}^{m} \to \F_{q}$,
	we define the lines table of $f$ as the function
	$f_{\mathrm{lines}} \; : \; \F_q^{2m} \; \longrightarrow \; \F_{q}^{d+1}$ that maps an input
	$({\bf a},{\bf b})\in \F_q^{2m}$
	to
	$f(\ell_{{\bf a},{\bf b}}(X))$, where $X$ is an indeterminate.
	We note that
	$f(\ell_{{\bf a},{\bf b}}(X))$
	is indeed a univariate polynomial of degree at most $d$ in the indeterminate $X$ and can be specified by the $d+1$ coefficients of $X^0,X^1,\ldots, X^d$.
\end{definition}

Throughout the paper, whenever we have a function $F$ that has range in $\F_{q}^{d+1}$, we will interpret the output as a univariate polynomial of degree at most $d$.

\paragraph*{Local characterization of low-degree polynomials}Another reason why low-degree polynomials are useful in PCPs is the fact that the family of low-degree polynomials is the same as the family of functions that are low-degree on every line. We state it formally in the following fact. The fact can be proved using a closely related characterization in \cite[Lemma 4]{Friedl-Sudan} and univariate polynomial interpolation. We give a proof for the sake of completeness in \Cref{app:encodings-app}.\\

\begin{fact}[Characterization of Low-Degree Polynomials]\label{fact:characterize-low-deg}
	Fix a degree parameter $c \in \mathbb{N}$ and fix a field $\F_{q}$ with $q > 2c$. Let $\zeta \in \F_{q}^{\times}$ be a generator of $\F_{q}^{\times}$ and consider the set $\{\zeta_0,\ldots,\zeta_{c+1}\}$ where $\zeta_0=0$ and $\zeta_{i} = \zeta^{i}$ for $i \in [c+1]$. Let $\eta_0=-1$ and for $1\leq i\leq c+1$ let:
	\begin{align*}
		\eta_{i} \; = \; \brac{ \prod_{ \substack{ 1 \leq j \leq c+1 \\ j \neq i} } \dfrac{-\zeta_j}{\zeta_{i} - \zeta_{j}} }.
	\end{align*}
	Then, a function $P: \F_{q}^{m} \to \F_{q}$ is a degree-$c$ polynomial if and only if
	\begin{align*}
		\sum_{i=0}^{c+1} \; \eta_{i} \cdot P(\mathbf{u} + \zeta_{i} \cdot \mathbf{v}) \; = \; 0, &  & \text{for every } \mathbf{u},\mathbf{v} \in \F_{q}^{m},
	\end{align*}
	where for scalar $\alpha \in \F_{q}$, the point $\alpha \cdot \mathbf{y}$ is defined as $(\alpha y_{1},\ldots, \alpha y_{m})$.
\end{fact}

\noindent
A simple corollary of the above characterization is that we can evaluate a low-degree polynomial at any given point using its evaluations on a line passing through that point. We record it formally below.\\

\begin{fact}\label{fact:local-correction-via-lines}
	Fix a degree-parameter $c \in \mathbb{N}$ and fix a field $\F_{q}$ of characteristic $2$ with $q > 2c$. Let $\zeta$ be a generator of $\F_{q}^{\times}$ and $\set{\zeta_{0},\ldots,\zeta_{c+1}}$ and $\set{\eta_{0},\ldots,\eta_{c+1}}$ be the same as in \Cref{fact:characterize-low-deg}. For every $m$-variate degree-$c$ polynomial $P \in \F_{q}[X_{1}, \ldots, X_{m}]^{\leq c}$, we have,
	\begin{align*}
		P(\mathbf{u}) \, = \, \sum_{i=1}^{c+1} \eta_{i} \cdot P(\mathbf{u} + \zeta_{i} \cdot \mathbf{v}) \quad \quad \text{for every } \mathbf{u},\mathbf{v} \in \F_{q}^{m}.
	\end{align*}
\end{fact}

\paragraph*{Vanishing of Low-Degree Polynomials on a Subset} For a subset $H\subseteq \F_q$ and for every $i \in [m]$,
we will define the polynomial $Z^{(i)}_{H}(\mathbf{X})$ as the vanishing polynomial on $H$ for the $i^{th}$ variable $X_{i}$, i.e.,
\begin{align*}
	Z^{(i)}_{H}(\mathbf{X}) \; := \; \prod_{\lambda \in H} (X_{i}-\lambda).
\end{align*}
We will use $Z_{H}(X_{1},\ldots,X_{m})$ to denote the tuple of polynomials $(Z^{(1)}_{H}, \ldots, Z^{(m)}_{H})$. The next lemma gives an equivalent condition for when a polynomial vanishes on the subset $H^{m} \subseteq \F_{q}^{m}$.\\
\begin{lemma}[Combinatorial Nullstellensatz]\label{lemma:comb-null}
	Fix any polynomial $P \in \F_{q}[\mathbf{X}]$. The polynomial $P$ vanishes on $H^{m}$, i.e. $P|_{H^{m}} \equiv 0$, if and only if there exist polynomials $Q_{1}(\mathbf{X}), \ldots, Q_{m}(\mathbf{X})$ with $\deg(Q_{1}),\ldots,\deg(Q_{m}) \leq \deg(P)-|H|$ such that
	\begin{align*}
		P(\mathbf{X}) \; = \; \sum_{i=1}^{m} \; Q_{i}(\mathbf{X}) \cdot Z^{(i)}_{H}(\mathbf{X}).
	\end{align*}
\end{lemma}

\noindent
The above lemma can be interpreted as giving a certificate for when a low-degree polynomial $P$ vanishes. In particular, the certificate consists of low-degree polynomials $Q_{1}, \ldots, Q_{m}$. Similar to the approach in \cite{ABSSW-PCP-one-composition}, we will define a ``vanishing certificate polynomial'' consisting of the $Q_{i}$'s and use that as a certificate of $P|_{H^{m}} \equiv 0$. Essentially, we will consider the polynomial we get after replacing $Z_{H}^{(i)}(\mathbf{x})$ by a new variable $y_{i}$.

\begin{lemmabox}
	\begin{corollary}[Vanishing Certificate of a Polynomial] \label{coro:meta-vanishing-poly}
		Fix any polynomial $P \in \F_{q}[\mathbf{X}]$. The polynomial $P$ vanishes on $H^{m}$, if and only if there exists a \textit{vanishing certificate polynomial} $\mathcal{M}_{P} \in \F_{q}[\mathbf{X}, Y_{1},\ldots, Y_{m}]$ satisfying the following properties:
		\begin{enumerate}
			\item Degree of $\mathcal{M}_{P}$ is at most $\deg(P)$.
			\item The polynomial $\mathcal{M}_{P}(\mathbf{X}, \mathbf{0})$ is identically the zero polynomial. In other words, $\mathcal{M}_{P} \in \mathbb{I}(\mathbf{Y})$.
			\item The following identity holds:
			      \begin{align*}
				      P(\mathbf{X}) \; = \; \mathcal{M}_{P}\paren{\mathbf{X}, Z_{H}^{(1)}(\mathbf{X}), \ldots, Z_{H}^{(m)}(\mathbf{X})}.
			      \end{align*}
		\end{enumerate}
	\end{corollary}
\end{lemmabox}

\paragraph*{3-colorability.} We state the $3$-colorability language below, which is an $\mathsf{NP}$-complete problem. Note that the choice of $3$-coloring as an $\mathsf{NP}$-complete problem instead of one of many others is a choice for simplicity.\\

\begin{definition}\label{defn:3-color}
	The decision problem $3$-$\mathsf{COLOR}$ is the following problem:\newline
	Given a graph $G = (V,E)$ with $n$ vertices, decide whether there exists a proper coloring of $G$ using $3$ colors, i.e. for every edge $(u,v) \in E$, the vertices $u$ and $v$ are assigned different colors.
\end{definition}

\subsection{Low-Degree Test and Local Correction over Large Alphabet}
In this subsection, we discuss the standard point-vs-line test for low-degree testing from \cite{ALMSS}. The analysis of the test is a \emph{robust} version of \Cref{fact:characterize-low-deg}, i.e., if a function is low-degree on most of the lines in $\F_{q}^{m}$, then it is close to a global low-degree polynomial.

We start by recalling the test and state its properties in \Cref{thm:low-degree-testing}. Recall that for a function with output in $\F_{q}^{d+1}$, we interpret the output as a degree-$d$ univariate polynomial.

\begin{algobox}
	\begin{algorithm}[H]
		\caption{Low-Degree Test (with lines table): $\mathsf{TabLDT}$}
		\label{algo:low-degree-test}

		\DontPrintSemicolon
		\KwIn{Degree parameter $d$ and oracle access to $(f,f')$ where $f: \F_{q}^{m} \to \F_{q}$ and $f': \F_{q}^{2m} \to \F_{q}^{d+1}$ }

		\vspace{1mm}

		Sample $\mathbf{a}, \mathbf{b} \sim \F_{q}^{m}$, $\lambda \sim \F_{q}^{\times}$ \;

		\vspace{1mm}

		\lIf{$f'[({\bf a},{\bf b})](\lambda) \neq f[{{\bf a}+ \lambda{\bf b}}]$ }{
		\Return{\reject}
		}
		\Return{\accept}

	\end{algorithm}
\end{algobox}

\noindent
\begin{theorem}[Low-Degree Test via Lines Table,~{\cite[Theorem 45]{ALMSS}}]\label{thm:low-degree-testing}
	There exist absolute positive constants $\delta_{0}, C$
	such that for every	\(\delta < \delta_{0}\),
	for every $d,q \in \mathbb{N}$ with $q > C d^{3}$, the following holds over $\F_{q}$.
	\begin{enumerate}
		\item If $f \in \mathcal{P}_{d}(m, \F_{q})$, then
		      $\tabldt^{f,f_{\mathrm{lines}}}_{d}(; \mathbf{a}, \mathbf{b}, \lambda)$
		      returns \accept{} with probability $1$, over the randomness of $(\mathbf{a}, \mathbf{b}, \lambda)$.
		\item For every $f: \F_{q}^{m} \to \F_{q}$ and for every $f' : \F_{q}^{2m} \to \F_{q}^{d+1}$, we have
		      \begin{align*}
			      \Pr_{\mathbf{a},\mathbf{b},\lambda}[\tabldt^{f,f'}_{d}(; \mathbf{a}, \mathbf{b}, \lambda) \; \text{returns \reject}] \; \leq \; \delta \;
			      \implies \; \delta(f, \, \mathcal{P}_{d}(m,\F_{q})) \; \leq \; 4 \delta.
		      \end{align*}
	\end{enumerate}
	Furthermore, $\tabldt^{f,f'}_{d}(; \mathbf{a}, \mathbf{b}, \lambda)$ makes $2$ oracle queries, uses
	$\bigO(m \log q)$ bits of randomness, and runs in time $\poly(m,d,\log q)$.\\
\end{theorem}

\begin{remark}
	For low-degree testing, there has been a long line of work on achieving better parameters in terms of field size and soundness guarantee. We refer the interested reader to \cite[Section 1]{HKSS-LDT} for a detailed overview of the results of low-degree testing and also for the state-of-the-art parameters (see \cite[Theorem 1.2]{HKSS-LDT}). We use the low-degree testing from \cite{ALMSS} because the algorithm and analysis are stated in the language of the lines table.
\end{remark}

\paragraph*{}Now we discuss the local correction algorithm for degree $d$ polynomials over $\F_{q}^{m}$ from \cite{ALMSS}. We first describe the local corrector and then analyze it in \Cref{thm:local-correction}. Recall that for a function with output in $\F_{q}^{d+1}$, we interpret the output as a degree-$d$ univariate polynomial.

\begin{algobox}
	\begin{algorithm}[H]
		\caption{Local Corrector (with lines table): $\tabcorr$}
		\label{algo:local-correction}

		\DontPrintSemicolon
		\KwIn{Degree parameter $d$, evaluation point $\mathbf{a} \in \F_{q}^{m}$, and oracle access to $(f,f')$  where $f: \F_{q}^{m} \to \F_{q}$ and $f': \F_{q}^{2m} \to \F_{q}^{d+1}$ }

		\vspace{1mm}
		Sample $\mathbf{b} \sim \F_{q}^{m}$, $\lambda \sim \F_{q}^{\times}$ \;

		\vspace{1mm}

		\lIf{$f'[({\mathbf{a},\mathbf{b}})](\lambda) \neq f[\mathbf{a} + \lambda \mathbf{b}]$ }{
		\Return{\reject}
		}
		\Return{$f'[(\mathbf{a},\mathbf{b})](0)$}
	\end{algorithm}

\end{algobox}

\noindent
\begin{theorem}[Local Correction via Lines Table]\label{thm:local-correction}
	(see \cite[Proposition 7.2.2.1]{ALMSS})
	There exist absolute positive constants $\delta_{0}, C > 0$  such that for every $\delta < \delta_{0}$, for every $d,q \in \mathbb{N}$ satisfying $q > Cd$, the following holds.\newline

	\begin{enumerate}
		\item If $f$ is a degree-$d$ polynomial, then for every $\mathbf{a} \in \F_{q}^{m}$, $\tabcorr^{f,f_{\mathrm{lines}}}_{d}(\mathbf{a} ; \mathbf{b}, \lambda)$ returns $f(\mathbf{a})$ with probability $1$, over the randomness $(\mathbf{b},\lambda)$.

		\item Let $f: \F_{q}^{m} \to \F_{q}$ be any function with the condition that there exists a degree-$d$ polynomial $P$ such that $\delta(f,P) \leq \delta$. Then for every $f' : \F_{q}^{2m} \to \F_{q}^{d+1}$, and every $\mathbf{a} \in \F_{q}^{m}$, we have:\\
		      Either $\tabcorr$ returns $\reject$
		      or computes $P(\mathbf{a})$ exactly with high probability, i.e.
		      \begin{align*}
			      \Pr_{\mathbf{b}, \lambda}\brac{ \tabcorr^{f,f'}_{d}(\mathbf{a} ; \mathbf{b}, \lambda)  = P(\mathbf{a}) \quad \text{OR} \quad \tabcorr^{f,f'}_{d}(\mathbf{a} ; \mathbf{b}, \lambda) \; \text{ returns } \reject }
			      \; \geq \;
			      1 - 2\sqrt{\delta} - \dfrac{d}{q-1}.
		      \end{align*}
	\end{enumerate}
	Furthermore, $\tabcorr^{f,f'}_{d}$ makes $2$ oracle queries, uses $\bigO(m \log q)$ bits of randomness, and runs in time $\poly(m,d,\log q)$.
\end{theorem}

\section{Two Encodings}\label{sec:encodings}
In this section, we will explain two different encodings that we will use to encode our PCP proofs. The first encoding will map a degree-$d$ univariate polynomial (for us $d$ will be $\mathrm{poly}(\log n) = \mathrm{poly}(q)$) to a multivariate polynomial of $\bigO(1)$-degree in, say $d^{\varepsilon}$ variables, for some small constant $\varepsilon$. As an example, say we want to encode a degree-$d$ univariate as a degree-$2$ polynomial in $2\sqrt{d}$ variables. The idea is to express every degree from $1$ to $d$ in its base-$(\sqrt{d} + 1)$ representation (the number of digits will be $2$). Thus, we use variables $X_{i,j}$ where $i \in \set{0,1}$ denotes the digit and $j \in \{0, \ldots, \sqrt{d}\}$ denotes the value at that digit. We describe this formally in \Cref{subsec:first-encoding}. In \Cref{subsec:tests-set-ml}, we discuss its syntactic test and self-correction algorithm.\\

\noindent
The second encoding will map an element in $\F_{q}$ to an $\F_{2}$-vector of length $2^{\mathrm{poly}(\log q)}$. This encoding will be done via constant-degree Hadamard encoding (higher-degree generalization of the well-known quadratic Hadamard encoding from \cite{ALMSS}). This then naturally extends to an encoding of a $\F_{q}$-valued function. We describe this formally in \Cref{subsec:second-encoding}. In \Cref{subsec:tests-hadamard}, we discuss its syntactic test and self-correction algorithm.

\subsection{Encoding Univariate Polynomials using Set-Multilinear Encoding}\label{subsec:first-encoding}
In this subsection, we discuss our first encoding that encodes a ``high''-degree (for us, it $\mathrm{poly}(\log n)$-degree) univariate polynomial to a constant-degree multivariate polynomial, which we refer to as \emph{set-multilinear} encoding. As discussed in the introduction, this transformation allows us to perform a low-degree test with constant queries over an alphabet of size $\mathrm{poly}(\log n)$. In this and following subsections, think of $d = \mathrm{poly}(\log n)$ and $c$ as a small absolute constant.\\

\begin{definition}[Set-Multilinear Encoding]\label{defn:degree-redn}
	Fix any field $\F$. Let $d,m_{1},c \in \mathbb{N}$ with $d < m_{1}^{c}$.
	Consider the following $\F$-linear map $\setml_{d,c,m_{1}}$:
	\begin{gather*}
		\setml_{d,c,m_{1}}: \F[Y]^{\leq d} \to \F[X_{0,0},\ldots,X_{0,m_{1}-1}, \ldots, X_{c-1,0},\ldots,X_{c-1,m_{1}-1}] \\
		\setml_{d,c,m_{1}}(Y^{k}) \; = \; X_{0,k_{0}} \cdots X_{c-1,k_{c-1}},
		\quad
		\text{ where $0 \leq k_{j} < m_{1}$ are such that }
		\; k \, = \, \sum_{j=0}^{c-1} \; k_{j} \cdot m_{1}^{j}.
	\end{gather*}
	The map $\setml_{d,c,m_{1}}$ is defined on all of $\F[Y]^{\leq d}$ by linearity. In simple words, $\setml_{d,c,m_{1}}$ sends $Y^{k}$ to the monomial encoding the $m_{1}$-ary representation of the integer $k$ (which uses $c$ digits because $d < m_{1}^{c}$).\newline
	Note that every monomial in $\setml_{d,c,m_{1}}(Y^{k})$ is set-multilinear with respect to the partition $\set{X_{0,\cdot}} \sqcup \ldots \sqcup \set{X_{c-1, \cdot}}$, i.e., it is multilinear and contains exactly one variable from each $\{X_{j,\cdot}\}$.\\

	We also extend this definition to a family of univariate polynomials of degree $d$, i.e. if $F: \F^{k} \to \F[Y]^{\leq d}$, then the $\F$-linear map $\setml^{\ast}_{d,c,m_{1}}$ takes a function $F$ from $\F^{k}$ to $\F[Y]^{\leq d}$, and outputs a function from $\F^{k}$ to $\F[X_{0,0},\ldots,X_{c-1,m_{1}-1}]$ as follows:
	\begin{align*}
		\setml_{d,c,m_{1}}^{\ast}(F)(\mathbf{u}) = \setml_{d,c,m_{1}}(F(\mathbf{u})).
	\end{align*}
\end{definition}

\noindent
Next, we define the following map that will let us evaluate the univariate polynomial via evaluation of the multivariate polynomial:
\begin{align*}
	\Phi_{c,m_{1}} : \F \to (\F^{m_{1}})^{c}
\end{align*}
\begin{equation}\label{eqn:reverse-degree-redn}
	\lambda \; \mapsto \; \paren{(1, \lambda, \lambda^2, \ldots, \lambda^{m_{1}-1}), \ldots, \paren{1,\lambda^{m_{1}^{c-1}},\lambda^{2m_{1}^{c-1}}, \ldots,\lambda^{(m_{1}-1)m_{1}^{c-1}}}}.
\end{equation}

From the definition of $\setml_{d,c,m_{1}}$ and $\Phi_{c,m_{1}}$, the following claim is immediate (proved in~\Cref{app:encodings-app}).\\

\begin{claim}\label{claim:eval-set-mult-degree-redn}
	Let $d,c,m_{1} \in \mathbb{N}$ with $d < m_{1}^{c}$.
	Let $\setml_{d,c,m_{1}}$ and $\Phi_{c,m_{1}}$ be as defined in \Cref{defn:degree-redn} and \Cref{eqn:reverse-degree-redn}.
	Then for every degree-$d$ univariate polynomial $P$ and for every $\lambda \in \F$, we have,
	\begin{align*}
		P(\lambda) \; = \; (\setml_{d,c,m_{1}}(P))(\Phi_{c,m_{1}}(\lambda)).
	\end{align*}
	In simple words, evaluation of the $cm_{1}$-variate degree-$c$ polynomial $\setml_{d,c,m_{1}}(P)$ on the point $\Phi_{c,m_{1}}(\lambda)$ equals $P(\lambda)$.
\end{claim}

\noindent
Hence $\setml_{d,c,m_{1}}$ sends degree-$d$ univariate polynomials to $cm_{1}$-variate polynomials of degree-$c$ (in fact, these are set-multilinear polynomials in $c$ parts). For our final PCP construction, we will have $d = \widetilde{\bigO}(\log^{2} n)$ and we will choose $c$ to be an absolute constant.

\subsection{Syntactic Test and Self-Correction of Set-Multilinear Encoding}\label{subsec:tests-set-ml}
In this subsection, we will discuss the algorithms to do low-degree test and self-correction for degree-$c$ multivariate polynomials, \emph{over the alphabet $\F_{q}$}. Note that every $\Psi_{d,c,m_{1}}$ encoding is a degree-$c$ multivariate polynomial. We have already seen algorithms for a very similar task in \Cref{algo:low-degree-test} and \Cref{algo:local-correction}. However, their alphabet was ``large'', i.e. $\F_{q}^{d+1}$. A large alphabet allowed us to get a constant query complexity even for $d = \mathrm{poly}(\log n)$. In the degree-$c$ setting, we think of $c$ as a constant, and in that case we do not have to use the lines table. We still get constant many queries (depending on $c$) and over the alphabet $\F_{q}$.
Before we describe these algorithms, it will be useful to recall \Cref{fact:characterize-low-deg}. We start with a syntactic test.

\noindent
\begin{lemma}[Standard Low-Degree Test]\label{lemma:syntactic-test-degree-reduction}
	There exists an absolute positive constant $C$ such that for every $\delta < 1/2(c+2)^{2}$, every $c, q \in \mathbb{N}$ where
	$q$ is a power of $2$ larger than $Cc^{3}$, the following holds.\newline
	Let $\zeta$ be a generator of $\F_{q}^{\times}$ and let $\set{\zeta_{0},\ldots,\zeta_{c+1}}$ and $\set{\eta_{0},\ldots,\eta_{c+1}}$ be the same as in \Cref{fact:characterize-low-deg}. Given a function $f: \F_{q}^{m'} \to \F_{q}$ and a degree parameter $c$, we will use $\stdldt^{f}_{c}(; \mathbf{u}, \mathbf{v})$ to denote the expression
	\begin{align*}
		\stdldt^{f}_{c}(; \mathbf{u}, \mathbf{v})
		\; := \;
		\sum_{i=0}^{c+1} \eta_{i} \cdot f(\mathbf{u} + \zeta_{i} \cdot \mathbf{v}).
	\end{align*}
	\begin{enumerate}
		\item If $f \in \mathcal{P}_{c}(m', \F_{q})$ is a polynomial of degree at most $c$, then
		      \begin{align*}
			      \Pr_{\mathbf{u},\mathbf{v}}\brac{ \stdldt^{f}_{c}(; \mathbf{u}, \mathbf{v}) \, = \, 0 } \; = \; 1.
		      \end{align*}

		\item For every function $f : \F_{q}^{m'} \to \F_{q}$,
		      \begin{align*}
			      \Pr_{\mathbf{u},\mathbf{v}}\brac{ \stdldt^{f}_{c}(; \mathbf{u}, \mathbf{v}) \, \neq \, 0 } \; \leq \; \delta \; \implies \; \delta(f, \, \mathcal{P}_{c}(m', \F_{q}))  \; \leq \; 2 \delta.
		      \end{align*}
	\end{enumerate}

\end{lemma}

The proof of~\Cref{lemma:syntactic-test-degree-reduction} follows roughly the same proof as that of \cite[Theorem 4.1]{RubSud} (by using the characterization from~\Cref{fact:characterize-low-deg}). This is also a special case of a more general theorem, \cite[Theorem 2.9]{kaufman2008algebraic}.\\

We now turn to local correction in the same setting.

\noindent
\begin{lemma}[Standard Local Correction]\label{lemma:self-correction-degree-reduction}
	There exists an absolute positive constant $C$ such that for every $\delta < 1/2(c+2)^{2}$, every $c, q \in \mathbb{N}$ where
	$q$ is a power of $2$ larger than $Cc^{3}$, the following holds.\newline
	Let $\zeta$ be a generator of $\F_{q}^{\times}$ and let $\set{\zeta_{0},\ldots,\zeta_{c+1}}$ and $\set{\eta_{0},\ldots,\eta_{c+1}}$ be the same as in \Cref{fact:characterize-low-deg}. For a function $f: \F_{q}^{m'} \to \F_{q}$ with oracle access and a degree parameter $c$, we will use $\stdcorr^{f}_{c}(\mathbf{u}; \mathbf{v})$ to denote
	\begin{align*}
		\stdcorr^{f}_{c}(\mathbf{u}; \mathbf{v})
		\; := \; \sum_{i=1}^{c+1} \eta_{i} \cdot f(\mathbf{u} + \zeta_{i} \cdot \mathbf{v}).
	\end{align*}
	\begin{enumerate}
		\item If $f \in \mathcal{P}_{c}(m', \F_{q})$ is a polynomial of degree $c$, then for every point $\mathbf{u} \in \F_{q}^{m'}$,
		      \begin{align*}
			      \Pr_{\mathbf{v}}\brac{ \stdcorr^{f}_{c}(\mathbf{u}; \mathbf{v}) \, = \, f(\mathbf{u}) } \; = \; 1.
		      \end{align*}

		\item Let $f : \F_{q}^{m'} \to \F_{q}$ be any function with the condition that there exists a degree $c$ polynomial $P$ such that $\delta(f,P) \leq \delta$. Then for every point $\mathbf{u} \in \F_{q}^{m'}$
		      \begin{align*}
			      \Pr_{\mathbf{v}}\brac{ \stdcorr^{f}_{c}(\mathbf{u}; \mathbf{v}) \, = \, P(\mathbf{u}) } \; \geq \; 1 - (c+1) \delta.
		      \end{align*}
	\end{enumerate}
\end{lemma}
\begin{proof}[Proof of \Cref{lemma:self-correction-degree-reduction}]
	This is simply due to \Cref{fact:local-correction-via-lines} and a union bound.
\end{proof}

\subsection{Encoding Field Elements using Low-Degree Hadamard Codes}\label{subsec:second-encoding}
In this subsection, we will discuss an encoding of $\F_{q}$ elements to $2^{\poly\log q}$ long vectors over $\F_{2}$.
This encoding also naturally extends to $\F_{q}$-valued functions (for our PCP construction,
these $\F_{q}$-valued functions will be low-degree polynomials over $\F_{q}$).
Let $q = 2^{t}$
(recall that throughout this article, we will always have \(q\) a power of
\(2\)).
In our PCP construction, we will have $t = \bigO(\log \log n)$. We move to the definition of the encoding now.\\

\noindent
Firstly,
the inclusion \(\F_{2} \subseteq \F_{q}\)
gives \(\F_{q}\) the structure of an \(\F_{2}\)-vector space,
and throughout the paper, let $\rho:\F_{q}\to \F_2^{t}$ be an arbitrary but fixed $\F_2$-linear bijection that identifies $\F_{q}$ as the vector space $\F_{2}^{t}$.
We will be dealing with the polynomial ring $\F_{2}[Z_{1},\ldots,Z_{t}]$,
which we will denote by $\F_{2}[\mathbf{Z}]$. We define the \emph{degree-$r$ Hadamard encoding} below. This is also known as degree-$r$ long code in the literature. In this and following subsections, think of $r$ as an absolute constant, and in fact, we will use $r = 3$, which is sufficient for our construction.\\

\begin{definition}[Degree-$r$ Hadamard encoding of an $\F_{q}$ element]
	Let $q = 2^{t}$ and fix a degree parameter $r \in \mathbb{N}$. For an element $u \in \F_{q}$,
	its \emph{degree-$r$ Hadamard encoding} is a function $\had{r}{u}: \F_{2}[\mathbf{Z}]^{\leq r} \to \F_{2}$,
	defined as follows:
	\begin{align*}
		\had{r}{u}(R) \; := \; R(\rho(u)), \quad \text{ for every } R \in \F_{2}[\mathbf{Z}]^{\leq r}.
	\end{align*}
	The encoding $\had{r}{u}$ is then the evaluation table of the above function.
	The length of $\had{r}{u}$ is equal to the number of degree-$r$ polynomials in $t$ variables over $\F_{2}$,
	which is equal to $2^{\binom{t}{\leq r}}$.
\end{definition}

\noindent
For $r = 1$ and $r = 2$, we get the usual Hadamard code and the quadratic Hadamard code, respectively. We now extend the previous definition to degree-$r$ Hadamard encoding of $\F_{q}$-valued functions.\\

\begin{definition}[Degree-$r$ Hadamard encoding of an $\F_{q}$-function]\label{defn:deg-r-Hadamard}
	Let $q = 2^{t}$, fix a degree parameter $r \in \mathbb{N}$, and fix a variable parameter $k \in \mathbb{N}$. For a function $f: \F_{q}^{k} \to \F_{q}$, the degree-$r$ Hadamard encoding of $f$ is a combined degree-$r$ Hadamard encoding of each of $f$'s evaluations. More formally, it's given by a function as follows:
	\begin{gather*}
		\had{r}{f}: \F_{q}^{k} \times \F_{2}[\mathbf{Z}]^{\leq r} \; \to \F_{2} \\
		\had{r}{f}(\mathbf{a}, R) \; := \; \had{r}{f(\mathbf{a})}(R) \; = \; R(\rho(f(\mathbf{a}))).
	\end{gather*}
	We also extend this definition to a family of functions as follows. For a family $\mathcal{F}$ of $\F_{q}$-functions, its degree-$r$ Hadamard encoding is a set of degree-$r$ Hadamard encodings, defined as:
	\begin{align*}
		\hadfamily{r}{\mathcal{F}} \; := \; \setcond{\had{r}{f}}{f \in \mathcal{F}}.
	\end{align*}
\end{definition}

\noindent
In other words, degree-$r$ Hadamard encoding of an $f: \F_{q}^{m} \to \F_{q}$ is obtained by applying degree-$r$ long code on every evaluation of $f$. For every $1 \leq j \leq r$, we will denote the restriction of a function $H: \F_{2}[\mathbf{Z}]^{\leq r} \to \F_{2}$ to degree-\(j\) polynomials by $H|_{\leq j}: \F_{2}[\mathbf{Z}]^{\leq j} \to \F_{2}$.

\paragraph*{Why degree-$r$?}\label{par:degr}
From the definition, we see that even degree-$1$ Hadamard encodings are enough to go from alphabet $\F_{q}$ to $\F_{2}$. So why do we need higher degree Hadamard encodings? The reason is that we are constructing PCPs for $\mathsf{NP}$ and degree-$1$ relations do not seem to be strong enough to capture $\mathsf{NP}$. Let us understand this. While constructing a PCP, we will have oracle access to some functions, and we have to check whether these functions satisfy a degree $> 1$ relation. For example, given a function $f$ as an oracle, we might be interested in knowing whether $(f^{3}-1)$ satisfies some algebraic property. Here we are checking whether a specific degree $3$ relation (i.e., $Y^{3}-1$) of $f$ has some algebraic property. Such relations are typical when expressing an $\mathsf{NP}$-complete problem using polynomials. A classic example of an $\mathsf{NP}$-complete problem that looks like this is $\mathsf{QUAD}$-$\mathsf{EQ}$, which asks whether a given set of quadratic equations has a common Boolean solution. Note that degree-$1$ relations are probably not strong enough to capture $\mathsf{NP}$-completeness, because a system of linear equations can be solved in polynomial time.

\paragraph*{}Let us imagine an abstract step that occurs in our PCP construction. Suppose $f: \F_{q}^{m} \to \F_{q}$ is a function, $P(Y)$ is a univariate polynomial of degree-$r$ (for example $Y^{3}-1$ for $r=3$), and we want to test whether $P(f)$ (which is a polynomial in $f$) satisfies some algebraic property. To test whether $P(f)$ satisfies the algebraic property or not, we might want to query the evaluation of $P(f(\mathbf{a}))$. However, we will only have access to the degree-$r$ Hadamard encoding of $f$. Instead of $P(f(\mathbf{a}))$, we can consider $\had{1}{P \circ f}(\mathbf{a}, L)$ for some $L \in \F_{2}[\mathbf{Z}]^{\leq 1}$. Therefore, we need to be able to access $\had{1}{P \circ f}(\mathbf{a}, L)$ using constant many oracle queries to $\had{r}{f}$ (we want constant many oracle queries so the final query complexity is constant). \Cref{claim:checking-relation-via-Hadamard} tells us exactly how to do this. Recall that $q=2^t$ and $\rho:\F_q\to \F_2^t$ is an $\F_2$-linear bijection. The proof is via a straightforward algebraic argument, so we skip the proof here and provide it in \Cref{app:encodings-app}.\\

\begin{lemmabox}
	\begin{claim}\label{claim:checking-relation-via-Hadamard}
		Let $r \in \mathbb{N}$ be a  degree parameter,
		and fix an arbitrary degree-$r$ univariate polynomial $P \in \F_{q}[Y]$.

		Then for every degree-$1$ polynomial $L \in \F_{2}[Z_{1},\ldots,Z_{t}]^{\leq 1}$,
		there exists a degree-$r$ polynomial $\Lambda_{P,L} \in \F_{2}[Z_{1}, \ldots, Z_{t}]^{\le r}$ such that the following holds:
		For every \(\lambda \in \F_{q}\), we have
		\begin{align*}
			\had{1}{P(\lambda)}(L) \; = \;
			\had{r}{\lambda}(\Lambda_{P,L}).
		\end{align*}
		Consequently, for every function
		$f: \F_{q}^{m} \to \F_{q}$
		and every $\mathbf{a} \in \F_{q}^{m}$,
		we have
		\begin{align*}
			\had{1}{P\circ f}(\mathbf{a}, L) \; = \;
			\had{r}{f}(\mathbf{a}, \Lambda_{P,L}).
		\end{align*}
		Finally, if \(P(Y) = \eta \cdot Y\) and \(L\) is homogeneous, then \(\Lambda_{P,L}\) is also
		homogeneous.
	\end{claim}
\end{lemmabox}

We further note that given $P$ and $L$, the above polynomial $\Lambda_{P,L}$ can be computed (in its coefficient representation) in time $\poly(\log^r q)$ by interpolation.
In the special case where \(P(Y) = \eta \cdot Y\),
we write \(\Lambda_{\eta, L}\) for \(\Lambda_{P, L}\).
We also note that \(\Lambda_{\eta, -}: \hpol \to \hpol\) is a bijection,
if \(\eta\) is non-zero.

\paragraph*{}We will make an observation that will be quite useful later for analyzing different subroutines of our PCP. The following observation says that the distance between two $\F_{q}$-functions and the distance between their respective low-degree Hadamard encodings are the same up to a constant factor.\\

\begin{lemmabox}
	\begin{observation}\label{obs:delta-fn-had-approx}
		Let $f,g : \F_{q}^{k} \to \F_{q}$ be two arbitrary functions and fix any degree parameter $r \in \mathbb{N}$.
		The following two equalities will come in handy in our calculations later. The first one is:
		\begin{gather*}
			\delta \paren{ \had{r}{f}, \, \had{r}{g} }
			\; = \;
			\frac{1}{2} \cdot \delta(f,g).
		\end{gather*}
		This follows because if $f(\mathbf{a}) = g(\mathbf{a})$ for some $\mathbf{a} \in \F_{q}^{k}$, then $\had{r}{f(\mathbf{a})} \equiv \had{r}{g(\mathbf{a})}$,
		while if \(f(\ba) \neq g(\ba)\),
		then $\rho(f(\mathbf{a})) \in \F_{2}^{t}$ and $\rho(g(\mathbf{a}))$ disagree on at least one coordinate
		\(Z_{i}\). So for every polynomial \(Q \in \F_{2}[\bZ]^{\leq r}\), we have
		\begin{align*}
			(Q+Z_{i})(\rho(f(\ba))) \, - \, (Q +Z_{i})(\rho(g(\ba)))
			\; \neq \;
			Q(\rho(f(\ba))) \, - \, Q(\rho(g(\ba))),
		\end{align*}
		which implies that exactly half of the polynomials in $\F_{2}[\mathbf{Z}]^{\leq r}$
		disagree on \(f(\ba), g(\ba)\).
		For the same reason, we get: 
		\begin{align*}
			\Pr_{\mathbf{u} \sim \F_{q}^{k}, L \sim \hpol}\brac{ \had{r}{f}(\mathbf{u}, L) \neq \had{r}{g}(\mathbf{u}, L) } \; = \; \dfrac{1}{2} \cdot \delta(f,g)
		\end{align*}
		where \(L\) is a random linear homogeneous polynomial.
	\end{observation}
\end{lemmabox}

\Cref{obs:delta-fn-had-approx} is quite useful for us because it allows us to go back and forth between checking whether $f \eqq g$ and checking whether $\had{r}{f} \eqq \had{r}{g}$. The reader can interpret \Cref{obs:delta-fn-had-approx} as follows:
\begin{itemize}
	\item If a test over alphabet $\F_{q}$ checks if $f \eqq g$, then its analogue test over alphabet $\F_{2}$ checks if $\had{r}{f} \eqq \had{r}{g}$.
	\item If $\had{r}{f}$ and $\had{r}{g}$ agree for a random pair in $\F_{q}^{m} \times \F_{2}[\mathbf{Z}]^{\leq 1}$, then $f$ and $g$ also agree on a random pair.
\end{itemize}
We will use \Cref{obs:delta-fn-had-approx} repeatedly in our analysis.

\subsection{Syntactic Test and Self-Correction of Low-Degree Hadamard Code}\label{subsec:tests-hadamard}
In this subsection, we will discuss two algorithms for degree-$r$ Hadamard encoding of a function: {\em Syntactic Test} and {\em Self-Correction}. In the syntactic test, given oracle access to a function, we want to test whether it is close to degree-$r$ Hadamard encoding of an element in $\F_{q}$. In self-correction, we are promised that the given oracle is close to degree-$r$ Hadamard encoding, and we want to compute it correctly for a specific degree-$r$ polynomial. We start by discussing the syntactic test.

\paragraph*{Syntactic Test} We are interested in the following task:
\begin{center}
	\textit{\textcolor{emphcolor}{Given oracle access to a function $H: \F_{2}[\mathbf{Z}]^{\leq r} \to \F_{2}$, decide whether $H$ is close to degree-$r$ Hadamard encoding of some element in $\F_{q}$.}}
\end{center}
For $r = 1$, \cite{BLR} gave a syntactic test for the Hadamard code (see \Cref{thm:BLR}). For $r = 2$, \cite[Proposition 34]{ALMSS} gave a syntactic test for the quadratic Hadamard code. We extend the test of \cite{ALMSS} to higher degrees.\newline
Note that a function $H: \F_{2}[\mathbf{Z}]^{\leq r} \to \F_{2}$ is a degree-$r$ Hadamard encoding of an element in $\F_{q}$ if and only if it is a restriction of a ring homomorphism
$\F_{2}[\mathbf{Z}] \to \F_{2}$. We state it formally in the following claim.\\

\begin{claim}[Local characterization of degree-$r$ Hadamard encoding]\label{claim:local-characterize-had}
	Let $H: \F_{2}[\mathbf{Z}]^{\leq r} \to \F_{2}$. There exists a $u \in \F_{q}$ such that $H = \had{r}{u}$ if and only if $H$ satisfies the following three properties:
	\begin{enumerate}
		\item The map $H$ preserves addition, i.e.,
		      for any two polynomials $P_{1}, P_{2} \in \F_{2}[\mathbf{Z}]^{\leq r}$, we have,
		      \begin{align*}
			      H(P_{1} + P_{2}) \; = \; H(P_{1}) \, + \, H(P_{2}).
		      \end{align*}
		\item The map $H$ preserves the multiplicative unit, i.e.,
		      \begin{align*}
			      H(1) \; = \; 1.
		      \end{align*}
		\item The map $H$ preserves multiplication, i.e.,
		      for any two polynomials $P, R \in \F_{2}[\mathbf{Z}]^{\leq r}$ satisfying $\deg(P) + \deg(R) \leq r$, we have,
		      \begin{align*}
			      H(P \cdot R) \; = \; H(P) \, \cdot \, H(R).
		      \end{align*}
		      Note that this condition only matters for $r \geq 2$.
	\end{enumerate}
\end{claim}
\begin{proof}[Proof of \Cref{claim:local-characterize-had}]
	For any $u \in \F_{q}$, it is easy to verify that $\had{r}{u}$ satisfies these three properties. Now suppose a function $H$ satisfies these three properties. This implies $H$ is fully determined by its values on $1,z_{1},\ldots,z_{t}$ (since they are the generators for $\F_{2}[\mathbf{Z}]$). Define $u \in \F_{q}$ as $\rho^{-1}(H(z_{1}),\ldots,H(z_{t}))$. Then it is easy to verify that $H = \had{r}{u}$. This finishes the proof of \Cref{claim:local-characterize-had}.
\end{proof}

\noindent
The above three properties, especially the first and the third properties, give us a local characterization of degree-$r$ Hadamard encodings. We design a syntactic test in \Cref{algo:syntactic-test-gen-Hadamard}, which is a \emph{robust} version of the above local characterizations. We also need an additional property from our syntactic test: if a function $H$ passes the syntactic test with high probability, then not only is it close to a degree-$r$ Hadamard encoding $\had{r}{u}$, but $H$'s restriction to degree $\leq 1$ polynomials is also close to a $\had{1}{u}$. This will be useful in designing our low-degree test and local correction algorithm over $\F_{2}$. We now describe the syntactic test.\\












\begin{algobox}

	\begin{algorithm}[H]
		\caption{Syntactic Test for Low-Degree Hadamard Codes: $\hadtest$ }
		\label{algo:syntactic-test-gen-Hadamard}

		\DontPrintSemicolon
		\KwIn{Oracle access to $H : \F_{2}[Z_{1},\ldots,Z_{t}]^{\leq r} \to \F_{2}$ and degree parameter $r$}

		\vspace{2mm}
		Sample $R, R' \sim \F_{2}[\mathbf{Z}]^{\leq r}$,
		$P \sim \F_{2}[\mathbf{Z}]^{\leq (r-1)}$
		and $L \sim \F_{2}[\mathbf{Z}]^{\leq 1}$ \;

		\vspace{2mm}

		\lIf{$H[R + R'] \neq H[R] + H[R']$ }{
			\Return{\reject}
		}

		\vspace{2mm}

		\lIf{$H[1 + R] - H[R] \neq 1$ }{
			\Return{\reject}
		}

		\vspace{2mm}
		\lIf{$H[L + R] - H[R] \neq H[L]$ }{
			\Return{\reject}
		}

		\vspace{2mm}
		\lIf{$H[L \cdot P + R] - H[R] \; \neq \;  (H[L + R] - H[R]) \cdot (H[P + R]-H[R])$ }{
			\Return{\reject}
		}
		\Return{\accept}

	\end{algorithm}
\end{algobox}












\noindent
\begin{lemma}[Syntactic Test for Generalized Hadamard Code]\label{lemma:syntactic-test-gen-Hadamard}
	Fix a degree parameter $r \in \mathbb{N}$.
	There exists a \(\delta_{0}\),
	such that for every $\delta < \delta_{0}$ and for every $H: \F_{2}[\mathbf{Z}]^{\leq r} \to \F_{2}$, the following holds.
	\begin{enumerate}
		\item \textbf{Completeness:} If there exists a $u \in \F_{q}$ such that $H$ is $\had{r}{u}$, then $\hadtest^{H}_{r}$ returns $\accept$ with probability $1$ over the internal randomness of $\hadtest$.

		\item \textbf{Soundness:} Suppose that
		      \begin{align*}
			      \Pr[\hadtest^{H}_{r} \, = \, \reject] \; \leq \; \delta.
		      \end{align*}
		      Then there exists an element \(u \in \F_{q}\), such that:
		      \begin{itemize}
			      \item $\delta(H, \had{r}{u}) \leq \delta$
			      \item $\delta(H|_{\leq 1}, \had{1}{u}) \leq 3\delta$, where $H|_{\leq 1}$ is the restriction of $H$ to $\F_{2}[\mathbf{Z}]^{\leq 1}$.
		      \end{itemize}
	\end{enumerate}
	Furthermore, $\hadtest^{H}_{r}$ makes $8$
	oracle queries, uses $\bigO \paren{\binom{t}{\leq r}}$ bits of randomness, and runs in time $t^{\mathcal{O}(r)}$.
\end{lemma}

\noindent
The proof of \Cref{lemma:syntactic-test-gen-Hadamard} is analogous to the analysis of the syntactic test for the quadratic Hadamard code in \cite[Proposition 34]{ALMSS}. We present a proof for the sake of completeness in \Cref{app:syntactic-test-gen-Hadamard}.

\paragraph*{Self-Correction.} We are interested in the following task.
\begin{center}
	\textit{\textcolor{emphcolor}{Given oracle access to a function $H: \F_{2}[\mathbf{Z}]^{\leq r} \to \F_{2}$ that is close to degree-$r$ Hadamard encoding of $u \in \F_{q}$, compute $\had{r}{u}(P)$ for a particular $P$.}}
\end{center}
Firstly, for the task to make sense, $\delta$ should be small enough to guarantee the existence of a unique $u$.
Using \Cref{obs:delta-fn-had-approx}, we get that for any two distinct $u, u' \in \F_{q}$,
their respective degree-$r$ Hadamard encodings differ on at least $1/2$-fraction,
i.e., $\delta(\had{r}{u}, \had{r}{u'}) \geq 1/2$.
Thus if $\delta < 1/4$,
there can be at most a single $u$ whose degree-$r$ Hadamard encoding is $\delta$-close to $H$.\newline
Suppose $H$ is $\delta$-close to $\had{r}{u}$,
for some $u \in \F_{q}$ ($u$ is unknown) and we are interested in computing $\had{r}{u}(P)$ using queries to $H$.
For this, we have the following simple algorithm:

\begin{center}
	\texttt{\textcolor{emphcolor}{Sample a uniformly random polynomial
			$R \sim \F_{2}[\mathbf{Z}]^{\leq r}$
			and compute
			$H[P + R] - H[R]$.}}
\end{center}
Since $R$ is uniformly random, $P+R$ is also uniformly random in $\F_{2}[\mathbf{Z}]^{\leq r}$.
Thus we have,
\begin{align*}
	\Pr_{R \sim \F_{2}[\mathbf{Z}]^{\leq r}}[H[R] \neq \had{r}{u}(R)] \, \leq \, \delta
	\quad \text{ and } \quad
	\Pr_{R \sim \F_{2}[\mathbf{Z}]^{\leq r}}[H[P+R] \neq \had{r}{u}(P+R)] \, \leq \, \delta.
\end{align*}
By union bound, $H$ agrees with $\had{r}{u}$ on both $R$ and $P+R$ with probability at least $1-2\delta$.
Thus $H[P+R] - H[R]$ equals $\had{r}{u}(P)$ with probability at least $1-2\delta$.




\section{Low-Degree Test over Binary Alphabet}\label{sec:ldt}
In this section, we will discuss our new low-degree test that works over the alphabet $\F_{2}$ (described later in \Cref{algo:new-low-deg-test-binary}). In the following subsection, \Cref{subsec:warmup-ldt}, we develop some intuition behind our final low-degree test. At the beginning of \Cref{subsec:ldt-F2}, we give more intuition behind our test, before giving the formal algorithm in \Cref{algo:new-low-deg-test-binary}. We state its guarantee in \Cref{thm:new-low-degree-test} and then prove it.\newline
As mentioned in the introduction,
we note that the low-degree test is the bottleneck in getting a PCP without proof compositions.
We overcome this by using the two encodings from the previous section.
Before stating the algorithm, which might seem cryptic at first glance, we explain the idea behind it.
As a warmup, we give a low-degree test over \emph{alphabet $\F_{q}$} (instead of directly going to alphabet $\F_{2}$).
Throughout this section, degree parameter $d$ should be thought of as $\mathrm{poly}(\log n)$, i.e., it's a growing parameter.

\subsection{Warmup: Low-Degree Test over alphabet $\F_{q}$}\label{subsec:warmup-ldt}
The goal of this subsection is to
\textit{design a low-degree test with constant query complexity over alphabet $\F_{q}$}.
Recall that $\tabldt$ (see \Cref{algo:low-degree-test})
is a low-degree test with constant query complexity but over the alphabet $\F_{q}^{d+1}$.
The reason the alphabet is $\F_{q}^{d+1}$ is that every entry of the lines table is a degree-$d$ univariate polynomial, written as a single element in $\F_{q}^{d+1}$ (and that's why the query complexity is constant).
Observe that in $\tabldt$, after a line is queried, we only need evaluation of that degree-$d$ univariate polynomial at a single point.
\begin{center}
	\textit{\textcolor{emphcolor}{The idea is to encode degree-$d$ univariate polynomials such that any evaluation of the polynomial can be inferred by reading constant-many locations of the encoding (even in presence of few errors).}}
\end{center}
We want that any evaluation can be inferred from constant-many locations so that the query complexity remains constant.
We would also need that the encoding is locally testable with constant queries.
To do this, we use the set-multilinear encoding $\Psi_{d,c,m_{1}}$ from \Cref{defn:degree-redn}.
Here, think of $c$ as a constant and $m_{1}$ as $\Theta(d^{1/c})$.
The encoding maps a degree-$d$ univariate polynomial $L(Y)$ to a constant-degree multivariate polynomial $\Psi_{d,c,m_{1}}(L)$.
Why can any evaluation of $L$ be inferred from constant-many evaluations of $\Psi_{d,c,m_{1}}(L)$?
This is because $L(\lambda)$ is equal to $\Psi_{d,c,m_{1}}(L)$ evaluated at the point $\Phi_{c,m_{1}}(\lambda)$ (see \Cref{claim:eval-set-mult-degree-redn}) and the fact that any constant-degree polynomial's evaluation at a point can be recovered from its evaluations on constant-many other points, even in the presence of error (see \Cref{lemma:self-correction-degree-reduction}). Since $d,c,m_{1}$ are clear in our context, we will simply denote these maps by $\setml$, $\Phi$, and $\Psi^{\ast}$ (recall $\setml^{\ast}$ from \Cref{defn:degree-redn}). Next, we describe the expected oracles and the queries.

\paragraph*{Expected Oracles and Queries}Suppose $P \in \mathcal{P}_{d}(m, \F_{q})$ is a degree-$d$ polynomial and let $P_{\mathrm{lines}}$ from $\F_{q}^{2m}$ to $\F_{q}^{d+1}$ be its lines table. Recall that $\setml^{\ast}(P_{\mathrm{lines}})$ is the following function:
\begin{gather*}
	\setml^{\ast}(P_{\mathrm{lines}}): \F_{q}^{2m} \times \F_{q}^{cm_{1}} \to \F_{q} \\
	((\mathbf{a}, \mathbf{b}), \mathbf{u}) \; \mapsto \; \Psi(P_{\mathrm{lines}}(\mathbf{a}, \mathbf{b}))(\mathbf{u}).
\end{gather*}
By definition of $\setml$, for every pair $(\mathbf{a}, \mathbf{b})$, $\setml(P_{\mathrm{lines}}(\mathbf{a}, \mathbf{b}))(\mathbf{X})$ is a degree-$c$ polynomial in $\mathbf{X}$-variables.
Suppose the verifier has oracle access to
$f: \F_{q}^{m} \to \F_{q}$ and to
$g: \F_{q}^{2m} \times \F_{q}^{cm_{1}} \to \F_{q}$
and wants to test whether $f$ is a degree-$d$ polynomial or far from every degree-$d$ polynomial.
The idea is to simulate $\tabldt$. Following $\tabldt$ (see \Cref{algo:low-degree-test}) and \Cref{claim:eval-set-mult-degree-redn}, the verifier samples a random pair $(\mathbf{a}, \mathbf{b}) \sim \F_{q}^{2m}$, samples a random $\lambda \sim \F_{q}^{\times}$ and would like to check whether
\begin{equation}\label{eqn:warmup-low-degree-condition}
	f[\mathbf{a} + \lambda \cdot \mathbf{b}]
	\; \text{ is equal to } \;
	g[(\mathbf{a},\mathbf{b}), \, \Phi(\lambda)].
\end{equation}
To make the above check work, the verifier needs to make a couple of modifications.
\begin{itemize}
	\item \emph{Syntactic Test on $g$:}
	      The above check is trying to simulate the $\tabldt$ test.
	      One of the reasons why $\tabldt$ works is that the lines table consists of degree-$d$ univariate polynomials, and not arbitrary functions.
	      So here also, we need to first check that
	      $g$ is obtained from applying $\Psi$ to a lines table.
	      In other words, we have to test, for every pair $(\mathbf{a}, \mathbf{b})$,
	      whether the function
	      $g[(\mathbf{a},\mathbf{b}), \, \cdot]: \F_{q}^{cm_{1}} \to \F_{q}$ is a degree-$c$ polynomial.
	      Equivalently, the verifier would like to do a low-degree test on
	      $g[(\mathbf{a},\mathbf{b}), \, \cdot]$ with degree parameter $c$,
	      for uniformly random $(\mathbf{a}, \mathbf{b})$.
	      Since $c$ is a constant, we simply use $\stdldt$ for this
	      (see \Cref{lemma:syntactic-test-degree-reduction})
	      and perform this test in constant queries.
	      Note that here we are using the local testability of our encoding $\Psi$,
	      which is essentially the local testability of Reed-Muller codes.

	\item \emph{Query point is not uniformly random:}
	      If $g[(\mathbf{a},\mathbf{b}), \cdot]$ passes the above syntactic test,
	      then we are only guaranteed that
	      $g[(\mathbf{a},\mathbf{b}), \cdot]$ is close to a degree-$c$ polynomial
	      $Q_{\mathbf{a}, \mathbf{b}}: \F_{q}^{cm_{1}} \to \F_{q}$.
	      To do the check in \Cref{eqn:warmup-low-degree-condition},
	      the verifier wants $Q_{\mathbf{a}, \mathbf{b}}(\Phi(\lambda))$,
	      but has oracle access to only $g[(\mathbf{a}, \mathbf{b}), \cdot]$.
	      Since the point $\Phi(\lambda)$ is not a uniformly random point in $\F_{q}^{cm_{1}}$,
	      querying $g[(\mathbf{a}, \mathbf{b}), \cdot]$
	      might not be equal to $Q_{\mathbf{a}, \mathbf{b}}(\Phi(\lambda))$
	      with high probability.
	      Thus the verifier self-corrects $g[(\mathbf{a},\mathbf{b}), \cdot]$
	      at the point $\Phi(\lambda)$.
	      Since we are self-correcting a function close to a constant-degree polynomial,
	      we can use $\stdcorr$ for this (see \Cref{lemma:self-correction-degree-reduction})
	      and perform this self-correction with constant queries.
\end{itemize}

\noindent
We are now ready to state our low-degree test over the alphabet $\F_{q}$.

\begin{algorithm}[H]
	\caption{Warmup: Low-Degree Test over $\F_{q}$ alphabet}
	\label{algo:warmup-low-deg-test}

	\DontPrintSemicolon

	\KwIn{Degree parameter $d$, Oracles $f: \F_{q}^{m} \to \F_{q}$, $g : \F_{q}^{2m} \times \F_{q}^{cm_{1}} \to \F_{q}$}

	\vspace{1mm}

	Sample $(\mathbf{a}, \mathbf{b}) \sim \F_{q}^{2m}$, sample $\mathbf{u}, \mathbf{v} \sim \F_{q}^{cm_{1}}$, sample $\lambda \sim \F_{q}^{\times}$ \;

	\vspace{1mm}

	\lIf{$\stdldt^{g[(\mathbf{a}, \mathbf{b}), \cdot]}_{c}(; \mathbf{u}, \mathbf{v}) \neq 0$}{\Return{\reject}}

	\lIf{ $\stdcorr^{g[(\mathbf{a}, \mathbf{b}), \cdot]}_{c}(\Phi(\lambda); \mathbf{u}) \neq f[\mathbf{a} + \lambda \cdot \mathbf{b}]$ }{\Return{\reject}	}

	\Return{\accept}

\end{algorithm}

In the next subsection, we describe the low-degree test that emulates \Cref{algo:low-degree-test} over the binary alphabet. For that, we will extend the above warm-up test by replacing the checks with their low-degree Hadamard encodings.

\subsection{Low-Degree Test over $\F_{2}$}\label{subsec:ldt-F2}
We are now ready to describe our \textit{low-degree test with constant query complexity over the alphabet $\F_{2}$}.
The warmup-test (\Cref{algo:warmup-low-deg-test}) in the previous subsection is over the alphabet $\F_{q}$,
because the oracles are $\F_{q}$-functions. Note that the checks in the warmup-test are checking whether two elements in $\F_{q}$ agree or not.
\begin{center}
	\textit{\textcolor{emphcolor}{The idea is to encode elements of $\F_{q}$ over alphabet $\F_{2}$ with good distance.}}
\end{center}
By good distance, we mean a constant relative distance.
If the encoding has constant distance, then we can check whether two $\F_{q}$ elements are equal or not by reading constant-many bits of their respective encodings.
To do this, we use degree-$r$ Hadamard
encoding\footnote{
	The reason we are using degree-$r$ Hadamard encoding
	instead of Hadamard encoding is that in a PCP for an $\mathsf{NP}$-complete problem, we would need to check degree-$r$ relations on the queried bits,
	see \Cref{par:degr}.
}
$\had{r}{\cdot}$ (see \Cref{defn:deg-r-Hadamard}). For us, $r$ would always be an absolute constant\footnote{In fact, for our construction of PCP, $r=3$ suffices.}.
As before, we would also need that this encoding is locally testable,
and we know degree-$r$ Hadamard encoding is locally testable (see \Cref{algo:syntactic-test-gen-Hadamard}).
The expected oracles will be degree-$r$ Hadamard encoding applied to each of the oracles used in the warmup-test
(i.e., Hadamard encoding applied to each evaluation of $\F_{q}$-functions)
and the queries will be random entries of degree-$r$ Hadamard encoding of appropriate $\F_{q}$ elements.

\paragraph*{Expected Oracles and Queries}
Suppose $P \in \mathcal{P}_{d}(m, \F_{q})$ is a degree-$d$ polynomial and let $\setml^{\ast}(P_{\mathrm{lines}})$ be the function defined in the previous subsection. The expected oracles are:
\begin{gather*}
	\had{r}{P} : \F_{q}^{m} \times \F_{2}[\mathbf{Z}]^{\leq r}  \to \F_{2} \\
	\had{r}{\setml^{\ast}(P_{\mathrm{lines}})} : \F_{q}^{2m} \times \F_{q}^{cm_{1}} \times \F_{2}[\mathbf{Z}]^{\leq r} \to \F_{2}.
\end{gather*}
Suppose the verifier has access to oracles $f: \F_{q}^{m} \times \F_{2}[\mathbf{Z}]^{\leq r} \to \F_{2}$ and $f' : \F_{q}^{2m} \times \F_{q}^{cm_{1}} \times \F_{2}[\mathbf{Z}]^{\leq r} \to \F_{2}$. As mentioned above, to check whether two $\F_{q}$ elements are equal or not, due to good distance, it is sufficient to check if their respective degree-$r$ Hadamard encodings agree on a random coordinate or not. With this, the idea is to do as in the warmup-test, but at a random coordinate in the Hadamard encoding. In other words, if the verifier wants to check whether $g_{1}(\mathbf{u}) = g_{2}(\mathbf{u})$ for two $\F_{q}$-functions $g_{1}, g_{2}$, then it suffices to check whether $\had{r}{g_{1}(\mathbf{u})}(L) = \had{r}{g_{2}(\mathbf{u})}(L)$ for a random $L \sim \F_{2}[\mathbf{Z}]^{\leq r}$.\\

\noindent
To make the above check work, the verifier needs to make a few modifications:
\begin{itemize}
	\item \emph{Hadamard Syntactic Test:} The above idea only works if the oracles are actual degree-$r$ Hadamard encodings. So the verifier would like to do a degree-$r$ Hadamard test on $f[\mathbf{a}, \cdot]$ and $f'[(\mathbf{a}, \mathbf{b}), \mathbf{u}, \cdot]$ for random $\mathbf{a}, \mathbf{b}, \mathbf{u}$. For this, we use $\hadtest$ (see \Cref{algo:syntactic-test-gen-Hadamard}), and since $r$ is a constant, this can be performed with a constant number of queries.

	\item \emph{Linearity holds for restricted queries:}
	      To do an analogue of $\stdldt$ and $\stdcorr$, a natural step would be to choose a random $Q \sim \F_{2}[\mathbf{Z}]^{\leq r}$ and query with respect to $Q$'s coordinate in the Hadamard encodings.
	      That is, say if $\stdldt$ queries $h(\mathbf{u})$ for some function $h$, then we query $\had{r}{h}(\mathbf{u}, Q)$.
	      However, this does not work directly because in $\stdldt$ and $\stdcorr$, we have to take a linear combination of the queries, but for arbitrary
	      $\alpha, \beta \in \F_{q}$,
	      $\had{r}{\alpha}(Q) + \had{r}{\beta}(Q)$
	      is not necessarily equal to
	      \(\had{r}{\alpha + \beta}(Q)\).
	      This however holds when \(Q\) is a homogeneous linear polynomial.
	      Thus the queries have to be with respect to a random homogeneous linear polynomial.
	\item \emph{Scalar Multiplication:}
	      Similar to the above point, if we do $\stdldt$ or $\stdcorr$
	      for some function \(h\),
	      then we need to evaluate \(\eta_{i} \cdot h(\bu)\) at some
	      homogeneous linear polynomial \(L\),
	      even though we only have query access to the Hadamard encoding of \(h\).
	      We then use \Cref{claim:checking-relation-via-Hadamard},
	      which gives us a homogeneous linear polynomial \(\etail\) (with slight abuse of notation)
	      such that
	      \(
	      \had{r}{\alpha}(\etail)
	      =
	      \had{r}{\eta_{i} \cdot  \alpha}(L)
	      \) for all $\alpha \in \F_q$.
	      We note that if \(L\) is uniformly random, then \(\etail\) is
	      also uniformly random.
\end{itemize}

\noindent
We now give the low-degree test over $\F_{2}$.

\begin{algobox}

	\begin{algorithm}[H]
		\caption{Low-Degree Test over the alphabet $\F_{2}$: $\vldt$}
		\label{algo:new-low-deg-test-binary}

		\DontPrintSemicolon

		\KwIn{Parameters $d,r$, Oracles $f: \F_{q}^{m} \times \F_{2}[\mathbf{Z}]^{\leq r} \to \F_{2}$, $f' : \F_{q}^{2m} \times \F_{q}^{cm_{1}} \times \F_{2}[\mathbf{Z}]^{\leq r} \to \F_{2}$}

		\vspace{2mm}

		Sample
		$(\mathbf{a}, \mathbf{b}) \sim \F_{q}^{2m}$,
		$\mathbf{u}, \mathbf{v} \sim \F_{q}^{cm_{1}}$,
		$\lambda \sim \F_{q}^{\times}$, and
		$L \sim \hpol$
		\vspace{1mm}

		Run $\hadtest^{f[\mathbf{a}, \cdot]}_{r}$ and $\hadtest^{f'[(\mathbf{a}, \mathbf{b}), \mathbf{u}, \cdot]}_{r}$ \;

		\vspace{2mm}

		Check if
		$f'[(\mathbf{a}, \mathbf{b}), \mathbf{u}, \Lambda_{\eta_0,L} ]+\sum_{i=1}^{c+1} f'[(\mathbf{a}, \mathbf{b}), \mathbf{u} + \zeta^{i} \mathbf{v}, \etail] \; \eqq \; 0$ \;

		\vspace{2mm}

		Check if
		$\sum_{i=1}^{c+1}  f'[(\mathbf{a},\mathbf{b}), \Phi(\lambda) + \zeta^{i} \mathbf{v}, \etail] \; \eqq \; f[\mathbf{a} + \lambda \mathbf{b}, L]$, where $\Phi$ is as in~\Cref{eqn:reverse-degree-redn}.

		\vspace{2mm}

		\lIf{any of the above tests returns \reject}{
			\Return{\reject}
		}

		\Return{\accept}

	\end{algorithm}
\end{algobox}

\noindent
\begin{theorem}[Low-Degree Test over $\F_{2}$]\label{thm:new-low-degree-test}
	There exist absolute constants $C,\delta_{0} > 0$
	such that for every choice of parameters
	$m, m_{1},c,d, d',q,r \in \mathbb{N}$
	satisfying
	\(d' = 4cd\),
	$m_{1} \geq c \geq 2$, $(m_{1}-1)^{c} \leq d < m_{1}^{c}$, \(q\) is a power of $2$  and
	\(q> Cd'^{3}\)
	and \(\delta < \delta_{0} \),
	the following holds:
	\begin{enumerate}
		\item \textbf{Completeness:}
		      If $f = \had{r}{P}$ for a degree-$d$ polynomial
		      $P \in \mathcal{P}_{d}(m,\F_{q})$ and $f' = \had{r}{P'}$ where $P'$ is
		      $\setml^*_{d, c,m_{1}}\paren{P_{\mathrm{lines}}}$
		      (recall \(\setml^*_{d, c,m_{1}}\) from \Cref{defn:degree-redn}),
		      then $\vldt^{f,f'}_{d,r}$ returns $\accept$ with probability $1$.
		\item \textbf{Soundness:}
		      For every $f,f'$ the following holds. If $\vldt^{f,f'}_{d,r}$ returns $\reject$ with probability at most $\delta$,
		      then there exists a degree-$d'$ polynomial $P \in \mathcal{P}_{d'}(m,\F_{q})$ such that
		      \begin{align*}
			      \delta\paren{ f, \; \had{r}{P} }, \; \delta\paren{ f|_{\leq 1}, \; \had{1}{P} } \; \leq \;
			      \mathcal{O}_{c}(\delta).
		      \end{align*}
	\end{enumerate}
	Furthermore, $\vldt^{f,f'}_{d,r}$ makes $\bigO_r(c)$ queries to the oracles $(f,f')$, uses $\bigO_r(m \log q + cm_{1} \log q+(\log q)^r)$ bits of randomness, and runs in time $\mathrm{poly}(m_1^c, \log q)$.
\end{theorem}

\noindent
\begin{remark}
	In the soundness condition of \Cref{thm:new-low-degree-test}, the guarantee is only that $f$ is close to a degree-$d'$ polynomial, where \emph{$d'$ is more than $d$}. Note that in \Cref{thm:low-degree-testing}, the soundness guarantee was that the oracle is close to a degree-$d$ polynomial, which is stronger than the above guarantee. The reason is roughly as follows. In \Cref{thm:low-degree-testing}, every entry of the lines table is a degree-$d$ univariate polynomial. In $\vldt$, Line $3$ checks that every entry of the supposed lines table is a degree-$c$ polynomial in $cm_{1}$ variables, and then after composing it with $\Phi$, we get degree-$d'$ univariate polynomial. Thus while reducing \Cref{thm:new-low-degree-test} to \Cref{thm:low-degree-testing}, we are only able to guarantee that every entry of the lines table is a degree-$d'$ univariate polynomial. We can get degree-$d$ by having additional tests that also test that the degree-$c$ polynomial is set-multilinear (i.e. the monomials have same degree pattern as expected). We skip these tests as they are not crucial for proving the PCP Theorem.
\end{remark}

\paragraph*{Proof Overview}As mentioned earlier, $\vldt$ is a simulation of $\tabldt$ via encodings $\Psi$ and $\had{r}{\cdot}$. Completeness follows directly from the definitions. The harder task is to show the soundness of $\vldt$, and for that we will fall back on the soundness of $\tabldt$. In simple words, we will define two functions $\mathrm{Corr}_{f}: \F_{q}^{m} \to \F_{q}$ (read as ``corrected $f$'') and $F : \F_{q}^{2m} \to \F_{q}^{d'+1}$ using the oracles $f$ and $f'$. Here $\mathrm{Corr}_{f}, F$ are proxies of $f,f'$ on which we will apply $\tabldt$. We will argue that if $\vldt$ returns $\reject$ with small probability, then $\tabldt^{\mathrm{Corr}_{f},F}$ also returns $\reject$ with small probability. Then \Cref{thm:low-degree-testing} guarantees the existence of a low-degree polynomial $P$ that is close to $\mathrm{Corr}_{f}$, and then we also get $\had{r}{P}$ is close to the input oracle $f$. We now discuss the proof in a bit more detail. Assume $\vldt$ returns $\reject$ with probability $\leq \delta$.
\begin{itemize}
	\item We define $\mathrm{Corr}_{f}$ pointwise, i.e., $\mathrm{Corr}_{f}(\mathbf{a})$ is chosen such that its Hadamard encoding minimizes the distance to $f[\mathbf{a}, \cdot]$. To define $F$, for every pair $(\mathbf{a}, \mathbf{b})$, we first choose a degree-$c$ polynomial $Q_{\mathbf{a}, \mathbf{b}}$ in $cm_{1}$ variables whose Hadamard encoding minimizes the distance to $f'[(\mathbf{a}, \mathbf{b}), \cdot, \cdot]$, and then apply the $\Phi$ map on $Q_{\mathbf{a}, \mathbf{b}}$.
	\item $\hadtest^{f[\mathbf{a}, \cdot]}_{r}$ returns $\reject$ with small probability implies $f$ is close to $\had{r}{\mathrm{Corr}_{f}}$.
	      We prove this in \Cref{claim:ldt-f-close-had}.
	\item $\hadtest^{f'[(\mathbf{a}, \mathbf{b}), \cdot, \cdot]}_{r}$
	      and Line $3$ together return $\reject$ with small probability implies that for most of the pairs $(\mathbf{a}, \mathbf{b})$, Hadamard encoding of $Q_{\mathbf{a}, \mathbf{b}}$ is close to
	      $f'[(\mathbf{a}, \mathbf{b}), \cdot, \cdot]$.
	      We prove this in \Cref{lemma:std-ldt-hadamard-encoding}.\newline
	      Since $Q_{\mathbf{a}, \mathbf{b}}$ is of degree-$c$, $Q_{\mathbf{a}, \mathbf{b}} \circ \Phi$ is a degree-$d'$ univariate polynomial. So we get for every $(\mathbf{a}, \mathbf{b})$, $F(\mathbf{a}, \mathbf{b})$ is a degree-$d'$ univariate polynomial.
	\item From the above two points and \Cref{obs:delta-fn-had-approx}, Line $4$ is roughly the same as checking if $\mathrm{Corr}_{f}(\mathbf{a} + \lambda \mathbf{b}) \eqq F(\mathbf{a}, \mathbf{b})(\lambda)$. \Cref{thm:low-degree-testing} and the second item finish the soundness analysis.
\end{itemize}

Before we delve into the proof of \Cref{thm:new-low-degree-test},
it will be convenient for us to prove the following lemma,
which will be used as an intermediate step in this and later proofs.\newline
\Cref{lemma:std-ldt-hadamard-encoding}
says that if a function
$g: \F_{q}^{k} \times \F_{2}[\mathbf{Z}]^{\leq r} \to \F_{2}$
(i.e., the same domain and codomain as a degree-$r$ Hadamard encoding of a function on $\F_{q}^{k}$)
passes $\hadtest^{g}_{r}$ and an analogue of $\stdldt$ with high probability,
then the local correction of \(g\) at any point \(\bu \in \F_{q}^{k}\)
will, with high probability,
return the value of the nearest polynomial.\\

\begin{lemma}\label{lemma:std-ldt-hadamard-encoding}
	Let $g: \F_{q}^{k} \times \F_{2}[\mathbf{Z}]^{\leq r} \to \F_{2}$ be a function,
	and let \(P: \F_{q}^{k} \to \F_{q}\) be the degree-\(c\) polynomial,
	whose degree-\(1\) Hadamard encoding is closest to \(g\). Let $\zeta_i, \eta_{i}$
	be as in \Cref{fact:characterize-low-deg}.
	Suppose $g$ satisfies the following two conditions:
	\begin{align*}
		\Pr_{\mathbf{u}}[\hadtest^{g[\mathbf{u}, \cdot]}_{r} \, \text{ returns } \, \reject] \; \leq \; \delta_{1} \; \text{ and } \;
		\Pr_{\mathbf{u}, \mathbf{v}, L \sim \hpol}
		\brac{\sum_{i=0}^{c+1} g[\mathbf{u} + \zeta_{i} \mathbf{v}, \etail ] \, \neq \, 0}
		\; \leq \; \delta_{1},
	\end{align*}
	where the first probability is also over the internal randomness of $\hadtest$.
	Then for every $\mathbf{u} \in \F_{q}^{k}$, we have,
	\begin{align*}
		\Pr_{\mathbf{v}, L \sim \hpol}
		\brac{ \sum_{i=1}^{c+1} g[\mathbf{u} + \zeta_{i} \mathbf{v}, \etail] \; \neq \; \had{1}{P}(\mathbf{u}, L) }
		\; \leq \;
		\mathcal{O}_{c}( \delta_{1}).
	\end{align*}
\end{lemma}




We defer the proof of the above lemma to
\Cref{subsec:proofs-intermediate-lemmas}
and proceed with the proof of
\Cref{thm:new-low-degree-test}.

\begin{proof}[Proof of \Cref{thm:new-low-degree-test}]
	We start by discussing the completeness of \Cref{algo:new-low-deg-test-binary}.

	\paragraph*{Completeness}
	Firstly,
	since both \(f, f'\) are the Hadamard encodings of functions,
	both the tests in Line $2$ return $\accept$.
	Now fix an arbitrary $(\ba, \bb) \in \F_{q}^{2m}$.
	Let $g(Y) = P^{(d)}_{\mathrm{lines}}(\ba, \bb)(Y)$.
	Since $g(Y)$ is a degree-$d$ univariate polynomial,
	$h(\mathbf{X}) := \Psi(g(Y))$ is a degree-$c$ polynomial.
	Fix arbitrary $\bu, \bv \in \F_{q}^{cm_{1}}$.
	\Cref{fact:characterize-low-deg} then implies
	\begin{align*}
		\sum_{i=0}^{c+1} \eta_{i} h(\mathbf{u} + \zeta_{i} \mathbf{v}) \; = \; 0.
	\end{align*}
	Then for every \textit{homogeneous} linear polynomial
	$L \in \hpol$, we get,
	\begin{align*}
		\sum_{i = 0}^{c+1} \etail(\rho(h(\mathbf{u} + \zeta_{i} \mathbf{v})))
		\; = \;
		L \paren{ \sum_{i = 0}^{c+1}  \rho(\eta_{i} h(\mathbf{u} + \zeta_{i} \mathbf{v})) }
		\; = \;
		L\paren{ \rho \paren{ \sum_{i=0}^{c+1} \eta_{i} h(\mathbf{u} + \zeta_{i} \mathbf{v}) } } \; = \; 0.
	\end{align*}
	This shows that the test in Line $3$ always returns $\accept$. Now it remains to show that Line $4$ also always returns $\accept$. By an analogous argument, Line $4$ corresponds to
	\begin{align*}
		L\paren{ \rho\paren{h(\Phi(\lambda))} } \, \eqq \, L\paren{ \rho(P(\mathbf{a} + \lambda \mathbf{b})) } \, = \, L\paren{ \rho\paren{g(\lambda)} } \,.
	\end{align*}
	Hence, by using the fact that $h(\Phi(\lambda)) = g(\lambda)$ and the completeness part of $\tabldt$ (\Cref{thm:low-degree-testing}), Line $4$ also always returns $\accept$. This finishes the discussion for completeness of $\vldt$.

	\paragraph*{Soundness}We have the following events:
	\begin{itemize}
		\item Let $\mathcal{E}_{1}$ denote the event $\hadtest^{f[\mathbf{a}, \cdot]}_{r}$ returns $\reject$.
		\item Let $\mathcal{E}_{2}$ denote the event $\hadtest^{f'[(\mathbf{a}, \mathbf{b}), \mathbf{u}, \cdot]}_{r}$ returns $\reject$.
		\item Let $\mathcal{E}_{3}$ denote the event $ \sum_{i=0}^{c+1} f'[(\mathbf{a}, \mathbf{b}), \mathbf{u} + \zeta_{i} \mathbf{v}, \etail]$ is not equal to $0$.
		\item Let $\mathcal{E}_{4}$ denote the event $ \sum_{i=1}^{c+1}  f'[(\mathbf{a},\mathbf{b}), \, \Phi(\lambda) + \zeta_{i} \mathbf{v}, \, \etail]$ is not equal to $f[\mathbf{a} + \lambda \mathbf{b}, \, L]$.
	\end{itemize}
	We want to show that if all the events happen with probability at most $\delta$, then the soundness guarantee of \Cref{thm:new-low-degree-test} holds.

	\paragraph*{Defining the proxies}
	We start by defining the proxy for $f$,
	and for \(f'\),
	which will represent the underlying function and its lines table
	respectively.
	\begin{definition}[Correction for \(f\)]\label{dfn:corr}
		For every $\mathbf{a} \in \F_{q}^{m}$, let $\mathbf{v}_{\mathbf{a}} \in \F_{q}$ denote the $\F_{q}$-element whose degree-$r$ Hadamard encoding is closest to $f[\mathbf{a}, \cdot]$, i.e.,
		\begin{align*}
			\mathbf{v}_{\mathbf{a}} \; := \; \arg \min_{\mathbf{v} \in \F_{q}} \; \delta\paren{ f[\mathbf{a}, \cdot], \; \had{r}{\mathbf{v}} }.
		\end{align*}
		We then define $\mathrm{Corr}_{f}: \F_{q}^{m} \to \F_{q}$ as $\mathrm{Corr}_{f}(\mathbf{a}) = \mathbf{v}_{\mathbf{a}}$.\\
	\end{definition}
	\begin{definition}[Correction for \(f'\)]\label{dfn:corr2}
		For every pair $(\mathbf{a}, \mathbf{b}) \in \F_{q}^{2m}$,
		let $Q_{\mathbf{a},\mathbf{b}} \in \mathcal{P}_{c}(cm_{1}, \F_{q})$
		be the degree-$c$ polynomial whose degree-$1$ Hadamard encoding is closest to $f'[(\mathbf{a}, \mathbf{b}), \cdot, \cdot]$, i.e.,
		\begin{equation}\label{eqn:defn-proxy-lines-table}
			Q_{\mathbf{a}, \mathbf{b}}
			\; := \;
			\arg \; \min_{Q \in \mathcal{P}_{c}} \;
			\delta \paren{ f'[(\mathbf{a}, \mathbf{b}), \cdot, \cdot]_{ \leq 1}, \, \had{1}{Q} }.
		\end{equation}
		Recall the map $\Phi_{c,m_{1}}$ from \Cref{eqn:reverse-degree-redn}.
		As the parameters $c,m_{1}$ are fixed, we will simply denote $\Phi_{c,m_{1}}$ by $\Phi$.
		We then define the function $F$ on $\F_{q}^{2m}$ as
		$F(\mathbf{a}, \mathbf{b}) = Q_{\mathbf{a},\mathbf{b}} \circ \Phi$, i.e., $F(\mathbf{a},\mathbf{b})(Y) = Q_{\mathbf{a}, \mathbf{b}}(\Phi(Y))$.
	\end{definition}
	By definition of $\Phi$,
	we know that $F(\mathbf{a}, \mathbf{b})(Y)$
	is a univariate polynomial of degree at most
	\begin{align*}
		c (m_{1}-1) m_{1}^{c-1}
		\leq
		cm_{1}^{c}
		\leq
		\dfrac{cdm_{1}^{c}}{(m_{1}-1)^{c}}
		\leq
		4cd,
	\end{align*}
	using that \((m_{1}-1)^{c} \leq d\) and \(c \leq m_{1}\).
	In simple words, $F$ is a collection of univariate polynomials of degree at most
	$d' =4c d$.

	Let $\mathcal{P}_{d'}$ denote the class $\mathcal{P}_{d'}(m, \F_{q})$. Via the triangle inequality, we have,
	\begin{equation}\label{eqn:ldt-soundness-triangle}
		\delta\paren{f, \had{r}{\mathcal{P}_{d'} }} \; \leq \; \delta\paren{f, \, \had{r}{\mathrm{Corr}_{f}}} \, + \, \delta\paren{\had{r}{\mathrm{Corr}_{f}}, \, \had{r}{\mathcal{P}_{d'} }}
	\end{equation}
	A similar triangle inequality holds for $r = 1$ too.
	In the above expression,
	we will bound the first summand in \Cref{claim:ldt-f-close-had}
	and will bound the second summand in \Cref{lemma:ldt-apply-og-ldt},
	which is the key technical lemma in the soundness analysis.
	We first prove \Cref{claim:ldt-f-close-had},
	which can be seen as a generalization of \Cref{lemma:syntactic-test-gen-Hadamard}
	from points to functions.
	\noindent
	\begin{claim}[Proxy for $f$]\label{claim:ldt-f-close-had}
		Assuming
		$\Pr[\mathcal{E}_{1}] \leq \delta$,
		we have
		\begin{align*}
			\delta \paren{ f, \; \had{r}{\mathrm{Corr}_{f}} },
			\delta \paren{ f|_{\leq 1}, \; \had{1}{\mathrm{Corr}_{f}} }
			\; \leq \;
			\mathcal{O}(\delta).
		\end{align*}
	\end{claim}

	\begin{lemma}\label{lemma:ldt-apply-og-ldt}
		Let $\mathrm{Corr}_{f}$ and $F$ be as defined above.
		Assuming that \(\vldt\)
		rejects with probability at most \(\delta\),
		we have,
		\begin{align*}
			\Pr\brac{ \tabldt^{\mathrm{Corr}_{f}, F}_{d'} \; \text{ returns } \; \reject }
			\; \leq \;
			\mathcal{O}_{c}( \delta).
		\end{align*}
	\end{lemma}

	\paragraph*{Why are the above claim and lemma sufficient?}
	Before discussing the proofs of \Cref{claim:ldt-f-close-had} and \Cref{lemma:ldt-apply-og-ldt},
	let us see why they are sufficient to show the soundness guarantee.
	Let \break
	\(\beta := \Pr[\tabldt^{\mathrm{Corr}_{f}, F}_{d'} \; \text{ returns } \; \reject]\).
	If \(\beta < \delta_{0}\), then \Cref{thm:low-degree-testing} gives
	\(\delta(\mathrm{Corr}_{f}, \mathcal{P}_{d'}) \leq 4\beta\),
	and otherwise trivially \(\delta(\mathrm{Corr}_{f}, \mathcal{P}_{d'}) \leq 1 \leq \beta/\delta_{0}\).
	So in both cases
	\begin{align*}
		\delta(\mathrm{Corr}_{f}, \mathcal{P}_{d'})
		\; \leq \;
		\max\paren{4, 1/\delta_{0}} \cdot \Pr[\tabldt^{\mathrm{Corr}_{f}, F}_{d'} \; \text{ returns } \; \reject].
	\end{align*}
	From \Cref{obs:delta-fn-had-approx}, we have,
	\begin{align*}
		\delta\paren{\had{r}{\mathrm{Corr}_{f}}, \, \hadfamily{r}{\mathcal{P}_{d'} }}
		\; = \;
		\frac{1}{2}\delta(\mathrm{Corr}_{f}, \mathcal{P}_{d'})
		\; \leq \;
		\mathcal{O}_{c}(\delta).
	\end{align*}
	Substituting this and \Cref{claim:ldt-f-close-had} in \Cref{eqn:ldt-soundness-triangle}, we get
	\begin{align*}
		\delta\paren{f, \, \hadfamily{r}{\mathcal{P}_{d'}}}
		\; \leq \;
		\mathcal{O}(\delta) + \mathcal{O}_{c}( \delta) = \mathcal{O}_{c}(\delta).
	\end{align*}
	Going through the above calculation also gives us the same upper bound on the relative distance between $f|_{\leq 1}$ and $\had{1}{\mathcal{P}_{d'}}$.
	Therefore, \Cref{claim:ldt-f-close-had} and \Cref{lemma:ldt-apply-og-ldt} together are sufficient to get the desired soundness guarantee.
	For the rest of the proof, we discuss their proofs, starting with the proof of \Cref{claim:ldt-f-close-had}.

	\begin{proof}[Proof of \Cref{claim:ldt-f-close-had}]
		Let
		\begin{align*}
			\rej_{\ba} :=
			\Pr[\hadtest^{f[\ba, \cdot]}_{r} \, = \, \reject]
		\end{align*}
		Then we have
		\begin{align*}
			\mathbb{E}_{\ba} \left[\rej_{\ba}  \right]
			=
			\Pr[\mathcal{E}_{1}] \leq \delta,
		\end{align*}
		so by Markov's inequality, we have
		\begin{align*}
			\Pr_{\ba} \left[\rej_{\ba} > \delta_{0}  \right]
			\leq \frac{\delta}{\delta_{0}}.
		\end{align*}
		For fixed \(\ba \in \F_{q}^{m}\)
		with \(\rej_{\ba} \leq \delta_{0}\),
		we have from the soundness of \Cref{lemma:syntactic-test-gen-Hadamard},
		that there exists a $u \in \F_{q}$ such that
		\begin{align*}
			\delta(f[\ba, \cdot],\phantom{|_{\leq 1}} \had{r}{u})
			 & \leq
			\phantom{3}\rej_{\ba}                       \\
			\delta(f[\ba, \cdot]|_{\leq 1}, \had{1}{u})
			 & \leq
			3 \rej_{\ba}.
			\intertext{Since \(\mathrm{Corr}_{f}(\ba)\)
				is defined to be the element of $\F_{q}$
				with closest degree-$r$ Hadamard encoding to \(f[\ba,\cdot]\),
				we have}
			\delta(f[\ba, \cdot],\had{r}{\mathrm{Corr}_{f}(\ba)})
			 & \leq
			\phantom{3}\rej_{\ba}
			\intertext{Averaging over every $\mathbf{a} \in \F_{q}^{m}$, we then
				get}
			\delta(f,\had{r}{\mathrm{Corr}_{f}})
			 & \leq \delta + \frac{\delta}{\delta_{0}}.
		\end{align*}
		giving us the desired bound.
		Following a similar argument, we get
		\begin{align*}
			\delta(f|_{\leq 1}, \had{1}{\mathrm{Corr}_{f}}) \leq 3\delta +
			\frac{\delta}{\delta_{0}}.
		\end{align*}
		This finishes the proof of \Cref{claim:ldt-f-close-had}.
	\end{proof}

	\begin{proof}[Proof of \Cref{lemma:ldt-apply-og-ldt}]
		$\tabldt^{\mathrm{Corr}_{f}, F}_{d'}$ returns $\reject$ if
		$\mathrm{Corr}_{f}(\mathbf{a} + \lambda \mathbf{b})
			\neq
			F(\mathbf{a}, \mathbf{b})(\lambda)$,
		for the choice of $(\mathbf{a}, \mathbf{b}, \lambda)$ as random strings.
		So our goal will be to upper-bound the following quantity:
		\begin{align*}
			\Pr_{\mathbf{a}, \mathbf{b}, \lambda}\left[
				\mathrm{Corr}_{f}(\mathbf{a} + \lambda \mathbf{b})
				\neq
				F(\mathbf{a}, \mathbf{b})(\lambda)
				\right]
		\end{align*}
		Using \Cref{obs:delta-fn-had-approx},
		we get,
		\begin{equation}\label{eqn:ldt-5}
			\Pr\left[
				\mathrm{Corr}_{f}(\mathbf{a} + \lambda \mathbf{b})
				\neq
				F(\mathbf{a}, \mathbf{b})(\lambda)
				\right]
			\; = \;
			2 \cdot \Pr\left[
			\had{r}{\mathrm{Corr}_{f}}(\mathbf{a} + \lambda \mathbf{b}, L)
			\neq
			\had{r}{F}(((\mathbf{a}, \mathbf{b}),\lambda), L)\right],
		\end{equation}
		where the probabilities are over
		$\mathbf{a}, \mathbf{b} \sim \F_{q}^{m}$,
		$\lambda \sim \F_{q}^{\times}$, and
		$L \sim \hpol$.\newline
		Thus, it suffices to upper-bound the term on the right in the above equality. For ease in writing, we will introduce a shorthand notation.\newline
		\textbf{Notation:}
		For every pair $(\mathbf{a}, \mathbf{b}) \in \F_{q}^{2m}$,
		for every $\mathbf{u} \in \F_{q}^{cm_{1}}$,
		for every $L \in \hpol$,
		we define:
		\begin{align*}
			\stdcorr^{f'[(\mathbf{a}, \mathbf{b}), \cdot, L]}_{c}(\mathbf{u}; \mathbf{v})
			:=
			\sum_{i = 1}^{c+1} f'[(\mathbf{a}, \mathbf{b}), \mathbf{u} + \zeta_{i} \mathbf{v}, \etail].
		\end{align*}
		To upper bound the term on the right,
		we will decompose it into three terms via a union bound:
		\begin{align*}
			     & \Pr\brac{ \had{r}{\mathrm{Corr}_{f}}(\mathbf{a} + \lambda \mathbf{b}, L) \neq \had{r}{F}((\mathbf{a}, \mathbf{b}), \lambda, L) }                   \\
			\leq & \Pr\brac{\had{r}{\mathrm{Corr}_{f}}(\mathbf{a} + \lambda \mathbf{b}, \, L) \neq f[\mathbf{a} + \lambda \mathbf{b}, L]}                             \\
			+    & \Pr\brac{ f[\mathbf{a} + \lambda \mathbf{b}, L] \neq \stdcorr^{f'[(\mathbf{a}, \mathbf{b}), \cdot ,L ]}_{c}(\Phi(\lambda); \mathbf{v}) }           \\
			+    & \Pr\brac{\stdcorr^{f'[(\mathbf{a}, \mathbf{b}), \cdot, L]}_{c}(\Phi(\lambda); \mathbf{v}) \neq \had{r}{F}((\mathbf{a}, \mathbf{b})(\lambda), L) }.
		\end{align*}

		\paragraph*{First and second summands}
		For the first summand, we have
		\begin{align*}
			\Pr_{\ba, \bb, L}                                      \brac{\had{r}{\mathrm{Corr}_{f}}(\mathbf{a} + \lambda \mathbf{b}, L) \neq f[\mathbf{a} + \lambda \mathbf{b}, L]} \\
			\leq \;
			2  \Pr_{\ba, \bb, L' \sim \F_{2}[\mathbf{Z}]^{\leq 1}}  \brac{
				\had{r}{\mathrm{Corr}_{f}}(\mathbf{a} + \lambda \mathbf{b},  L') \neq
			f[\mathbf{a} + \lambda \mathbf{b}, L'] }                                                                                                                                \\
			\leq 2 \cdot \delta\paren{f|_{\leq 1}, \had{1}{\mathrm{Corr}_{f}} }
			\; \leq \; \mathcal{O}(\delta).
		\end{align*}
		The first inequality is because the density of
		homogeneous $\F_{2}$-linear polynomials in degree-$1$ polynomials is $1/2$,
		and the final inequality is due to \Cref{claim:ldt-f-close-had}.

		The second summand is exactly $\Pr[\mathcal{E}_{4}]$, which is at most $\delta$, as per our assumption made at the beginning of the proof.

		\paragraph*{Third summand}
		By the definition of $F$, the third summand is equivalent to
		\begin{equation}\label{eqn:third-summand}
			\Pr\brac{
			\stdcorr^{f'[(\mathbf{a}, \mathbf{b}), \cdot, L]}_{c}(\Phi(\lambda); \mathbf{v})
			\, \neq \,
			\had{r}{Q_{\mathbf{a}, \mathbf{b}}}(\Phi(\lambda), \, L) }.
		\end{equation}
		For a fixed pair $(\mathbf{a}, \mathbf{b}) \in \F_{q}^{2m}$,
		define \(\rej_{\ba, \bb}\) as the probability that
		\(\vldt\) rejects when \((\ba, \bb)\) is sampled.
		For fixed \((\ba, \bb)\), we then have by
		\Cref{lemma:std-ldt-hadamard-encoding}
		\begin{align*}
			\Pr\brac{
			\had{r}{Q_{\mathbf{a}, \mathbf{b}}}(\Phi(\lambda), L)
			\neq
			\stdcorr^{f'[(\mathbf{a}, \mathbf{b}), \cdot, L]}_{c}(\Phi(\lambda); \mathbf{v})}
			\leq
			\mathcal{O}_{c}(\rej_{\ba, \bb}).
		\end{align*}
		From our assumption that \(\vldt\) rejects with low probability we have
		\begin{align*}
			\mathbb{E}_{\mathbf{a}, \mathbf{b}}[\mathrm{Rej}_{\mathbf{a}, \mathbf{b}}]
			\leq \delta,
		\end{align*}
		so averaging  over \((\ba, \bb)\) we get
		\begin{align*}
			\Pr\brac{\had{r}{Q_{\mathbf{a}, \mathbf{b}}}(\Phi(\lambda), L) \neq \stdcorr^{f'[(\mathbf{a}, \mathbf{b}), \cdot, L]}_{c}(\Phi(\lambda); \mathbf{v})}                                                                           \leq
			\mathcal{O}_{c}(\delta).
		\end{align*}
		Hence, this gives us an upper bound on the third summand.
		Putting it together with \Cref{eqn:ldt-5}, we get
		\begin{align*}
			\Pr_{\mathbf{a}, \mathbf{b}, \lambda}[\mathrm{Corr}_{f}(\mathbf{a} + \lambda \mathbf{b}) \neq F(\mathbf{a}, \mathbf{b})(\lambda)] \\
			= 2
			\Pr\brac{ \had{r}{\mathrm{Corr}_{f}}(\mathbf{a} + \lambda \mathbf{b}, L) \neq \had{r}{F}((\mathbf{a}, \mathbf{b}), \lambda, L) }
			\leq
			\mathcal{O}_{c}(\delta).
		\end{align*}
		This finishes the proof of \Cref{lemma:ldt-apply-og-ldt}.
	\end{proof}
	This finishes the proof of \Cref{thm:new-low-degree-test}, modulo the proof of \Cref{lemma:std-ldt-hadamard-encoding}, which we prove in the next subsection.
\end{proof}

\subsection{Proof of \Cref{lemma:std-ldt-hadamard-encoding}}\label{subsec:proofs-intermediate-lemmas}
\begin{proof}[Proof of \Cref{lemma:std-ldt-hadamard-encoding}]
	Let \(\mathcal{P}_{c}\) denote the family of degree-\(c\) polynomials
	and \(\mathrm{Corr}_{g}: \F_{q}^{k} \to \F_{q}\)
	the correction function defined in \Cref{dfn:corr}.
	Using the triangle inequality we get
	\begin{align*}
		\delta \paren{ g|_{\leq 1}, \, \hadfamily{1}{\mathcal{P}_{c}} }
		\; \leq \;
		\delta \paren{ g|_{\leq 1}, \, \had{1}{\mathrm{Corr}_{g}} }
		\; + \;
		\delta \paren{ \had{1}{\mathrm{Corr}_{g}}, \,
			\hadfamily{1}{\mathcal{P}_{c}}} .
	\end{align*}
	From the first assumption of \Cref{lemma:std-ldt-hadamard-encoding},
	along with \Cref{claim:ldt-f-close-had}
	we have
	\begin{equation}\label{eqn:ldt-lines-close-deg-d-3}
		\delta\paren{g|_{\leq 1}, \, \had{1}{\mathrm{Corr}_{g}}}
		\; \leq \; \mathcal{O}(\delta_{1}).
	\end{equation}
	\noindent
	We now focus on proving an upper bound on the distance between
	$\had{1}{\mathrm{Corr}_{g}}$ and $\hadfamily{1}{\mathcal{P}_{c}}$.
	From \Cref{obs:delta-fn-had-approx}, we have,
	\begin{align*}
		\delta\paren{\had{1}{\mathrm{Corr}_{g}}, \, \hadfamily{1}{\mathcal{P}_{c}}}
		\; = \;
		\frac{1}{2}\delta(\mathrm{Corr}_{g}, \, \mathcal{P}_{c}).
	\end{align*}
	We will show the following claim.

	\begin{claim}\label{claim:ldt-dist-corr-g}
		Let $\mathrm{Corr}_{g}$ be as defined above. Then,
		\begin{align*}
			\delta(\mathrm{Corr}_{g}, \mathcal{P}_{c})
			\; \leq \;
			\mathcal{O}_{c}(\delta_{1}).
		\end{align*}
	\end{claim}
	\begin{proof}[Proof of \Cref{claim:ldt-dist-corr-g}]

		To upper bound the distance between $\mathrm{Corr}_{g}$ and $\mathcal{P}_{c}$,
		we will use $\stdldt$ and \Cref{lemma:syntactic-test-degree-reduction}.
		In particular, we will upper bound the following probability:
		\begin{equation}\label{eqn:ldt-lines-close-deg-d-1}
			\Pr_{\mathbf{u}, \mathbf{v}}
			\brac{ \stdldt^{\mathrm{Corr}_{g}}_{c}(; \mathbf{u}, \mathbf{v}) \neq 0}
			\; = \;
			\Pr_{\mathbf{u}, \mathbf{v}}\brac{
				\sum_{i = 0}^{c+1} \eta_{i} \mathrm{Corr}_{g}(\mathbf{u} + \zeta_{i} \mathbf{v}) \neq 0}.
		\end{equation}
		Fix any pair \(\bu, \bv\) for which we have:
		\begin{align*}
			\sum_{i = 0}^{c+1} \eta_{i} \mathrm{Corr}_{g}(\mathbf{u} + \zeta_{i} \mathbf{v}) \neq 0.
		\end{align*}
		Then for $L \sim \hpol$, we get,
		\begin{align*}
			\Pr_{L}\left[\sum_{i = 0}^{c+1} \had{r}{\mathrm{Corr}_{g}}((\mathbf{u} + \zeta_{i} \mathbf{v}), \etail) \neq 0\right]
			\; = \;
			\Pr_{L}\brac{ L \paren{ \rho \paren{ \sum_{i = 0}^{c+1} \eta_{i}\mathrm{Corr}_{g}(\mathbf{u} + \zeta_{i} \mathbf{v}) } } \neq 0 }
			\; = \;
			\dfrac{1}{2},
		\end{align*}
		where for the final equality we use \Cref{obs:delta-fn-had-approx}.
		Thus,
		\begin{align*}
			\Pr_{\mathbf{u}, \mathbf{v}}\brac{
				\sum_{i = 0}^{c+1} \eta_{i} \mathrm{Corr}_{g}(\mathbf{u} + \zeta_{i} \mathbf{v}) \neq 0}
			\; = \;
			2 \cdot \Pr_{\mathbf{u}, \mathbf{v}, L}
			\left[\sum_{i = 0}^{c+1} \had{r}{\mathrm{Corr}_{g}}((\mathbf{u} + \zeta_{i} \mathbf{v}), \etail) \neq 0\right].
		\end{align*}
		So our goal now reduces to getting an upper bound on the probability of the right-hand event.

		\noindent
		For random \(\bu, \bv\),
		the points \(\bu + \zeta_{i} \bv \) are (non-independent) uniformly random,
		so by a union bound using \Cref{eqn:ldt-lines-close-deg-d-3},
		along with the fact that homogeneous linear polynomials have density \(1/2\)
		among linear polynomials, and \(\etail\) is a random homogeneous linear
		polynomial, we get
		\begin{align*}
			\Pr_{\bu, \bv, L}\left[\sum_{i = 0}^{c+1} \had{r}{\mathrm{Corr}_{g}}((\mathbf{u} + \zeta_{i} \mathbf{v}), \etail) \neq
			\sum_{i = 0}^{c+1} g[(\mathbf{u} + \zeta_{i} \mathbf{v}), \etail]\right]
			\leq
			2(c + 2) \mathcal{O}(\delta_{1}).
		\end{align*}
		From the second assumption of \Cref{lemma:std-ldt-hadamard-encoding}, we have
		\begin{align*}
			\Pr_{\bu, \bv, L}\left[\sum_{i = 0}^{c+1} g[(\mathbf{u} + \zeta_{i} \mathbf{v}), \etail]\neq 0\right]
			 & \; \leq \; \delta_{1}.
		\end{align*}
		Combining the last two inequalities, we get
		\begin{align*}
			\Pr_{\bu, \bv}\left[\sum_{i = 0}^{c+1} \eta_{i} \mathrm{Corr}_{g}(\mathbf{u} + \zeta_{i} \mathbf{v}) \neq 0 \right]
			=  2\Pr_{\mathbf{u}, \mathbf{v}, L}\left[\sum_{i = 0}^{c+1} \had{r}{\mathrm{Corr}_{g}}((\mathbf{u} + \zeta_{i} \mathbf{v}), \etail) \neq 0\right]
			\; \leq \;
			\alpha(c)	\cdot \delta_{1},
		\end{align*}
		where \(\alpha(c)\) is some constant depending on \(c\).
		We will now apply \Cref{lemma:syntactic-test-degree-reduction},
		which holds when \(\alpha(c) \cdot \delta_{1}\)
		is less than \(1/2(c + 2)^{2}\).
		However if \(\alpha(c) \cdot \delta_{1} \geq 1/2(c + 2)^{2}\),
		then we have
		\begin{align*}
			\delta(\mathrm{Corr}_{g}, \mathcal{P}_{c})
			\leq 1 \leq
			2(c + 2)^{2} \alpha(c) \cdot \delta_{1},
		\end{align*}
		So independent of the value of \(\alpha(c) \cdot \delta_{1}\) we get
		\begin{align*}
			\delta(\mathrm{Corr}_{g}, \mathcal{P}_{c})
			\; \leq \;
			2 (c + 2)^{2} \alpha(c) \cdot \delta_{1} =
			\mathcal{O}_{c}(\delta_{1})
		\end{align*}
		finishing the proof of \Cref{claim:ldt-dist-corr-g}.
	\end{proof}
	Combining
	\Cref{eqn:ldt-lines-close-deg-d-3} and
	\Cref{claim:ldt-dist-corr-g}, we get,
	\begin{align*}
		\delta\paren{ g|_{\leq 1}, \hadfamily{1}{\mathcal{P}_{c}} }
		\; \leq \; \mathcal{O}(\delta_{1}) +  \mathcal{O}_{c}(\delta_{1})
		\; \leq \; &
		\mathcal{O}_{c}(\delta_{1}).
	\end{align*}
	Now let \(P\) be a degree-\(c\) polynomial such that
	$\delta\paren{ g|_{\leq 1}, \had{1}{P}}
		\; \leq \; \mathcal{O}_{c}(\delta_{1})$.
	Since \(P\) is a degree-\(c\) polynomial,
	we have from \Cref{fact:local-correction-via-lines}
	\begin{align*}
		\sum_{i=1}^{c+1} \had{1}{P}( \mathbf{u} + \zeta_{i} \mathbf{v}, \etail)
		\; = \; \had{1}{ \sum_{i=1}^{c+1} \eta_{i} P(\mathbf{u} + \zeta_{i} \mathbf{v})}(L)
		\; = \; \had{1}{P}(\mathbf{u}, L).
	\end{align*}
	The above relation says that if for every $1 \leq i \leq c+1$,
	the queried value
	$g[\mathbf{u} + \zeta_{i} \mathbf{v}, \etail]$
	equals
	$\had{1}{P}( \mathbf{u} + \zeta_{i} \mathbf{v}, \etail)$,
	then the sum returns
	$\had{1}{P}(\mathbf{u}, L)$.

	\noindent
	For every $1 \leq i \leq (c+1)$,
	the point $\mathbf{u} + \zeta_{i}  \mathbf{v}$,
	is uniformly distributed over $\F_{q}^{k}$,
	since $\mathbf{v}$ is uniformly distributed over $\F_{q}^{k}$ and $\zeta_{i} \neq 0$.
	Since \(L \sim \hpol\), each \(\etail\) is also uniformly random, so by a union bound we get
	\begin{align*}
		\Pr_{\bv, L}\brac{ \sum_{i=1}^{c+1} g[\mathbf{u} + \zeta_{i} \mathbf{v}, \etail] \; \ne \; \had{r}{P}(\mathbf{u}, L) }
		\; \leq \; (c+1) \cdot \Pr_{\mathbf{v}, L}\brac{ g[\mathbf{v}, L] \neq \had{r}{P}(\mathbf{v}, L) }.
	\end{align*}
	Since homogeneous linear polynomials $\hpol$ have density $1/2$ in the space of degree-$1$ polynomials $\F_{2}[\mathbf{Z}]^{\leq 1}$,
	for every $\bu \in \F_{q}^{k}$ we get
	\begin{align*}
		\Pr_{\bv, L}\brac{
			\sum_{i=1}^{c+1} g[\mathbf{u} + \zeta_{i} \mathbf{v}, \etail]
			\; \ne \;
			\had{r}{P}(\mathbf{u}, L) }
		\; \leq \; 2 (c+1) \cdot \delta\paren{g|_{\leq 1}, \had{1}{P}}
		= \mathcal{O}_{c}(\delta_{1})
	\end{align*}
	which finishes the proof.
\end{proof}

\section{Local Correction over Binary Alphabet}\label{sec:lc}
In this section, we discuss our local correction algorithm for degree-$1$ Hadamard encoding of a low-degree $\F_{q}$-polynomial, at a random linear polynomial in $\F_{2}[\mathbf{Z}]$. We start by developing some intuition for our test. Then we give a formal description in \Cref{algo:local-corr-F2}, state its guarantee in \Cref{thm:local-corr-F2}, and prove it in \Cref{app:proof-local-correction}. More particularly, we are interested in the following task:
\begin{center}
	\textit{\textcolor{emphcolor}{Suppose we have oracle access to $H: \F_{q}^{m} \times \F_{2}[\mathbf{Z}]^{\leq 1} \to \F_{2}$ that is close to $\had{1}{P}$ for a degree-$d$ polynomial. For a given $\mathbf{a} \in \F_{q}^{m}$ and a random linear polynomial $L$, we want to output $\had{1}{P}(L)$ with high probability.}}
\end{center}
We show that using degree-$1$ Hadamard encoding of $\setml^{\ast}(P_{\mathrm{lines}})$, we get a local correction algorithm that has constant query complexity over the alphabet $\F_{2}$. The algorithm is completely described in \Cref{algo:local-corr-F2}, its guarantee is stated in \Cref{thm:local-corr-F2}, and it is proved in \Cref{app:proof-local-correction}. Note that we are fixing $\mathbf{a} \in \F_{q}^{m}$ but allow the second coordinate to be random in $\hpol$. This might seem slightly unconventional for a local correction task. We start by explaining why this local correction suffices for us and give an intuition behind \Cref{algo:local-corr-F2}, before describing the algorithm.

\paragraph*{Why does the above local correction suffice?}In our PCP construction, we will have a situation similar to the following: Let $P_{1}$ and $P_{2}$ be degree-$d$ polynomials on $\F_{q}^{m}$ and we have oracles $\Pi_{1}$ and $\Pi_{2}$ that are close to $\had{1}{P_{1}}$ and $\had{1}{P_{2}}$ respectively. We would like to check whether $P_{1}(\mathbf{a})$ equals $P_{2}(\mathbf{a})$ for a fixed $\mathbf{a}$, via querying the oracles $\Pi_{1}$ and $\Pi_{2}$. Similar to \Cref{obs:delta-fn-had-approx}, if $P_{1}(\mathbf{a}) \neq P_{2}(\mathbf{a})$, then $\had{1}{P_{1}}(\mathbf{a}, L)$ is not equal to $\had{1}{P_{2}}(\mathbf{a}, L)$ for a constant fraction of $L \in \hpol$. In short, we are actually interested only in local correction of the underlying low-degree polynomial, and that's why the second coordinate $L$ can be random. Also, the reason we want $L$ to be homogeneous and not arbitrary degree-$1$ is the following (it is similar to the reason we used homogeneous linear polynomials in $\vldt$, in \Cref{algo:new-low-deg-test-binary}): We would like to perform analogues of $\stdldt$ and $\stdcorr$, and $\had{1}{u+v}(Q)$ is not necessarily equal to $\had{1}{u}(Q) + \had{1}{v}(Q)$ for arbitrary $Q \in \F_{2}[\mathbf{Z}]^{\leq 1}$, but it does hold for $Q \in \hpol$.

\noindent
Our local correction algorithm over $\F_{2}$ has a close resemblance to our constant-query low-degree test over $\F_{2}$ (see \Cref{algo:new-low-deg-test-binary}). Recall the constant-query algorithm $\tabcorr$ over alphabet $\F_{q}^{d+1}$ (see \Cref{thm:local-correction}).
The idea is to simulate $\tabcorr$ using the two encodings, $\Psi$ (see \Cref{defn:degree-redn}) and $\had{1}{\cdot}$ (see \Cref{defn:deg-r-Hadamard}).
Similar to $\vldt$ (\Cref{algo:new-low-deg-test-binary}), we ask the prover to:
\begin{itemize}
	\item Firstly, for every line in $\F_{q}^{m}$, encode  its corresponding degree-$d$ univariate polynomial in the lines table using $\Psi_{d,c,m_{1}}$ for constant $c$.

	\item Secondly, now encode all the $\F_{q}$ functions using degree-$1$ Hadamard encoding.
\end{itemize}

We also want to emphasize that finally in our PCP construction, the local corrector will be applied after the low-degree test $\vldt$.
So the algorithm is designed while keeping in mind that ``the oracles are close to the expected oracles'',
i.e., if we expected the oracle to be degree-$1$ Hadamard encoding on a degree-$c$ polynomial, then it's close to such a function.
So it will be helpful to recall the guarantees provided by \Cref{thm:new-low-degree-test} before going to the local corrector.
We start by stating the expected oracles from an honest prover.

\paragraph*{Expected Oracles and Queries}
Suppose $P \in \mathcal{P}_{d}(m, \F_{q})$ is a degree-$d$ polynomial and let $\setml^{\ast}(P_{\mathrm{lines}})$ be function defined in \Cref{subsec:ldt-F2}. The expected oracles are:
\begin{gather*}
	\had{1}{P} : \F_{q}^{m} \times \F_{2}[\mathbf{Z}]^{\leq 1}  \to \F_{2} \\
	\had{1}{\setml^{\ast}(P_{\mathrm{lines}})} : \F_{q}^{2m} \times \F_{q}^{cm_{1}} \times \F_{2}[\mathbf{Z}]^{\leq 1} \to \F_{2}.
\end{gather*}
Suppose the verifier has access to oracles
$f: \F_{q}^{m} \times \F_{2}[\mathbf{Z}]^{\leq 1} \to \F_{2}$ and
$f' : \F_{q}^{2m} \times \F_{q}^{cm_{1}} \times \F_{2}[\mathbf{Z}]^{\leq 1} \to \F_{2}$.
Following $\tabcorr$, the verifier wants to check if
$f'[(\mathbf{a}, \mathbf{b}), \Phi(\lambda), \cdot]$ equals
$f[\mathbf{a} + \lambda \mathbf{b}, \cdot]$.
One would like to do something similar to Line $4$ of $\vldt$.
However, for this we need the guarantee that $f'[(\mathbf{a}, \mathbf{b}), \cdot, \cdot]$ is Hadamard encoding of a degree-$c$ polynomial in $\mathbf{X}$ variables.
From $\vldt$, we only have this guarantee for most of the lines (i.e., pairs $(\mathbf{a}, \mathbf{b})$),
but here $\mathbf{a}$ is fixed, and therefore the guarantee from $\vldt$ is not sufficient.
Thus we first do a $\hadtest$ and then an analogue of $\stdldt$ (similar to Line $3$ in $\vldt$) on the function
$f'[(\mathbf{a}, \mathbf{b}), \cdot, \cdot]$.
After this, the verifier can return the corrected value of $f'[(\mathbf{a}, \mathbf{b}), \Phi(0), L]$. This is similar to Line $4$ in $\vldt$, except there is no comparison against $f$.
We now describe our algorithm.

\begin{algobox}
	\begin{algorithm}[H]
		\caption{Local Correction over $\F_{2}$ - $\corr$}
		\label{algo:local-corr-F2}
		\DontPrintSemicolon

		\KwIn{
			Parameters $d,c$,
			Evaluation point $\mathbf{a} \in \F_{q}^{m}$,
			Oracles
			$f: \F_{q}^{m} \times \F_{2}[\mathbf{Z}]^{\leq 1} \to \F_{2}$,
			$f': \F_{q}^{2m} \times \F_{q}^{cm_{1}} \times \F_{2}[\mathbf{Z}]^{\leq 1} \to \F_{2}$, and
			$L \in \hpol$}

		\vspace{2mm}

		Sample
		$\mathbf{b} \sim \F_{q}^{m}$,
		$\mathbf{u}, \mathbf{v} \sim \F_{q}^{cm_{1}}$,
		$\lambda \in \F_{q}^{\times}$ \;

		\vspace{2mm}

		Run $\hadtest^{f'[(\mathbf{a}, \mathbf{b}), \mathbf{u}, \cdot]}_{1}$ \;

		\vspace{2mm}

		Check if $f'[(\mathbf{a}, \mathbf{b}), \mathbf{u}, \Lambda_{\eta_0,L}] + \sum_{i=1}^{c+1} f'[(\mathbf{a}, \mathbf{b}), \mathbf{u} + \zeta^{i} \mathbf{v}, \etail] \eqq 0$ \;

		\vspace{2mm}

		Check if
		$\sum_{i=1}^{c+1} f'[(\mathbf{a}, \mathbf{b}), \Phi(\lambda) + \zeta^{i} \mathbf{v}, \etail]  \eqq f[\mathbf{a} + \lambda \mathbf{b}, L]$ \;

		\vspace{2mm}

		\lIf{any of the above tests returns \reject}{
			\Return{\reject}
		}

		\vspace{2mm}

		\lElse{\Return{ $\sum_{i=1}^{c+1} f'[(\mathbf{a}, \mathbf{b}), \Phi(0) + \zeta^{i} \mathbf{v}, \etail]$ }}
	\end{algorithm}
\end{algobox}

\noindent
\begin{theorem}[Local Correction over $\F_{2}$]\label{thm:local-corr-F2}
	There exists an absolute constant $C > 0$
	such that for every set of parameters
	$m_{1},c,d, d',q \in \mathbb{N}$
	satisfying
	$m_{1} \geq c \geq 2$, $d < m_{1}^{c} \leq d'/c$, \(q\) is a power of $2$ and \(q> Cd'^{3}\),
	the following holds:
	\begin{enumerate}
		\item \textbf{Completeness:}
		      If there exists $P \in \mathcal{P}_{d}(m,\F_{q})$ such that $f = \had{1}{P}$ and $f' = \had{1}{P'}$, where $P' = \setml_{d,c,m_{1}}^{\ast}\paren{P_{\mathrm{lines}}}$,
		      then for every $\mathbf{a}$ and $L$, $\corr^{f,f'}_{d,c}(\mathbf{a},L ; \bb, \bu, \bv, \lambda)$ returns $\had{1}{P}(\mathbf{a}, L)$ with probability $1$.
		\item \textbf{Soundness:}
		      Assume there exists a polynomial
		      $P \in \mathcal{P}_{d'}(m,\F_{q})$ such that $f$ is $\delta$-close to $\had{1}{P}$.\newline
		      Then for every $f'$ and $\mathbf{a} \in \F_{q}^{m}$, the following holds:
		      If \(\corr^{f,f'}_{d,c}(\mathbf{a}, \cdot )\) rejects with
		      probability at most \(\eta\), then we have
		      \begin{align*}
			      \Pr_{L,\bb, \bu, \bv,\lambda}[
				      \corr^{f,f'}_{d,c}(\mathbf{a},L ; \mathbf{b}, \mathbf{u}, \mathbf{v}, \lambda)
				      \text{ returns }
				      \had{1}{P}(\mathbf{a}, L)]
			      \geq
			      1 - \mathcal{O}_{c}\left(d'/q + \eta + \delta\right),
		      \end{align*}
		      where in the above probability, $L \sim \hpol$.
	\end{enumerate}
	Furthermore, $\corr^{f,f'}_{d,c}(\mathbf{a}, L; )$ makes $\bigO(c)$ queries to the oracle, uses $\bigO(m \log q + cm_{1} \log q)$ bits of randomness, and runs in time $\mathrm{poly}(m_1^c, \log q)$.
\end{theorem}

\noindent
We prove \Cref{thm:local-corr-F2} in \Cref{app:proof-local-correction}.

\section{Zero-on-Grid Test over Binary Alphabet}\label{sec:zero-test}
In this section,
we will describe a test to decide whether a low-degree polynomial vanishes on a subset of its domain via its Hadamard encoding. We start by developing some intuition for our test. Then we give a formal description in \Cref{algo:zero-on-variety}, state its guarantee in \Cref{lemma:zero-on-variety}, and prove it in \Cref{app:proof-zero-test}. We are interested in the following task:
\begin{center}
	\textit{{\textcolor{emphcolor}{Fix a subset $H^{m} \subset \F_{q}^{m}$.
					Given oracle access to the Hadamard encoding of a low-degree $m$-variate polynomial $P$,
					decide whether $P$ vanishes on $H^{m}$.}}}
\end{center}
In our PCP construction,
$H$ will be of size $\bigO(\log n)$ and $P$ will be of degree $\bigO(m|H|)$.
Let us first understand a test for a simpler task.
Assume we have oracle access to $P$ and we want a test over the alphabet $\F_{q}^{d+1}$.
This exact task was done in \cite{ABSSW-PCP-one-composition},
using the vanishing certificate polynomial (see \Cref{coro:meta-vanishing-poly}).
Recall that $P$ vanishes on $H^{m}$ if and only if
the vanishing certificate polynomial exists.
The test is as follows.
The verifier gets oracle access to $P$ and to the vanishing certificate $\mathcal{M}_{P}$
(for simplicity, assume for now that both the polynomials are low-degree polynomials of the correct degree).
The verifier samples $\mathbf{a} \sim \F_{q}^{m}$ and checks:
\begin{itemize}
	\item Is $\mathcal{M}_{P}(\mathbf{a}, \mathbf{0})$ equal to $0$?
	\item Is $\mathcal{M}_{P}(\mathbf{a}, Z_{H}(\mathbf{a}))$ equal to $P(\mathbf{a})$?
\end{itemize}
The above test has good soundness again because of the polynomial distance lemma over $\F_{q}$ (\Cref{thm:odlsz}).
In the construction of PCP in \cite{ABSSW-PCP-one-composition},
they only had the guarantee that the oracle,
which is supposed to be $\mathcal{M}_{P}$,
is in fact only close to a low-degree polynomial.
Since the two queried points
$(\mathbf{a}, \mathbf{0})$ and
$(\mathbf{a}, Z_{H}(\mathbf{a}))$ are not uniformly random,
they cannot directly apply the polynomial distance lemma.
Thus~\cite{ABSSW-PCP-one-composition} used local correction on the oracle at these points.
Note that $\mathcal{M}_{P}$ is a degree-$d$ polynomial,
which means the verifier does the local correction using the lines table,
$\tabcorr$ (see \Cref{algo:local-correction}).\\
Now we come back to our setting, where we only have access to Hadamard encodings and not the low-degree polynomials themselves.
In this case, the verifier samples
$\mathbf{a} \sim \F_{q}^{m}$, $L \sim \hpol$
and checks:
\begin{itemize}
	\item Is $\had{1}{\mathcal{M}_{P}}((\mathbf{a}, \mathbf{0}), L)$ equal to $0$?
	\item Is $\had{1}{\mathcal{M}_{P}}((\mathbf{a}, Z_{H}(\mathbf{a})), L)$ equal to $\had{1}{P}(\mathbf{a}, L)$?
\end{itemize}
Finally, the verifier needs to do a local correction at these points. In the above discussion, the verifier uses $\tabcorr$. To make it work over $\F_{2}$, the verifier instead uses $\corr$ (see \Cref{algo:local-corr-F2}). We now give the description of the test.

\begin{algobox}
	\begin{algorithm}[H]
		\caption{Zero-on-Grid Test over $\F_{2}$: $\vvanish$}
		\label{algo:zero-on-variety}

		\DontPrintSemicolon

		\KwIn{Subset $H\subseteq \F_q$, Parameters $d,c$, Oracle access to
			$f: \F_{q}^{m} \times \F_{2}[\mathbf{Z}]^{\leq 1} \to \F_{2}$,
			$\Pi: \F_{q}^{2m} \times \F_{2}[\mathbf{Z}]^{\leq 1} \to \F_{2}$, and
			$\Pi' : \F_{q}^{4m} \times \F_{q}^{cm_{1}} \times \F_{2}[\mathbf{Z}]^{\leq 1} \to \F_{2}$}

		Sample
		$\mathbf{a} \sim \F_{q}^{m}$,
		$\bm{\beta} \sim \F_{q}^{2m}$,
		$\bu, \bv \sim \F_{q}^{cm_{1}}$,
		$\lambda \sim \F_{q}^{\times}$, and
		$L \sim \hpol$

		\vspace{2mm}

		Check if
		$\corr^{\Pi, \Pi'}_{d,c}((\ba, \mathbf{0}), L ; \bm{\beta},\bu, \bv, \lambda ) \; \eqq \; 0$ \;

		\vspace{2mm}

		Check if
		$\corr^{\Pi, \Pi'}_{d,c}((\ba, Z_{H}(\ba)), L ; \bm{\beta},\bu, \bv, \lambda ) \; \eqq \; f[\ba, L]$

		\vspace{2mm}

		\lIf{either of the above tests returns $\reject$}{
			\Return{\reject}
		}

		\lElse{\Return{\accept}}

	\end{algorithm}
\end{algobox}

\noindent
\begin{theorem}[Zero-on-Grid Test]\label{lemma:zero-on-variety}
	There exists an absolute constant $C > 0$
	such that for every set of parameters
	$m, m_{1},c,d, d',q \in \mathbb{N}$
	satisfying
	$m_{1} \geq c \geq 2$, $d < m_{1}^{c} \leq d'/c$, \(q\) is a power of \(2\) with  \(q> Cd'^{3}\),
	and \(H \subseteq \F_{q}\),
	the following holds:
	\begin{enumerate}
		\item \textbf{Completeness:}
		      Suppose $f = \had{1}{P}$ for a polynomial
		      $P \in \mathcal{P}_{d}(m, \F_{q})$ satisfying $P|_{H^{m}} \equiv 0$.
		      Let $M \in \mathcal{P}_{d}(2m, \F_{q})$
		      be the corresponding vanishing-certificate polynomial from \Cref{coro:meta-vanishing-poly}.
		      If $\,\Pi = \had{1}{M}$ and $\Pi' = \had{1}{M'}$ where
		      $M' = \setml_{d,c,m_{1}}^{\ast}(M_{\mathrm{lines}})$,
		      then $\vvanish^{f,\Pi,\Pi'}_{H}$ returns $\accept$ with probability $1$.

		\item \textbf{Soundness:}
		      Assume there exists a polynomial $P \in \mathcal{P}_{d'}(m,\F_{q})$
		      such that $f$ is $\delta$-close to $\had{1}{P}$
		      and there exists a polynomial	$M \in \mathcal{P}_{d'}(2m, \F_{q})$
		      such that $\Pi$ is $\delta$-close to $\had{1}{M}$.\newline
		      If $P|_{H^{m}} \not \equiv 0$,
		      then for all $\Pi'$, $\vvanish^{f,\Pi,\Pi'}_{H}$ returns $\reject$ with probability at least
		      \begin{align*}
			      \eta :=
			      \Omega_{c}(1) - \frac{2d' |H|}{q} - \delta.
		      \end{align*}
		      Furthermore, $\vvanish^{f,\Pi,\Pi'}_{H}$ makes $\bigO(c)$ oracle queries,
		      uses $\bigO(m \log q + cm_{1} \log q)$ random bits,
		      and runs in time $\mathrm{poly}(m_1^c, \log q)$.
	\end{enumerate}
\end{theorem}

\noindent
We prove \Cref{lemma:zero-on-variety} in \Cref{app:proof-zero-test}.

\section{The PCP Verifier}\label{sec:pcp}


In this section, we give the description of a PCP verifier $\mathcal{V}$ and the expected proof for the language $3$-$\mathsf{COLOR}$.

\paragraph*{Setup}\label{para:params} Throughout this section, we will use $n$ to denote the number of vertices in an input graph $G = (V,E)$. We will identify (after adding a suitable number of isolated vertices) the vertices of the input graph $V$ with $H^{m}$ for suitable $m\in \N$ and a subset $H \subseteq \F_{q}$ (for an appropriate choice of the field size $q=2^t$ and the set size $h=|H|$).
Let $\rho:\F_q \to \F_2^t$ be an $\F_2$-linear bijection and $c\ge 2$ be a constant.
Let $\omega, \zeta \in \F_q^{\times}$ where \(\omega\) is of order $3$ and \(\zeta\) is a generator respectively (assume that $q-1$ is divisible by $3$). Let $D,m_1 \in \N$ be given by $D=c_2hm$ (for some large enough constant $c_2\in \N$) and $(m_1-1)^c\le D < m_1^c$. Let $\Psi = \Psi_{D,c,m_{1}}$ and $\Psi^*=\Psi^*_{D,c,m_1}$ be from \Cref{defn:degree-redn}, $\Phi = \Phi_{c,m_{1}}$ be from \Cref{eqn:reverse-degree-redn}, and $\had{r}{\cdot}$ be from~\Cref{defn:deg-r-Hadamard}. Recall that $\cP_d(m,\F_q)$ denotes the family of $m$-variate degree-$d$ polynomials over $\F_q$.

\paragraph*{Expected Proof} Assume $G = (V,E) \in 3$-$\mathsf{COLOR}$ and let $\mathsf{Color}:H^m\to \{1,\omega,\omega^2\}$ be a proper 3-coloring of $G$. Before describing the PCP verifier, we will first describe the corresponding proof oracles. For this, we use the shorthand of $\widehat{f} = \LDE(f)$ to denote the {\em low-degree extension} of a function $f:H^{m'} \to \F_q$ (see~\Cref{fact:low-deg-extension}). There are $8=4\times 2$ proof oracles:
\begin{enumerate}
	\item Let $\chi:=\widehat{\mathsf{Color}} \in \cP_D(m,\F_q)$ and $\chi_{\mathrm{lines}}$ denote the lines table of $\chi$.\\ \\
	      \textcolor{magenta}{The proof consists of $\had{3}{\chi}: \F_{q}^{m} \times \F_{2}[\mathbf{Z}]^{\leq 3} \to \F_{2}$ and \newline
		      $\had{3}{\setml^*(\chi_{\mathrm{lines}})}:
			      \F_{q}^{2m} \times \F_{q}^{cm_{1}} \times \F_{2}[\mathbf{Z}]^{\leq 3}
			      \to
			      \F_{2}$.}

	\item Let $\mathsf{Val} \in \cP_D(m,\F_q)$ denote the {\em valid-coloring polynomial}, defined as:
	      \begin{align*}
		      \mathsf{Val}(\veca) \; := \; \prod_{i\in \set{1,\omega,\omega^{2}} }(\chi(\veca)-i) \; = \; \chi(\veca)^{3} - 1.
	      \end{align*}

	      \begin{observation}\label{obs:zero-valid-coloring}
		      Suppose $\chi: V \to \set{1,\omega,\omega^{2}}$ is an extension of a (not necessarily proper) 3-coloring of $G$. Then for every $\mathbf{a} \in V$,
		      \begin{align*}
			      \mathsf{Val}(\mathbf{a}) \; = \; 0.
		      \end{align*}
		      In other words, the polynomial $\mathsf{Val}$ is identically zero on the subset $H^{m} \subseteq \F_{q}^{m}$. It is easy to see that the converse also holds, i.e., if $\mathsf{Val}$ is identically zero on the subset $H^{m}$, then $\chi$ gives a valid coloring of $V$.
	      \end{observation}

	      \noindent
	      Since the polynomial $\mathsf{Val}$ vanishes on $H^{m}$, \Cref{coro:meta-vanishing-poly} tells us there exists a polynomial $\mathcal{M}_{\mathsf{Val}} \in \F_{q}[X_{1},\ldots,X_{m}, Y_{1}, \ldots, Y_{m}]$ satisfying the properties listed in \Cref{coro:meta-vanishing-poly}. Let $\mathcal{M}_{\mathsf{Val}, \mathrm{lines}}$ denote the lines table of $\mathcal{M}_{\mathsf{Val}}$.\\ \\
	      \textcolor{magenta}{The proof consists of $\had{1}{\mathcal{M}_{\mathsf{Val}}}: \F_{q}^{2m} \times \F_{2}[\mathbf{Z}]^{\leq 1} \to \F_{2}$ and\newline $\had{1}{\setml^*(\mathcal{M}_{\mathsf{Val}, \mathrm{lines}})}: \F_{q}^{4m} \times \F_{q}^{cm_{1}} \times \F_{2}[\mathbf{Z}]^{\leq 1} \to \F_{2}$.}

	\item Let $\chi'\in \cP_D(2m,\F_q)$ denote the polynomial defined as $$\chi'(\mathbf{a}, \mathbf{b}) = \chi(\mathbf{a}) - \chi(\mathbf{b}).$$\\
	      \textcolor{magenta}{The proof consists of $\had{3}{\chi'}: \F_{q}^{2m} \times \F_{2}[\mathbf{Z}]^{\leq 3} \to \F_{2}$ and \newline
		      $\had{3}{\Psi^*(\chi'_{\mathrm{lines}})}: \F_{q}^{4m} \times \F_{q}^{cm_{1}} \times \F_{2}[\mathbf{Z}]^{\leq 3} \to \F_{2}$.}
	\item Let $\mathsf{Prop} \in \cP_D(2m,\F_q)$ denote the {\em proper-coloring polynomial}, defined as:
	      \begin{align*}
		      \mathsf{Prop}(\veca,\vecb) \; := \;\widehat{E}(\veca,\vecb)\;\cdot\; \prod_{i\in \{1,\omega,\omega^{2}\}}(\chi'(\veca,\vecb)- i) \; = \widehat{E}(\veca,\vecb) \;\cdot \; (\chi'(\veca,\vecb)^{3}-1).
	      \end{align*}

	      \begin{observation}\label{obs:zero-proper-coloring}
		      Suppose $\chi: V \to \set{1,\omega,\omega^{2}}$ is an extension of a proper $3$-coloring of $G$. Then for every $\mathbf{a}, \mathbf{b} \in V$,
		      \begin{align*}
			      \mathsf{Prop}(\mathbf{a}, \mathbf{b}) \; = \; 0.
		      \end{align*}
		      In other words, the polynomial $\mathsf{Prop}$ is identically zero on $H^{2m}$. It is easy to see that the converse is also true, i.e., if $\mathsf{Prop}$ is identically zero on $H^{2m}$, then $\chi$ gives a proper $3$-coloring of $G$.
	      \end{observation}




	      Since $\mathsf{Prop}$ vanishes on $H^{2m}$,~\Cref{coro:meta-vanishing-poly} tells us there exists a polynomial $\mathcal{M}_{\mathsf{Prop}} \in \F_{q}[X_{1}, \ldots, X_{2m},$ $ Y_{1}, \ldots, Y_{2m}]$ satisfying the properties listed in \Cref{coro:meta-vanishing-poly}. Let $\mathcal{M}_{\mathsf{Prop}, \mathrm{lines}}$ denote the lines table of $\mathcal{M}_{\mathsf{Prop}}$. \\
	      \textcolor{magenta}{The proof consists of $\had{1}{\mathcal{M}_{\mathsf{Prop}}}: \F_{q}^{4m} \times \F_{2}[\mathbf{Z}]^{\leq 1} \to \F_{2}$ and \newline $\had{1}{\setml^*(\mathcal{M}_{\mathsf{Prop}, \mathrm{lines}})}: \F_{q}^{8m} \times \F_{q}^{cm_{1}} \times \F_{2}[\mathbf{Z}]^{\leq 1} \to \F_{2}$.}
\end{enumerate}

\paragraph*{Expected Queries}We now explain the queries our PCP verifier $\mathcal{V}$ is going to make while expecting the proof oracles as mentioned above. Firstly, $\mathcal{V}$ runs the low-degree test $\vldt$ on all the oracles. Guided by \Cref{obs:zero-valid-coloring}, $\mathcal{V}$ does the following:
\begin{itemize}
	\item Gets oracle access to a function that is supposed to be $\had{1}{\mathsf{Val}}$, using the oracle that is supposed to be $\had{3}{\chi}$. Here, by ``gets oracle access'' we mean that the function can be queried anywhere using constant many oracle queries to the oracle that is supposed to be $\had{3}{\chi}$ (we require constant many oracle queries so the final query complexity is a constant).
	\item Runs the zero-on-grid test $\vvanish$ on this new oracle.
\end{itemize}
In fact, the first step is slightly more subtle, which we elaborate on now. From \Cref{obs:zero-valid-coloring}, we see that $\had{1}{\mathsf{Val}}$ is essentially a degree-$3$ version of $\had{3}{\chi}$, i.e. for every $\mathbf{a} \in \F_{q}^{m}$ and $L' \in \F_{2}[\mathbf{Z}]^{\leq 1}$, we have, $\had{1}{\mathsf{Val}}(\mathbf{a}, L')$ is $\had{3}{\chi}(\mathbf{a}, \Lambda_{L'})$ where $\Lambda_{L'} \in \F_{2}[\mathbf{Z}]^{\leq 3}$ is given by \Cref{claim:checking-relation-via-Hadamard}. Observe that even for uniformly random $L' \sim \F_{2}[\mathbf{Z}]^{\leq 1}$, the distribution of $\Lambda_{L'}$ is far from uniform distribution\footnote{An informal way to believe this is by noting that the number of $\Lambda_{L'}$'s is at most the number of degree-$1$ polynomials in $\F_{2}[\mathbf{Z}]$, which is quite small compared to the number of polynomials in $\F_{2}[\mathbf{Z}]^{\leq 3}$.} on $\F_{2}[\mathbf{Z}]^{\leq 3}$. From $\vldt$, we only have the guarantee that we have an oracle that is close to $\had{3}{\chi}$, and thus we have to do self-correction of degree-$3$ Hadamard encoding.\newline
Again, guided by \Cref{obs:zero-proper-coloring}, $\mathcal{V}$ gets oracle access to a function that is supposed to be $\had{1}{\mathsf{Prop}}$, using the oracle that is supposed to be $\had{3}{\chi'}$, and then runs $\vvanish$ on it. Also, for the two $\vvanish$ tests to give a consistent result, $\mathcal{V}$ has to check that the oracles that are supposed to be $\had{3}{\chi}$ and $\had{3}{\chi'}$ are consistent.

\paragraph*{}Now we are ready to fully describe our PCP verifier $\mathcal{V}$ (\Cref{algo:final-verifier}). As we saw above, the expected oracles come in pairs $(\Pi_{f}, \Pi'_{f})$, where $\Pi$ is a degree-$r$ Hadamard encoding of a function $f$ and $\Pi'$ is a degree-$r$ Hadamard encoding of $\Psi^{\ast}(f_{\mathrm{lines}})$. We will use $(\widetilde{\Pi}_{f}, \widetilde{\Pi'}_{f})$ to denote the corresponding received oracles. Finally, let $\Lambda$ be as defined in \Cref{claim:checking-relation-via-Hadamard}. All randomness used in~\Cref{algo:final-verifier} is uniform and independent and $\mathbf{Z}=(Z_1,\dots,Z_t)$.

\begin{algobox}
	\begin{algorithm}[H]
		\caption{PCP Verifier $\mathcal{V}$}
		\label{algo:final-verifier}

		\DontPrintSemicolon
		{\nonl {\bf Parameters:}} $n,h,m,H,\F_q,t,\rho,c,\omega,\zeta,D,m_1,c_1,c_2$ as discussed above.\\
		\KwIn{Graph $G = (V,E)$ where $V = H^{m}$.}
		{\nonl \textbf{Oracles}:}
		\begin{multicols}{2}
			{\nonl{$ \widetilde{\Pi}_{\chi} : \F_{q}^{m} \times \F_{2}[\mathbf{Z}]^{\leq 3} \to \F_{2} $ },} \\
			{\nonl{$\widetilde{\Pi}_{ \mathcal{M}_{\mathsf{Val}} } : \F_{q}^{2m} \times \F_{2}[\mathbf{Z}]^{\leq 1} \to \F_{2}$},} \\
			{\nonl{$ \widetilde{\Pi}_{\chi'} : \F_{q}^{2m} \times \F_{2}[\mathbf{Z}]^{\leq 3} \to \F_{2} $ },} \\
			{\nonl{$\widetilde{\Pi}_{ \mathcal{M}_{\mathsf{Prop}}}  : \F_{q}^{4m} \times \F_{2}[\mathbf{Z}]^{\leq 1} \to \F_{2}$},} \\
			{\nonl{$\widetilde{\Pi'}_{\chi} : \F_{q}^{2m} \times \F_{q}^{cm_{1}} \times \F_{2}[\mathbf{Z}]^{\leq 3} \to \F_{2} $},} \\
			{\nonl{$\widetilde{\Pi'}_{ \mathcal{M}_{\mathsf{Val}}} : \F_{q}^{4m} \times \F_{q}^{cm_{1}} \times \F_{2}[\mathbf{Z}]^{\leq 1} \to \F_{2}$},} \\
			{\nonl{$ \widetilde{\Pi'}_{\chi'} : \F_{q}^{4m} \times \F_q^{cm_1} \times \F_{2}[\mathbf{Z}]^{\leq 3} \to \F_{2} $ },} \\
			{\nonl{$\widetilde{\Pi'}_{ \mathcal{M}_{\mathsf{Prop}}} : \F_{q}^{8m} \times \F_{q}^{cm_{1}} \times \F_{2}[\mathbf{Z}]^{\leq 1} \to \F_{2}$}.} \\
		\end{multicols}
		\vspace{4mm}


		Sample $\mathbf{a}, \mathbf{b} \sim \F_{q}^{m}, R\sim \F_2[\mathbf{Z}]^{\le 3},$ and $L\sim \hpol$

		Given $\veca',\vecb'\in \F_q^{m}$ and $L'\in \F_2[\mathbf{Z}]^{\le 1}$, let $\Lambda_{L'} \gets \Lambda_{Y^{3}-1, L'}$ and $\Lambda_{\veca',\vecb',L'} \gets \Lambda_{\LDE(E)(\mathbf{a'},\mathbf{b}')\cdot (Y^{3}-1), L'}$ as defined in \Cref{claim:checking-relation-via-Hadamard}

		\vspace{2mm}
		\textbf{Subroutine} \SetKwFunction{SC}{$\texttt{{\color{purple}SC}}$}
		\SC{$g: \F_{2}[\mathbf{Z}]^{\leq 3} \to \F_{2}$, \, $\Lambda \in \F_{2}[\mathbf{Z}]^{\leq 3}$}{  \KwRet{ $g[\Lambda+R] - g[R]$ }}

		\vspace{2mm}

		Given $\veca'\in \F_q^{m}$ and $L'\in \F_2[\mathbf{Z}]^{\le 1}$, let $\widetilde{\Pi}_{\mathsf{Val}}[\mathbf{a}', L'] \gets$ \SC{$\widetilde{\Pi}_{\chi}[\mathbf{a}',\cdot], \Lambda_{L'}$}

		\vspace{2mm}

		Given $\veca',\vecb'\in \F_q^{m}$ and $L'\in \F_2[\mathbf{Z}]^{\le 1}$, let $\widetilde{\Pi}_{\mathsf{Prop}}[(\mathbf{a}', \mathbf{b}'), L'] \gets$ \SC{$\widetilde{\Pi}_{\chi'}[(\mathbf{a}',\mathbf{b}'), \cdot], \Lambda_{\veca',\vecb',L'}$}

		\vspace{2mm}

		\lFor{every pair of oracles $(\widetilde{\Pi}, \widetilde{\Pi'})$}
		{
			Run $\vldt^{\widetilde{\Pi}, \widetilde{\Pi'}}_{D}$
		}

		\vspace{2mm}

		\label{line:zero-val} Run $\vvanish_H^{\widetilde{\Pi}_{\mathsf{Val}}, \widetilde{\Pi}_{ \mathcal{M}_{\mathsf{Val}} }, \widetilde{\Pi}_{ \mathcal{M}_{\mathsf{Val}} }'}$

		\vspace{2mm}

		Check if $\widetilde{\Pi}_{\chi'}[(\mathbf{a}, \mathbf{b}), L] \eqq \widetilde{\Pi}_{\chi}[\mathbf{a}, L] - \widetilde{\Pi}_{\chi}[\mathbf{b}, L]$ \;


		\vspace{2mm}

		Run $\vvanish_H^{\widetilde{\Pi}_{\mathsf{Prop}}, \widetilde{\Pi}_{ \mathcal{M}_{\mathsf{Prop}} }, \widetilde{\Pi}_{ \mathcal{M}_{\mathsf{Prop}} }'}$

		\vspace{2mm}

		\lIf{any of the above tests return \reject}{
			\Return{\reject}
		}
		\lElse{\Return{\accept}}
	\end{algorithm}
\end{algobox}

\noindent
\pcpthm*

\begin{proof}[Proof of~\Cref{thm:pcp}]
	We claim that~\Cref{algo:final-verifier} is the desired PCP verifier for an appropriate setting of parameters as follows.
	\paragraph{Choice of Parameters} Let  $c,c_1,c_2\in \N$ be sufficiently large constants. Given $n\in \N$, let $h:=\lceil \log n\rceil$, $m:=\lceil\log n/\log \log n\rceil$ (so $h^m \ge n$)\footnote{We can add a suitable number of isolated vertices to make $|V|=h^m$.}, $t := 2 \lceil c_1 \log(hm)\rceil$ (so $2^t-1$ is divisible by $3$), $D=c_2hm$, and $m_1:=\lfloor D^{1/c}\rfloor + 1$. Let $q:=2^t=\poly\log n$ and $H\subseteq \F_q$ be an arbitrary subset of size $h$.

	\paragraph{Proof size, Query complexity, and Running time} We now analyze the resulting complexity of our PCP verifier: We first note that $m_1=\bigO(m)$ assuming\footnote{Indeed, for the same reason, we cannot set the value of $c$ to be $1$ as that would make $\lines{m}$ at least $\Theta(hm)$, which is too large in our setting since we want the final oracle proof sizes $q^{O(m+\lines{m})}$ to be $\poly(n)$.} $c\ge 3$. Hence the total proof size is $q^{\bigO(m)}\cdot q^{\bigO(cm_1)}\cdot 2^{\bigO(t^3)} = q^{\bigO(m)}\cdot \poly(n) = \poly(n)$. The number of random bits used by the verifier (including all the subroutines) is $\bigO(m\log q + cm_1\log q + t^3)=\bigO(\log n)$. Since there are $\bigO(1)$ calls to $\SC,\widetilde{\Pi}_{\mathsf{Val}},\widetilde{\Pi}_{\mathsf{Prop}}, \vldt$, and $\vvanish$, the total number of oracle queries made by the verifier is also $\bigO(1)$ as required. Finally, the total running time is dominated by $q^{\bigO(m)}=\poly(n)$ needed to compute $\LDE(E)$ in Line 2.

	\paragraph*{Completeness}Suppose $\mathsf{Color}:H^m\to \{1,\omega,\omega^2\}$ is a proper 3-coloring of $G$.
	Recall the expected proof oracles from the beginning of \Cref{sec:pcp}. In particular, $\chi=\widehat{\mathsf{Color}}=\LDE(\mathsf{Color})$ and $\widehat{E} = \LDE(E)$. Consider the proof oracles as
	\begin{gather*}
		\widetilde{\Pi}_{\chi} \, = \, \had{3}{\chi},  \quad \widetilde{\Pi'}_{\chi} = \had{3}{\Psi^*(\chi_{\mathrm{lines}})}, \\
		\widetilde{\Pi}_{\mathcal{M}_{\mathsf{Val}}} \, = \, \had{1}{\mathcal{M}_{\mathsf{Val}}}, \quad \widetilde{\Pi'}_{\mathcal{M}_{\mathsf{Val}}} \, = \, \had{1}{\Psi^*(\mathcal{M}_{\mathsf{Val},\mathrm{lines}})}, \\
		\widetilde{\Pi}_{\chi'} \, = \, \had{3}{\chi'},  \quad \widetilde{\Pi'}_{\chi'} = \had{3}{\Psi^*(\chi'_{\mathrm{lines}})}, \\
		\widetilde{\Pi}_{\mathcal{M}_{\mathsf{Prop}}} \, = \, \had{1}{\mathcal{M}_{\mathsf{Prop}}}, \quad
		\widetilde{\Pi'}_{\mathcal{M}_{\mathsf{{Prop}}}} \, = \, \had{1}{\Psi^*(\mathcal{M}_{\mathsf{Prop},\mathrm{lines}})}.
	\end{gather*}
	We claim that with these oracles, the verifier $\cV$ returns $\accept$ with probability $1$. We first remark that these oracles are well-defined since all the polynomials considered above have degree at most $5hm \le D$ (taking $c_2\ge 5$). From the completeness part of \Cref{thm:new-low-degree-test}, all the tests $\vldt^{\widetilde{\Pi},\widetilde{\Pi'}}$ in Line 6 always return $\accept$.
	We now move to the zero tests. We first note that the ${\SC(g,\Lambda})$ calls in Lines 4 and 5 return $g[\Lambda]$ using the linearity of $\had{3}{\cdot}$. We now claim that the intermediate oracle function defined in Line 4 is $\widetilde{\Pi}_{\mathsf{Val}} \equiv \had{1}{\mathsf{Val}}$: this follows because for all $\veca'\in \F_q^{m}$ and $L'\in \F_2[\mathbf{Z}]^{\le 1}$, using~\Cref{claim:checking-relation-via-Hadamard}, we have
	$$\widetilde{\Pi}_{\mathsf{Val}}[\veca',L'] = \widetilde{\Pi}_\chi[\veca',\Lambda_{Y^3-1,L'}] = L'(\rho(\chi(\veca')^3-1)) = \had{1}{\mathsf{Val}}(\veca',L').$$
	Next, we observe that the check in Line 8 always passes since $L$ is homogeneous degree-$1$:
	$$\widetilde{\Pi}_{\chi'}[(\veca,\vecb),L] = L(\rho(\chi'(\veca,\vecb))) = L(\rho(\chi(\veca)-\chi(\vecb))) = L(\rho(\chi(\veca)))-L(\rho(\chi(\vecb))) =  \widetilde{\Pi}_\chi[\veca,L]-\widetilde{\Pi}_\chi[\vecb,L].$$
	Finally, we have $\widetilde{\Pi}_{\mathsf{Prop}} \equiv \had{1}{\mathsf{Prop}}$ since for all $\veca',\vecb'\in \F_q^m$ and $L'\in \F_2[\mathbf{Z}]^{\le 1}$:
	\begin{align*}\widetilde{\Pi}_{\mathsf{Prop}}[(\veca',\vecb'),L'] & = \widetilde{\Pi}_{\chi'}[(\veca',\vecb'),\Lambda_{\widehat{E}(\veca',\vecb') \cdot (Y^3-1),L'}] \\
                                                                  & = L'(\rho(\widehat{E}(\veca',\vecb')\cdot(\chi'(\veca',\vecb')^3-1)))                            \\
                                                                  & = \had{1}{\mathsf{Prop}}((\veca',\vecb'),L').\end{align*}
	Hence, using the completeness property of $\vvanish$ (\Cref{lemma:zero-on-variety}) along with~\Cref{obs:zero-valid-coloring} and~\Cref{obs:zero-proper-coloring}, we conclude that the tests in Lines 7 and 9 of $\cV$ also always return $\accept$.

	\paragraph*{Soundness}

	Suppose the rejection probability of $\cV$ is at most $\gamma$, where $\gamma=\gamma(c,c_1,c_2)>0$ is a sufficiently small constant that we will choose later, for an input graph $G$ and some received proof oracles $\widetilde{\Pi}_\chi,\widetilde{\Pi'}_\chi,\widetilde{\Pi}_{\mathcal{M}_{\mathsf{Val}}},\widetilde{\Pi'}_{\mathcal{M}_{\mathsf{Val}}},\widetilde{\Pi}_{\chi'},\widetilde{\Pi'}_{\chi'},\widetilde{\Pi}_{\mathcal{M}_{\mathsf{Prop}}},\widetilde{\Pi'}_{\mathcal{M}_{\mathsf{Prop}}}$. We will show that $G \in 3$-$\mathsf{COLOR}$. Indeed, we will show that $\widetilde{\Pi}_\chi$ is close to (the Hadamard encoding of a low-degree-extension of) a proper 3-coloring of $G$.\newline
	Since $\cV$ rejects if and only if any of the internal tests reject, we observe that each call of $\vldt$ in Line 6 rejects with probability at most $\gamma$. Hence, by the soundness guarantee of $\vldt$ (\Cref{thm:new-low-degree-test}), we note that there exist polynomials $\widetilde{\chi} \in \cP_{4cD}(m,\F_q)$, $\widetilde{\mathcal{M}_{\mathsf{Val}}} \in \cP_{4cD}(2m,\F_q)$, $\widetilde{\chi'} \in \cP_{4cD}(2m,\F_q)$, and $\widetilde{\mathcal{M}_{\mathsf{Prop}}} \in \cP_{4cD}(4m,\F_q)$ such that they are $\bigO_{c}(\gamma)$-close to their respective oracles, i.e.,
	\begin{align}
		\delta\paren{\widetilde{\Pi}_\chi, \;\had{3}{\widetilde{\chi}}}, \quad \delta\paren{\widetilde{\Pi}_\chi|_{\le 1},\had{1}{\widetilde{\chi}}} \; \leq \; \bigO_c(\gamma),\label{eqn:close-oracles}                      \\
		\delta\paren{\widetilde{\Pi}_{\mathcal{M}_{\mathsf{Val}}}, \; \had{1}{\widetilde{\mathcal{M}_{\mathsf{Val}}}}} \; \leq \; \bigO_c(\gamma),\label{eqn:close-oracles-val}                                                \\
		\delta\paren{\widetilde{\Pi}_{\chi'}, \; \had{3}{\widetilde{\chi'}}}, \quad \delta\paren{\widetilde{\Pi}_{\chi'}|_{\le 1}, \; \had{1}{\widetilde{\chi'}}} \; \leq \; \bigO_c(\gamma),\label{eqn:close-oracles-chi-prm} \\
		\delta\paren{\widetilde{\Pi}_{\mathcal{M}_{\mathsf{Prop}}}, \; \had{1}{\widetilde{\mathcal{M}_{\mathsf{Prop}}}}} \; \leq \; \bigO_c(\gamma).\label{eqn:close-oracles-prop}
	\end{align}
	Now we build oracles based on the above polynomials and argue that they are close to the corresponding intermediate oracles constructed in~\Cref{algo:final-verifier}. In particular, let
	$\widetilde{\mathsf{Val}}\in \cP_{12cD}(m,\F_q)$, $\widetilde{\chi''}\in \cP_{4cD}(2m,\F_q)$, and $\widetilde{\mathsf{Prop}}\in \cP_{20cD}(2m,\F_q)$ be the polynomials defined by
	\begin{align*}
		\widetilde{\mathsf{Val}}(\veca'):=\widetilde{\chi}(\veca')^3-1, \quad \widetilde{\chi''}(\veca',\vecb'):=\widetilde{\chi}(\veca')-\widetilde{\chi}(\vecb'), \quad \widetilde{\mathsf{Prop}}(\veca',\vecb'):=\widehat{E}(\veca',\vecb')\cdot(\widetilde{\chi''}(\veca',\vecb')^3-1).
	\end{align*}

	We quickly argue that our constructed polynomial $\widetilde{\chi''}$ in fact is identical to the polynomial $\widetilde{\chi'}$ that was guaranteed from $\vldt$ (see \Cref{eqn:close-oracles-chi-prm}).\\

	\begin{claim}\label{clm:chi-chi}
		$\widetilde{\chi ''}$ and $\widetilde{\chi '}$ are identical as polynomials.
	\end{claim}
	\begin{proof}[Proof of \Cref{clm:chi-chi}]
		Since the check in Line 8 rejects with probability at most $\gamma$, we have
		\begin{align*}
			\Pr\bigg[\widetilde{\Pi}_{\chi'}[(\mathbf{a}, \mathbf{b}), L] \ne \widetilde{\Pi}_{\chi}[\mathbf{a}, L] - \widetilde{\Pi}_{\chi}[\mathbf{b}, L]\bigg] \; \leq \; \gamma.
		\end{align*}
		From~\Cref{eqn:close-oracles} and~\Cref{eqn:close-oracles-chi-prm}, we have
		\begin{align*}
			\Pr\bigg[\widetilde{\Pi}_{\chi}[\veca,L] \ne \had{1}{\widetilde{\chi}}(\veca,L)\bigg], \quad \Pr\bigg[\widetilde{\Pi}_{\chi'}[(\veca,\vecb),L] \ne \had{1}{\widetilde{\chi'}}((\veca,\vecb),L)\bigg] \; \leq \; 2\cdot \bigO_c(\gamma).
		\end{align*}
		Combining the above three inequalities by a union bound and~\Cref{obs:delta-fn-had-approx}, we get
		$$\Pr_{\veca,\vecb\sim \F_q^m}[\widetilde{\chi'}(\veca,\vecb) \ne \widetilde{\chi}(\veca)-\widetilde{\chi}(\vecb)] \le \bigO_c(\gamma),$$ which by~\Cref{thm:odlsz} implies that $\widetilde{\chi'} \equiv \widetilde{\chi''}$ (taking $\gamma$ to be sufficiently smaller than $1-4cD/q$). This finishes the proof of \Cref{clm:chi-chi}.
	\end{proof}

	In the rest of the proof, we will show that
	\begin{equation}\label{eqn:val-tilde-zero}
		\widetilde{\mathsf{Val}}|_{H^m} \; \equiv \; 0, \text{~and}
	\end{equation}
	\begin{equation}\label{eqn:prop-tilde-zero}
		\widetilde{\mathsf{Prop}}|_{H^{2m}} \; \equiv \; 0.
	\end{equation}
	If both \Cref{eqn:val-tilde-zero} and \Cref{eqn:prop-tilde-zero} hold, then the converse direction of \Cref{obs:zero-valid-coloring} and \Cref{obs:zero-proper-coloring} imply that $\widetilde{\chi}|_{V}$ is a proper 3-coloring of $G$, which would give us the desired soundness guarantee. Thus it remains to show \Cref{eqn:val-tilde-zero} and \Cref{eqn:prop-tilde-zero} hold.

	\paragraph*{Proof of \Cref{eqn:val-tilde-zero}}
	Observe that for a uniformly random input in the domain of $\widetilde{\Pi}_{\mathsf{Val}}$, the two query points inside the $\SC$ call in Line 4 are marginally uniform in the domain of $\widetilde{\Pi}_\chi$ because $R \sim \F_{2}[\mathbf{Z}]^{\leq 3}$. Now, suppose $\veca'\in \F_q^m,L'\in \F_2[\mathbf{Z}]^{\le 1},$ and $R\in \F_2[\mathbf{Z}]^{\le 3}$ are such that $\widetilde{\Pi}_\chi[\veca',\Lambda_{L'}+R] = \had{3}{\widetilde{\chi}}(\veca',\Lambda_{L'}+R)$ and $\widetilde{\Pi}_\chi[\veca',R] = \had{3}{\widetilde{\chi}}(\veca',R)$. Then, by the definition in Line 4 and using linearity of degree-$3$ Hadamard encoding and \Cref{claim:checking-relation-via-Hadamard}, we have,
	\begin{align*}
		\widetilde{\Pi}_{\chi}[\mathbf{a}', \Lambda_{L'}+R] - \widetilde{\Pi}_{\chi}[\mathbf{a}', R] \; = \; \had{1}{\widetilde{\mathsf{Val}}}(\veca',L').
	\end{align*}
	Thus, by a union bound and~\Cref{eqn:close-oracles}, we get
	\begin{align}\label{eqn:exp-r-close}
		\E_{R\sim \F_2[\mathbf{Z}]^{\le 3}}\brac{\delta(\widetilde{\Pi}_{\mathsf{Val}}, \had{1}{\widetilde{\mathsf{Val}}}) } \; \leq \; 2\cdot \delta\paren{\widetilde{\Pi}_\chi,\had{3}{\widetilde{\chi}}} \;
		\leq \; \bigO_{c}(\gamma).
	\end{align}
	Since $\cV$ rejects with probability at most $\gamma$ by assumption, we also have
	$$\Pr_{R\sim \F_2[\mathbf{Z}]^{\le 3}}\bigg[\vvanish_H^{\widetilde{\Pi}_{\mathsf{Val}}, \widetilde{\Pi}_{ \mathcal{M}_{\mathsf{Val}} }, \widetilde{\Pi}_{ \mathcal{M}_{\mathsf{Val}} }'}\text{~returns~}\reject\bigg] \le \gamma,$$ where the probability is also over the randomness of the test.

	Combining the above two bounds, we observe that there exists $R\in \F_2[\mathbf{Z}]^{\le 3}$ such that both the following two conditions hold simultaneously for some $\gamma'=\bigO_c(\gamma)$:
	\begin{align*}\delta(\widetilde{\Pi}_{\mathsf{Val}}, \had{1}{\widetilde{\mathsf{Val}}})
		\le \gamma'
		\text{~and~} \Pr\bigg[\vvanish_H^{\widetilde{\Pi}_{\mathsf{Val}}, \widetilde{\Pi}_{ \mathcal{M}_{\mathsf{Val}} }, \widetilde{\Pi}_{ \mathcal{M}_{\mathsf{Val}} }'}\text{~returns~}\reject\bigg] \le \gamma',\end{align*}
	where we fix $\widetilde{\Pi}_{\mathsf{Val}}$ corresponding to the above $R$ in the $\SC$ call used to define $\widetilde{\Pi}_{\mathsf{Val}}$.

	Using the above bounds along with~\Cref{eqn:close-oracles-val} and applying the soundness part of $\vvanish$ (\Cref{lemma:zero-on-variety}), we conclude that $\widetilde{\mathsf{Val}}|_{H^m} \equiv 0$ assuming that $\gamma' < \eta$, where $\eta:=\Omega_c(1)-100cDh/q-\gamma'$ is from~\Cref{lemma:zero-on-variety}. The inequality $\gamma'<\eta$ indeed holds true for some suitable $\gamma'=\mathcal{O}_c(1)$, by using $q=\omega(h^2m)$. Thus we have established~\Cref{eqn:val-tilde-zero}.

	\paragraph*{Proof of \Cref{eqn:prop-tilde-zero}}The argument is similar to above, but we work with the polynomial $\widetilde{\chi'}$ instead of $\widetilde{\chi}$. Suppose $\veca',\vecb'\in \F_q^m, L'\in \F_2[\mathbf{Z}]^{\le 1}_{\hom},$ and $R\in \F_2[\mathbf{Z}]^{\le 3}$ are such that
	\begin{gather*}
		\widetilde{\Pi}_{\chi'}[(\mathbf{a}',\mathbf{b}'), \Lambda_{\veca',\vecb',L'}+R] \;  =  \; \had{3}{\widetilde{\chi'}}((\veca',\vecb'),\Lambda_{\veca',\vecb',L'}+R) \\
		\widetilde{\Pi}_{\chi'}[(\mathbf{a}',\mathbf{b}'), R] \; = \; \had{3}{\widetilde{\chi'}}((\veca',\vecb'),R)
	\end{gather*}
	Then by the definition of $\widetilde{\Pi}_{\mathsf{Prop}}$ in Line 5 of \Cref{algo:final-verifier}, we get that $\widetilde{\Pi}_{\mathsf{Prop}}[(\veca',\vecb'),L']$ is equal to
	\begin{align*}
		\had{3}{\widetilde{\chi'}}((\veca',\vecb'),\Lambda_{\veca',\vecb',L'})
		 & = \Lambda_{\veca',\vecb',L'}(\rho(\widetilde{\chi'}(\veca',\vecb')))                                                                      \\
		 & = L'(\rho(\widehat{E}(\veca',\vecb')\cdot(\widetilde{\chi'}(\veca',\vecb')^3-1))) \tag{using \Cref{claim:checking-relation-via-Hadamard}} \\
		 & = L'(\rho(\widetilde{\mathsf{Prop}}(\veca',\vecb'))) \tag{using \Cref{clm:chi-chi}}                                                       \\
		 & = \had{1}{\widetilde{\mathsf{Prop}}}((\veca',\vecb'),L').
	\end{align*}
	Thus, using~\Cref{eqn:close-oracles-chi-prm} and following a similar argument as in the derivation of~\Cref{eqn:exp-r-close}, we obtain
	\begin{align*}
		\E_{R\sim \F_2[\mathbf{Z}]^{\le 3}}[\delta(\widetilde{\Pi}_{\mathsf{Prop}}, \had{1}{\widetilde{\mathsf{Prop}}})]
		\le \bigO_c(\gamma).
	\end{align*}
	Finally, using the fact that $\widetilde{\Pi}_{\mathcal{M}_{\mathsf{Prop}}}$ is $\bigO_c(\gamma)$-close to degree-$4cD$ (i.e.,~\Cref{eqn:close-oracles-prop}) and the soundness guarantee of $\vvanish_H^{\widetilde{\Pi}_{\mathsf{Prop}}, \widetilde{\Pi}_{\mathcal{M}_{\mathsf{Prop}}},\widetilde{\Pi'
			}_{\mathcal{M}_{\mathsf{Prop}}}}$, we conclude that
	\begin{align*}
		\widetilde{\mathsf{Prop}}|_{H^{2m}} \equiv 0,
	\end{align*}
	completing the proof of~\Cref{eqn:prop-tilde-zero} and~\Cref{thm:pcp}.
\end{proof}

\section*{Acknowledgments}
We would like to thank Oded Goldreich and Prahladh Harsha for conversations leading to this work. We thank Troy Lee for pointing out a mistake in an earlier version of \Cref{lemma:syntactic-test-degree-reduction}. We thank an anonymous reviewer of SOSA 2027 for pointing us to the encoding of \cite{DFKRS11}. GPT-6.1 Sol was used for proofreading and formatting. No part of the article is AI-generated.

\printbibliography[
	heading=bibintoc,
	title={References}
] 

\addtocontents{toc}{\protect\setcounter{tocdepth}{1}}

\appendix

\section{\texorpdfstring{Proofs from \Cref{sec:preliminaries} and \Cref{sec:encodings}}{Proof of the claim relating Hadamard encodings}}\label[appendix]{app:encodings-app}

\begin{proof}[Proof of \Cref{fact:characterize-low-deg}]
	Before we prove this characterization for arbitrary $m$, it will be helpful to prove it for $m = 1$. We want to show that a function $f: \F_{q} \to \F_{q}$ is a degree-$c$ polynomial if and only if
	\begin{align*}
		\sum_{i=0}^{c+1} \eta_{i} \cdot f(u + \zeta_{i} \cdot v) = 0, &  & \text{for every } u, v \in \F_{q}.
	\end{align*}
	Let us start with the forward direction. By linearity, it suffices to show for monomials of degree at most $c$, i.e., for $1, t,t^{2},\ldots,t^{c}$. Fix any monomial $t^{k}$ for $0 \leq k \leq c$. We will first show it for $u = 0$ and $v = 1$, i.e.,
	\begin{equation}\label{eqn:low-deg-characterize-1}
		\sum_{i = 0}^{c+1} \eta_{i} \cdot \zeta_{i}^{k} = 0.
	\end{equation}
	Note that \(t^{k}\) sends the points $\zeta_{1},\ldots,\zeta_{c+1}$
	to $\zeta_{1}^{k},\ldots,\zeta_{c+1}^{k}$,
	so if we express \(t^{k}\) in Lagrange interpolation basis and substitute $t = \zeta_{0} = 0$, we get
	\begin{align*}
		\sum_{i=1}^{c+1} \eta_{i} \cdot \zeta_{i}^{k} \; = \; \zeta_{0}^{k}.
	\end{align*}
	Since $\eta_{0} = -1$, we get
	\(\sum_{i=0}^{c+1} \eta_{i} \zeta_{i}^{k} = 0\),
	thus showing \Cref{eqn:low-deg-characterize-1}.\\
	Now choose an arbitrary $u$ and $v$. Note that for every $0 \leq i \leq c+1$, we have $(u + \zeta_{i} \cdot v)^{k}$ is a sum of $u^{j} v^{k-j} \cdot \zeta_{i}^{k-j}$ for different $j$. Thus we can use \Cref{eqn:low-deg-characterize-1} repeatedly for different $j$ and this also handles arbitrary $u,v$.\\

	Now we will argue that if a function $f$ satisfies the property for every $u,v \in \F_{q}$, then $f$ is indeed a degree-$c$ polynomial. The idea is to first construct an actual degree-$c$ polynomial $P(t)$ by interpolation and then show that $P(x) = f(x)$ for every $x \in \F_{q}$. Consider the points $\set{\zeta_{0},\zeta_{1},\ldots,\zeta_{c}}$ (this corresponds to looking at the line where $u = 0$ and $v = 1$). Now interpolate a degree-$c$ polynomial $P(t)$ through these points with evaluations $\set{f(\zeta_{0}), f(\zeta_{1}), \ldots, f(\zeta_{c})}$. From the above argument and the assumption on $f$, we know the following two equalities hold:
	\begin{align*}
		\sum_{i=0}^{c+1} \eta_{i} \cdot P(\zeta_{i}) \; = \; 0 \quad \text{ and } \quad \sum_{i=0}^{c+1} \eta_{i} \cdot f(\zeta_{i}) \; = \; 0.
	\end{align*}
	Since $P(\zeta_{i}) = f(\zeta_{i})$ for every $0 \leq i \leq c$, we also get that $P(\zeta_{c+1}) = f(\zeta_{c+1})$. By choice of $\zeta_{i}$'s, we have shown that $P$ and $f$ agree on $0,\zeta,\zeta^{2},\ldots,\zeta^{c+1}$. To show that $P$ and $f$ agree on $\zeta^{c+2}$, we repeat the above argument by interpolating a degree-$c$ polynomial through points $\set{0,\zeta^2,\ldots,\zeta^{c+1}}$, which corresponds to setting $u=0$ and $v=\zeta$. Note that this will agree with $P$ from the previous step on $c+1$ points namely $\{0,\zeta^2,\ldots,\zeta^{c+1}\}$, and two degree-$c$ polynomials agree on $c+1$ points if and only if they are the same polynomial. Thus we will also get that $P(\zeta^{c+2}) = f(\zeta^{c+2})$. We keep repeating this argument till we reach $\zeta^{q-1}$, and since $\zeta$ is a generator of $\F_{q}^{\times}$, we would have shown that $P$ and $f$ agree everywhere. Therefore, $f$ is a degree-$c$ polynomial.\\

	\noindent
	Now we would argue that the above characterization holds for every $m$. We start by arguing for the forward direction. If $P: \F_{q}^{m} \to \F_{q}$ is a degree-$c$ polynomial, then for every pair $(\mathbf{u}, \mathbf{v}) \in \F_{q}^{m}$, the restriction $P(\mathbf{u} + t \cdot \mathbf{v})$ is a univariate polynomial $P_{\mathbf{u},\mathbf{v}}: \F_{q} \to \F_{q}$ of degree-$c$. From the above discussion, we get that
	\begin{align*}
		\sum_{i = 0}^{c+1} \eta_{i} \cdot P_{\mathbf{u}, \mathbf{v}}(\zeta_{i}) = 0.
	\end{align*}
	Now it remains to argue the reverse direction. Suppose $f: \F_{q}^{m} \to \F_{q}$ is a function that satisfies the property stated in \Cref{fact:characterize-low-deg} for every pair $(\mathbf{u}, \mathbf{v})$. Fix an arbitrary $(\mathbf{u}, \mathbf{v})$ and let $f_{\mathbf{u}, \mathbf{v}}(t) = f(\mathbf{u} + t \cdot \mathbf{v})$ be the univariate restriction to line $\ell_{\mathbf{u}, \mathbf{v}}$. We will first show that $f_{\mathbf{u}, \mathbf{v}}$ is degree-$c$, by relying on the above discussion. Choose any $\alpha, \beta \in \F_{q}$ and consider the pair $(\mathbf{u} + \alpha \mathbf{v}, \beta \mathbf{v})$. Then, we have
	\begin{align*}
		\sum_{i=0}^{c+1} \eta_{i} \cdot f_{\mathbf{u}, \mathbf{v}}(\alpha + \zeta_{i} \cdot \beta) \; = \;  \sum_{i=0}^{c+1} \eta_{i} \cdot f(\mathbf{u} + (\alpha + \zeta_{i} \cdot \beta) \cdot \mathbf{v}) \; = \; 0.
	\end{align*}
	Therefore, $f_{\mathbf{u},\mathbf{v}}(t)$ is of degree-$c$. Since our choice of $(\mathbf{u}, \mathbf{v})$ was arbitrary, this holds for every pair. Now \cite[Lemma 4]{Friedl-Sudan} implies that $f$ is indeed a degree-$c$ polynomial. This finishes the reverse direction for arbitrary $m$ and also finishes the proof of \Cref{fact:characterize-low-deg}.
\end{proof}

\begin{proof}[Proof of~\Cref{claim:eval-set-mult-degree-redn}]
	For an integer $0\le k\le d$, let $[k_{c-1},\dots, k_0]$ denote its $m_1$-ary representation. Letting $P(Y)=\sum_{k=0}^d \alpha_k\cdot Y^k$, we have
	\begin{align*}\setml_{d,c,m_{1}}(P)                            & = \sum_{k=0}^d \alpha_k\cdot X_{0,k_0}\cdot X_{1,k_1}\cdots X_{c-1,k_{c-1}}
              \intertext{which then implies}
              (\setml_{d,c,m_{1}}(P))(\Phi_{c,m_{1}}(\lambda)) & = \sum_{k=0}^d \alpha_k\cdot \lambda^{k_0}\cdot \lambda^{k_1\cdot m_1} \cdots \lambda^{k_{c-1}\cdot m_1^{c-1}} \\
                                                               & = \sum_{k=0}^d \alpha_k\cdot \lambda^{k} = P(\lambda).\end{align*}
	This finishes the proof of \Cref{claim:eval-set-mult-degree-redn}.
\end{proof}

\begin{proof}[Proof of \Cref{claim:checking-relation-via-Hadamard}]
	We first show it when \(L\) is homogeneous so
	\(L(x + y) = L(x) + L(y)\).
	Let \(v_{i} := \rho^{-1}(e_{i})\),
	where \((e_{i})\) is the standard basis for \(\F_{2}^{t}\).
	Let
	\begin{align*}
		P(Y) = \sum_{k=0}^{r} c_{k} Y^{k}.
	\end{align*}

	In this case we have for any element \(\lambda = \sum_{i \in [t]} \alpha_{i} v_{i} \in \F_{q}\),
	\begin{align*}
		\had{1}{P(\lambda)}(L) =
		L(\rho(P(\lambda))) =
		\sum^{r}_{k=0} L(\rho(c_{k} \lambda^{k}))
		=
		\sum^{r}_{k=0} \sum_{i_{1},\ldots, i_{k} \in [t]} L(\rho(c_{k} v_{i_{1}} \cdots v_{i_{k}}) )\alpha_{i_{1}}\cdots\alpha_{i_{k}}.
	\end{align*}
	We note that this expression is a degree \(r\) polynomial in the \(\alpha_{i}\)'s
	and define this to be \(\Lambda_{P,L}(\alpha_{1}, \ldots, \alpha_{t})\).
	We then have
	\begin{align*}
		\had{r}{\lambda}(\Lambda_{P,L})
		= \Lambda_{P,L}(\rho(\lambda))
		= \Lambda_{P,L}(\alpha_{1}, \ldots, \alpha_{t}),
	\end{align*}
	showing it has the wanted property.
	If \(P(Y) = \eta \cdot Y\),
	then
	\begin{align*}
		\Lambda_{P,L}(\alpha_{1}, \ldots, \alpha_{t}) =
		\sum_{i \in [t]}
		L(\rho(\eta v_{i}) )\alpha_{i},
	\end{align*}
	showing that it is homogeneous in this case.

	For \(L' = L + b\) a sum of a homogeneous linear	function and a constant \(b \in \F_{2}\),
	we set \(\Lambda_{P,L'} := \Lambda_{P,L}+b\).
	We then have
	\begin{align*}
		\had{1}{P(\lambda)}(L') =\had{1}{P(\lambda)}(L) + b
		= \had{r}{\lambda}(\Lambda_{P,L}) + b
		= \had{r}{\lambda}(\Lambda_{P,L'}).
	\end{align*}
	This finishes the proof of \Cref{claim:checking-relation-via-Hadamard}.
\end{proof}

\section{Analysis of \texorpdfstring{$\mathsf{HadTest}$}{HadTest}}\label[appendix]{app:syntactic-test-gen-Hadamard}

To prove \Cref{lemma:syntactic-test-gen-Hadamard}, we will use the local testability of the Hadamard code.\\

\begin{theorem}[BLR Linearity Testing~\cite{BLR,BLR2}]\label{thm:BLR}
	Let $r \in \mathbb{N}$ denote the degree parameter, and let
	$H: \F_{2}[Z_{1},\ldots,Z_{t}]^{\leq r} \to \F_{2}$ be arbitrary. Let
	\begin{align*}
		\varepsilon \; := \; \Pr_{R, R' \, \sim \, \F_{2}[\mathbf{Z}]^{\leq r} }[H[R + R'] \; \neq \; H[R] + H[R']]
	\end{align*}
	denote the rejection probability of the linearity test. Then the following holds:
	\begin{enumerate}
		\item If $H$ is $\had{r}{u}$ for some $u \in \F_{q}$, then $\varepsilon = 0$.

		\item There is a linear function
		      \(\mathcal{L}: \F_{2}[\mathbf{Z}]^{\leq r} \to \F_{2}\)
		      with $\delta(H, \mathcal{L}) \leq \varepsilon$.
	\end{enumerate}
\end{theorem}

\noindent
\Cref{thm:BLR} relates the rejection probability to distance for degree-$1$ Hadamard codes.
Using this, we will also prove \Cref{lemma:syntactic-test-gen-Hadamard}.

\begin{proof}[Proof of \Cref{lemma:syntactic-test-gen-Hadamard}]
	For completeness of the algorithm,
	we observe that if $H = \had{r}{u}$ for some $u \in \F_{q}$,
	then $H$ passes every test in $\hadtest$
	and thus $\hadtest^{H}_{r}$ always returns $\accept$.

	We now focus on the soundness of $\hadtest^{H}_{r}$.
	The following events can make the algorithm reject:
	\begin{itemize}
		\item Event \(\mathcal{E}_{\mathrm{lin}}\) (line 2):
		      \(H[R + R'] \neq H[R] + H[R']\).
		\item Event \(\mathcal{E}_{\mathrm{unit}}\) (line 3):
		      \(H[1 + R] - H[R] \neq 1\).
		\item Event \(\mathcal{E}_{\deg 1}\) (line 4):
		      \(H[L + R] - H[R] \neq H[L]\).
		\item Event \(\mathcal{E}_{\mathrm{mult}}\)
		      (line 5):
		      \begin{align*}
			      H[R + L \cdot P] - H[R] \; \neq \;
			      (H[R+L] - H[R]) \cdot (H[R+P]-H[R]).
		      \end{align*}
	\end{itemize}
	We will first show that if
	\(\mathcal{E}_{\mathrm{lin}}\), \(\mathcal{E}_{\mathrm{unit}}\),
	and \(\mathcal{E}_\mathrm{mult}\) all happen with
	probability less than \(\delta\), then \(H\) is \(\delta\)-close to a
	degree-\(r\) Hadamard encoding \(\had{r}{u}\) for some \(u \in \F_{q}\).

	First note that if event \(\mathcal{E}_{\mathrm{lin}}\)
	happens with probability less than \(\delta\),
	then by \Cref{thm:BLR},
	there exists a linear function
	\(\mathcal{L}: \F_{2}[Z_{1}, \ldots, Z_{t}]^{\leq r} \to \F_{2}\),
	such that \(\delta(H, \mathcal{L}) \leq \delta\).
	By \Cref{claim:local-characterize-had},
	if \(\mathcal{L}(1)=1\) and \(\mathcal{L}(P \cdot R) = \mathcal{L}(P)
	\cdot \mathcal{L}(R)\)
	for \(\deg P + \deg R \leq r\), then \(\mathcal{L}\)
	is the degree-\(r\) Hadamard encoding \(\had{r}{u}\) for some \(u \in \F_{q}\).

	\paragraph*{Unit condition}
	By a union bound, we have
	\begin{align*}
		\Pr_{R\, \sim \, \F_{2}[\mathbf{Z}]^{\leq r}}
		[ \mathcal{L}(1 + R) \neq H[1+R] \; \vee \;
			\mathcal{L}(R) \neq H[R]] \leq 2\delta.
	\end{align*}
	Since the event \(\mathcal{E}_{\mathrm{unit}}\) happens with probability
	less than \(\delta\), we also have
	\begin{align*}
		\Pr_{R\, \sim \, \F_{2}[\mathbf{Z}]^{\leq r}}
		[H[1+R] - H[R] \neq 1] \leq \delta.
	\end{align*}
	Since \(\delta < \frac{1}{3}\),
	this implies that \(\mathcal{L}(1)=1\).

	\paragraph*{Multiplicative condition}
	For any function \(f: \F_{2}[\mathbf{Z}]^{\leq r} \to \F_{2}\),
	define the function
	\begin{gather*}
		A_{f}:\F_{2}[\mathbf{Z}]^{\leq 1} \times \F_{2}[\mathbf{Z}]^{\leq r-1} \quad \to \quad \F_{2}  \\
		(L,P) \quad \mapsto \quad f(L \cdot P) - f(L) \cdot f(P).
	\end{gather*}
	If \(f\) is a linear function, then \(A_{f}\) is a
	bilinear form in the coefficients of \(L\) and \(P\).
	By the polynomial distance lemma over small fields (\Cref{thm:odlsz}),
	if \(A_{\mathcal{L}}\) is non-zero, then
	it must be non-zero on at least \(\frac{1}{4}\) of the inputs.

	Using $\delta(\mathcal{L}, H) \leq \delta$ and union bound, we get,
	\begin{align*}
		\Pr_{L, R}[\mathcal{L}(L) \neq
		H(L + R) - H(R)] \leq 2 \delta \\
		\Pr_{P, R}[\mathcal{L}(P) \neq
		H(P + R) - H(R)] \leq 2 \delta \\
		\Pr_{L, P, R}[\mathcal{L}(L \cdot P) \neq
			H(L \cdot P + R) - H(R)] \leq 2 \delta,
	\end{align*}
	since all the points \(H\) is evaluated on are marginally uniformly random
	polynomials of degree at most \(r\).
	It follows that
	\begin{gather*}
		\Pr_{L,P}[\mathcal{L}(L \cdot P) \neq \mathcal{L}(L) \cdot \mathcal{L}(P)] \\
		\leq 6 \cdot \delta +
		\Pr_{L,P, R}[\underbrace{H[R + L \cdot P] + H[R] \; \neq \;  (H[R+L] + H[R]) \cdot (H[R+P] + H[R])}_{\mathcal{E}_{\mathrm{mult}}}].
	\end{gather*}
	In the above inequality, the second summand is exactly the probability of the event $\mathcal{E}_{\mathrm{mult}}$, which we assume to be at most $\delta$.
	If we then assume \(A_{\mathcal{L}}\) is non-zero,
	then we have
	\begin{align*}
		\dfrac{1}{4} \; \leq \; \Pr_{L,P}[\mathcal{L}(L \cdot P) - \mathcal{L}(L) \cdot \mathcal{L}(P) \neq 0] \; \leq \; 7 \delta.
	\end{align*}
	If we assume that \(\delta < \frac{1}{28}\),
	we get a contradiction, so we have
	\(\mathcal{L}(L \cdot P) = \mathcal{L}(L) \cdot \mathcal{L}(P)\)
	for any \(L \in \F_{2}[\mathbf{Z}]^{\leq 1}, P \in \F_{2}[\mathbf{Z}]^{\leq r-1}\).
	We can then extend this to
	\(\mathcal{L}(P_{1} \cdot P_{2}) = \mathcal{L}(P_{1}) \cdot \mathcal{L}(P_{2})\)
	for any \(P_{1}, P_{2} \in \F_{2}[\mathbf{Z}]^{\leq r-1}\)
	with \(\deg P_{1} + \deg P_{2} \leq r\),
	by using that any polynomial \(P_{1}\) can be written as a sum of products of linear polynomials.
	This shows that \(\mathcal{L}\) satisfies the conditions to being the
	Hadamard encoding of some \(u\), finishing the first part of the lemma.

	For the second part, we show that the restriction of \(H\) to degree 1 is also close
	to \(\had{1}{u}\).
	Since \(\delta(H, \had{r}{u}) \leq \delta\),
	we have for any \(L \in \F_{2}[\mathbf{Z}]^{\leq 1}\)
	\begin{align*}
		\Pr_{R}\brac{ \had{r}{u}(L + R) \neq H(L + R) \; \vee \; \had{r}{u}(R) \neq H(R) } \; \leq \; 2 \delta.
	\end{align*}
	We then have
	\begin{align*}
		\Pr_{L}\brac{ H[L] \neq \had{1}{u}(L) } \;
		\leq \; 2 \delta \, + \,
		\Pr_{L,R}[ \; \underbrace{H[L] \neq H[L + R] + H[R]}_{\mathcal{E}_{\deg 1}} \;] \;
		\leq  \; 3 \delta.
	\end{align*}
	This finishes the proof of \Cref{lemma:syntactic-test-gen-Hadamard}.
\end{proof}

\section{Proof of Local Correction}\label[appendix]{app:proof-local-correction}
\begin{proof}[Proof of \Cref{thm:local-corr-F2}]
	For completeness, the argument is almost verbatim to the completeness argument of \Cref{algo:new-low-deg-test-binary} (see the proof of \Cref{thm:new-low-degree-test}).
	Lines 2, 3, and 4 always return $\accept$ and Line 6 returns
	$$
		\sum_{i=1}^{c+1} f'[(\mathbf{a}, \mathbf{b}), \Phi(0) + \zeta^{i} \mathbf{v}, \etail] = \sum_{i=1}^{c+1} \had{1}{\Psi^*(P_{\mathrm{lines}})}[(\mathbf{a}, \mathbf{b}), \Phi(0) + \zeta^{i} \mathbf{v}, \etail] ,
	$$
	which we show is equal to $\had{1}{P}(\veca,L)$. We note that since $\Psi(P_{\mathrm{lines}}(\veca,\vecb))$ is a degree-$c$ polynomial, using~\Cref{fact:local-correction-via-lines} we have that
	$$\sum_{i=1}^{c+1} \eta_{i} \Psi(P_{\mathrm{lines}}(\veca,\vecb))(\Phi(0)+\zeta^i \mathbf{v}) = \Psi(P_{\mathrm{lines}}(\veca,\vecb))(\Phi(0)) = P_{\mathrm{lines}}(\veca,\vecb)(0) = P(\veca).$$ Hence, using that $\etail(\rho(x)) = L(\rho(\eta_{i} x))$ (\Cref{claim:checking-relation-via-Hadamard}) and linearity of $\rho$ and $L$, we conclude that $\corr$ indeed returns $L(\rho(P(\veca))) = \had{1}{P}(\mathbf{a}, L)$.

	\paragraph*{Soundness:}
	Assume there exists a polynomial \(P \in \mathcal{P}_{d'}(m,\F_{q})\)
	such that $f$ is $\delta$-close to $\had{1}{P}$.
	Let \(\eta\) be the probability that \(\corr^{f,f'}(\ba, L)\) rejects.
	We first introduce some notation similar to \Cref{thm:new-low-degree-test},
	which will be useful for the proof.

	\textbf{Notation:}
	For every pair $(\mathbf{a}, \mathbf{b}) \in \F_{q}^{2m}$,
	for every $\mathbf{u} \in \F_{q}^{cm_{1}}$,
	and for every $L \in \hpol$,
	we will use
	$\stdcorr^{f'[(\mathbf{a}, \mathbf{b}), \cdot, L]}_{c}(\mathbf{u}; \mathbf{v})$
	to denote the following sum:
	\begin{align*}
		\stdcorr^{f'[(\mathbf{a}, \mathbf{b}), \cdot, L]}_{c}(\mathbf{u}; \mathbf{v})
		\; := \;
		\sum_{i = 1}^{c+1} f'[(\mathbf{a}, \mathbf{b}), \, \mathbf{u} + \zeta^{i} \mathbf{v}, \, \etail].
	\end{align*}
	Note that $\corr^{f,f'}_{d,c}(\mathbf{a}, L ; \mathbf{b}, \mathbf{u}, \mathbf{v}, \lambda)$ returns $\stdcorr^{f'[(\mathbf{a}, \mathbf{b}), \cdot, L]}_{c}(\Phi(0); \mathbf{v})$,
	if it does not reject.
	We define \(Q_{\ba, \bb}\) as the degree-\(c\) polynomial,
	whose degree-\(1\) Hadamard encoding is	closest to
	\(f'[(\ba, \bb), \cdot, \cdot]\),
	and \(F(\ba, \bb) = Q_{\ba, \bb} \circ \Phi\).
	Note that $\deg (F(\ba, \bb)) \leq c (m_{1}-1)m_{1}^{c-1} \leq c m_{1}^{c} \leq d'.$

	There are three events which can make the test reject:
	\begin{itemize}
		\item Event $\mathcal{E}_{1}$: $\hadtest^{f'[(\mathbf{a}, \mathbf{b}), \bu, \cdot]}_{1}$ returns \reject.
		\item Event $\mathcal{E}_{2}$: $f'[(\mathbf{a}, \mathbf{b}), \mathbf{u}, L]+\sum_{i=1}^{c+1} f'[(\mathbf{a}, \mathbf{b}), \mathbf{u} + \zeta^{i} \mathbf{v}, \etail]$ is not equal to $0$.
		\item Event $\mathcal{E}_{3}$: $\sum_{i=1}^{c+1} f'[(\mathbf{a}, \mathbf{b}), \Phi(\lambda) + \zeta^{i} \mathbf{v}, \etail]$ is not equal to $f[\ba + \lambda \bb, L]$.
	\end{itemize}

	If the first two events happen with low probability,
	then by \Cref{lemma:std-ldt-hadamard-encoding},
	with high probability \(f'[(\ba, \bb), \cdot , \cdot]\)
	must be close to the Hadamard encoding of a degree-$c$  multivariate polynomial, which by definition is \(Q_{\ba, \bb}\).
	This then implies that we can do standard local correction on
	\(f'[(\ba, \bb), \Phi( \lambda ) , \cdot ]\)
	to get the value of
	$\had{1}{Q_{\mathbf{a},\mathbf{b}}}[\Phi( \lambda ), \cdot ]$ for any $\lambda \in \F_{q}$.
	We state it more formally in the next claim.

	\begin{claim}\label{clm:lcf2}
		Suppose that \(\corr^{f,f'}_{c}(\ba, L)\) rejects
		with probability less than \(\eta\)
		for a random \(L\).
		Then for any \(\lambda \in \F_{q}\) we have
		\begin{align*}
			\Pr_{\bb, \bv, L} \left[
			\stdcorr^{f'[(\ba, \bb), \cdot, L]}_{c}( \Phi(\lambda) ; \bv)
			\neq
			\had{1}{F(\ba, \bb)}(\lambda, L)
			\right]
			\leq
			\mathcal{O}_{c}(\eta),
		\end{align*}
		where the probability is over
		$\bb \sim \F_{q}^{m}$,
		$\bv \sim \F_{q}^{cm_{1}}$, and
		$L \sim \hpol$.
	\end{claim}
	%

	%
	We assume the claim for now,
	and prove it at the end of the section.

	If the event \(\mathcal{E}_{3}\) also happens with low probability,
	then the local correction of \(f'\) is close to \(f\),
	which also is close to \(\had{1}{P}\) from our assumption.
	We then get the following bound by the
	union bound
	for a random \(\lambda \in \F_{q}^{\times}\)
	\begin{align*}
		             & \Pr_{\phantom{\bv,}\bb, \lambda, L} \left[\had{1}{F(\ba, \bb)}(\lambda, L)
		\neq \had{1}{P}(\ba+\bb \lambda, L) \right]                                               \\
		\;\leq \;    &
		\Pr_{\bb, \bv, \lambda, L} \left[
		\had{1}{F(\ba, \bb)}(\lambda, L)
		\neq
		\stdcorr^{f'[(\ba, \bb), \cdot, L]}_{c}( \Phi(\lambda) ; \bv)
		\right]                                                                                   \\
		\;+\;        &
		\Pr_{\bb, \bv, \lambda, L} \left[
		\stdcorr^{f'[(\ba, \bb), \cdot, L]}_{c}( \Phi(\lambda) ; \bv)
		\neq f[\ba + \lambda \bb, L] \right]                                                      \\
		\;		+   \;   &
		\Pr_{\phantom{\bv,} \bb,\lambda, L} \left[
			f[\ba + \lambda \bb, L]
		\neq \had{1}{P}(\ba + \lambda \bb, L) \right]                                             \\
		\;	\leq   \; &
		\mathcal{O}_{c}( \eta) + \eta + 2 \delta.
	\end{align*}
	The first summand is bounded by \Cref{clm:lcf2}.
	The second summand is one of the rejection criteria for the test
	so it is bounded by \(\eta\).
	The last summand is bounded by \(2\delta\),
	by the assumption that \(\had{1}{P}\)
	is \(\delta\)-close to \(f\),
	and half of linear polynomials are homogeneous.
	Applying \Cref{obs:delta-fn-had-approx},
	we then get
	\begin{align*}
		\Pr_{\bb, \lambda} \left[
			F(\ba, \bb)(\lambda)
			\neq
			P(\ba+\bb \lambda)
			\right]
		\;\leq \;
		\mathcal{O}_{c}(\eta + \delta).
	\end{align*}
	Apply \Cref{thm:local-correction}\footnote{One can also simply apply~\Cref{thm:odlsz} here instead.} on \(P\) and \(F\) to get
	\begin{align*}
		\Pr_{\mathbf{b}, \lambda}\brac{ F(\ba, \bb)(0)  = P(\mathbf{a})
			\quad \text{OR} \quad
			F(\ba, \bb)(\lambda) \neq P(\ba + \bb \lambda)}
		\; \geq \;
		1- \dfrac{d'}{q-1}.
	\end{align*}
	Combining the two inequalities together with \Cref{obs:delta-fn-had-approx},
	we get
	\begin{align*}
		\Pr_{\mathbf{b}, L}\brac{
			\had{1}{F(\ba, \bb)}(0, L)
			\neq \had{1}{P}(\mathbf{a}, L)}
		\; \leq \;
		\mathcal{O}_{c}( d'/q + \eta + \delta).
	\end{align*}
	Then applying \Cref{clm:lcf2} with \(\lambda = 0\)
	we have
	\begin{align*}
		\Pr_{\mathbf{b}, \bv, L}\brac{
		\stdcorr^{f'[(\ba, \bb), \cdot, L]}_{c}(\Phi(0) ; \bv)
		\neq
		\had{1}{P}(\mathbf{a}, L)}
		\; \leq \;
		\mathcal{O}_{c}( d'/q + \eta + \delta) + \mathcal{O}_{c}(\eta)
		=
		\mathcal{O}_{c}( d'/q + \eta + \delta).
	\end{align*}
	Since \(\corr\) returns \(\stdcorr^{f'[(\ba, \bb), \cdot, L]}_{c}(\Phi(0) ; \bv)\)
	whenever it does not reject, and it rejects with probability at most \(\eta\),
	a union bound gives
	\begin{align*}
		\Pr\brac{\corr^{f,f'}_{d,c}(\ba, L) \text{ does not return } \had{1}{P}(\ba, L)}
		\; \leq \;
		\eta + \mathcal{O}_{c}( d'/q + \eta + \delta).
	\end{align*}
	This finishes the proof of \Cref{thm:local-corr-F2}, assuming the proof of \Cref{clm:lcf2}.
\end{proof}
We end by proving \Cref{clm:lcf2}.

\begin{proof}[Proof of \Cref{clm:lcf2}]
	For \(\bb \in \F_{q}^{m}\), set
	\begin{align*}
		\rej_{\bb} :=
		\Pr_{L, \bu, \bv, \lambda}
		\left[\corr^{f,f'}_{c}(\ba, L ; \bb, \bu, \bv, \lambda) \text{ returns \reject }\right].
	\end{align*}
	Now fix \(\bb \in \F_{q}^{m}\).
	We have
	\(\Pr[\mathcal{E}_{1} | \bb],\Pr[\mathcal{E}_{2} | \bb] \leq \rej_{\bb}\),
	so applying \Cref{lemma:std-ldt-hadamard-encoding}
	to \(f'[(\ba, \bb), \cdot, \cdot]\), we get that for every $\mathbf{u} \in \F_{q}^{cm_{1}}$,
	\begin{align*}
		\Pr_{\bv, L} \left[
		\stdcorr^{f'[(\ba, \bb), \cdot, L]}_{c}(\bu ; \bv)
		\neq
		\had{1}{Q_{\ba, \bb}}(\bu, L)
		\right]
		\; \leq \;
		\mathcal{O}_{c}(\rej_{\bb}).
	\end{align*}
	We will use this inequality with \(\bu = \Phi(\lambda)\).
	In this case by definition we have \(F(\ba, \bb)(\lambda) = Q_{\ba,
			\bb}(\Phi(\lambda))\).
	Averaging over \(\bb\)
	and using that
	\(\mathbb{E}_{\bb}[\rej_{\bb}] \leq \eta\)
	we then get
	\begin{align*}
		\Pr_{\bb, \bv, L}
		\left[
		\stdcorr^{f'[(\ba, \bb), \cdot, L]}_{c}(\Phi(\lambda) ; \bv)
		\neq
		\had{1}{F(\ba, \bb)(\lambda)}(L)
		\right]
		\leq
		\mathcal{O}_{c}(\eta),
	\end{align*}
	giving the wanted claim.
\end{proof}

\section{Proof of Zero-on-Grid Test}\label[appendix]{app:proof-zero-test}
\begin{proof}[Proof of \Cref{lemma:zero-on-variety}]
	We start by discussing completeness.

	\paragraph*{Completeness}
	Line $2$ and Line $3$ always pass because of \Cref{coro:meta-vanishing-poly} and the completeness property of $\corr$.

	\paragraph*{Soundness}
	We will consider the following events (we will hide the internal randomness of $\corr$ for the sake of clarity in writing):
	\begin{enumerate}
		\item Event $\mathcal{E}_{1}$:
		      $\corr^{\Pi, \Pi'}_{d}((\mathbf{a}, \mathbf{0}), L)$
		      returns \reject{}
		      or is not equal to $0$, and
		\item Event $\mathcal{E}_{2}$:
		      $\corr^{\Pi, \Pi'}_{d}((\mathbf{a}, Z_{H}(\mathbf{a})), L)$
		      returns \reject{}
		      or is not equal to $f[\mathbf{a}, L]$.
	\end{enumerate}
	Suppose the probability of $\mathcal{E}_1 \vee \mathcal{E}_2$ is at most $\eta$.
	We will then show that \(P|_{H^{m}} \equiv 0\).
	To do this, we show that the polynomial \(M \in \mathcal{P}_{d'}(2m, \F_{q}) \)
	from the soundness hypothesis of \Cref{lemma:zero-on-variety} satisfies the properties from
	\Cref{coro:meta-vanishing-poly}, i.e.,
	\begin{enumerate}
		\item $M(\mathbf{X}, \mathbf{0}) = 0$, as polynomials.
		\item $M(\mathbf{X}, Z_{H}(\mathbf{X})) = P(\mathbf{X})$, as polynomials (here $P \in \mathcal{P}_{d'}(m, \F_{q})$ from the soundness hypothesis of \Cref{lemma:zero-on-variety}).
	\end{enumerate}
	These two conditions imply that \(P|_{H^{m}} \equiv 0\).
	To see this, note that the first condition implies that \(M\) can be written in the
	form
	\begin{align*}
		M(\bx, \by) = \sum_{i=1}^{m} M_{i}(\bx, \by) \cdot y_{i}
	\end{align*}
	for some polynomials \(M_{1}, \ldots, M_{m}\). From the second condition we then have
	\begin{align*}
		P(\bx) =
		M(\bx, Z_{H}(\bx)) =
		\sum_{i=1}^{m} M_{i}(\bx, Z_{H}(\bx)) \cdot Z_{H}^{(i)}(\bx),
	\end{align*}
	so since \(Z_{H}(\bx)\) vanishes on \(H^{m}\),
	we have \(P\) also vanishes on \(H^{m}\).
	We will now show that both conditions hold.

	For the first condition we have the following union bound:
	\begin{align*}
		\Pr_{\ba, L}\brac{
		\had{1}{M}((\mathbf{a}, \mathbf{0}), L)
		\; \neq \; 0 }
		\; \leq \; &
		\Pr_{\ba, L}\brac{
		\had{1}{M}((\mathbf{a}, \mathbf{0}), L)
		\; \neq \;
		\corr^{\Pi, \Pi'}_{d}((\mathbf{a}, \mathbf{0}), L)} \\
		\; + \;    &
		\Pr_{\ba, L}\brac{
		\corr^{\Pi, \Pi'}_{d}((\mathbf{a}, \mathbf{0}), L)
		\; \neq \;
		0
		}                                                   \\
		\; \leq \; &
		\alpha(c) \cdot \left( d'/q + \eta + \delta\right)
	\end{align*}
	where the first summand is bounded by
	\(
	\alpha(c) \cdot \left( d'/q + \eta + \delta\right)
	\)
	from \Cref{thm:local-corr-F2},
	where \(\alpha(c)\) is a large constant depending on \(c\),
	and the second summand is bounded by \(\eta\)
	by assumption that event $\mathcal{E}_{1}$ happens with probability at most \(\eta\).
	From \Cref{obs:delta-fn-had-approx},
	we then get
	\begin{align*}
		\Pr_{\ba}\brac{M(\mathbf{a}, \mathbf{0}) \; \neq \; 0 }
		\; \leq \; &
		2 \alpha(c) \cdot \left(d'/q + \eta + \delta\right)
		< 1 - \frac{d'}{q},
	\end{align*}
	where the last inequality holds under the assumption that
	\(\eta < 1/(2 \alpha(c)) -2d'/q -\delta\).
	Since $M$ is a polynomial of degree $d'$, from \Cref{thm:odlsz},
	it follows that
	\(M(\mathbf{X}, \mathbf{0}) = 0 \).
	This shows that $M$ satisfies the first condition.

	The proof that the second condition also holds for \(M\) is almost
	identical. Recall that every polynomial in $Z_{H}(\mathbf{X})$ is a degree-$|H|$ polynomial,
	so $M(\mathbf{X}, Z_{H}(\mathbf{X}))$ is a degree-$(d'|H|)$ polynomial.
	By a union bound we have
	\begin{align*}
		           & \Pr_{\ba, L}\brac{
		\had{1}{M}((\ba, Z_{H}(\ba)), L)
		\; \neq \; \had{1}{P}(\ba, L) }              \\
		\; \leq \; &
		\Pr_{\ba, L}\brac{
		\had{1}{M}((\ba, Z_{H}(\ba)), L)
		\; \neq \;
		\corr^{\Pi, \Pi'}_{d}((\ba, Z_{H}(\ba)), L)} \\
		\; + \;    &
		\Pr_{\ba, L}\brac{
		\corr^{\Pi, \Pi'}_{d}((\ba, Z_{H}(\ba)), L)
		\; \neq \;
		f[\ba, L]
		}                                            \\
		\; + \;    &
		\Pr_{\ba, L}\brac{
			f[\ba, L]
			\; \neq \;
			\had{1}{P}(\ba, L)
		}                                            \\
		\; \leq \; &
		\alpha(c) \cdot \left( d'/q + \eta + \delta\right)
	\end{align*}
	where the first two summands are bounded as above,
	with the addition that the last summand is bounded by \(2\delta\)
	by assumption.
	Using \Cref{obs:delta-fn-had-approx}, we get,
	\begin{align*}
		\Pr_{\ba}\brac{
			M(\mathbf{a}, Z_{H}(\ba))
			\; \neq \; P(\ba) }
		\; \leq \;
		2\alpha(c) \cdot \left( \frac{d'|H| }{q} + \eta + \delta\right)
		\; < \;
		1 - \dfrac{d'|H|}{q},
	\end{align*}
	where the last inequality holds under the assumption that
	\(\eta < 1/(2 \alpha(c)) -2d'|H|/q -\delta\). From \Cref{thm:odlsz},
	it then follows that
	\(M(\mathbf{X}, Z_{H}(\mathbf{X})) = P(\mathbf{X}) \) as polynomials.
	The two criteria of \Cref{coro:meta-vanishing-poly}
	are then satisfied,
	so we must have \(P|_{H^{m}} \equiv 0\) finishing the proof.
\end{proof}

\section{Honest Prover and Full PCP Verifier}\label[appendix]{app:full-verifier}

We give a complete description of the PCP verifier for the language $3$-$\mathsf{COLOR}$ in \Cref{algo:pcp-verifier-final}. After that, in \Cref{algo:compl-prover}, we give the description of an honest prover which takes as input a graph in $3$-$\mathsf{COLOR}$ and outputs proof oracles that the verifier of \Cref{algo:pcp-verifier-final} always accepts. Given a proper 3-coloring, an honest prover can construct these proof oracles in polynomial time by running~\Cref{algo:compl-prover}.

\paragraph{Comparison with~\Cref{algo:final-verifier}} We note that~\Cref{algo:pcp-verifier-final} is essentially equivalent to~\Cref{algo:final-verifier}, except here we unravel all the subroutines of~\Cref{algo:final-verifier} into a single algorithm. However, we make some minor changes to make the overall description simpler, so some ingredients in this version are redundant. In particular, we use degree-3 Hadamard encodings (and either verify these using some additional checks or simply ignore the additional information) for all the proof oracles, as opposed to using degree-1 encodings for some oracles like in~\Cref{algo:final-verifier}. For example, we incorporate the checks of $\hadtest$ into two lines (Lines 13 and 14) compactly. Comparing at a more granular level, the $\texttt{SC}$ subroutine serves the same purpose as it does in~\Cref{algo:final-verifier} and the $\texttt{LC}$ subroutine simulates $\stdcorr$ (which is used in $\vldt$ and $\vvanish$). Lines 10 to 15 correspond to the initial low-degree tests in Line 6 of~\Cref{algo:final-verifier}, Line 21 corresponds to Line 8 of~\Cref{algo:final-verifier}, and the remaining Lines 17 to 23 correspond to the zero-on-grid tests in Lines 7 and 9 of~\Cref{algo:final-verifier}. Similarly,~\Cref{algo:compl-prover} generates proof oracles that are (almost) identical to the ones constructed in the completeness proof of~\Cref{algo:final-verifier}.
In the following two algorithms, all randomness (highlighted in color \rand{red}) is sampled independently from an underlying uniform distribution. Recall that $\set{A \to B}$ refers to the set of functions from $A$ to $B$, $\cP_d(m,\F_q)$ denotes degree-$d$ functions over $\F_q^m$, and for a function $f:H^m\to \F_q$ with $H\subseteq \F_q$, $\LDE(f) \in \cP_{(|H|-1)m}(m,\F_q)$ denotes the low-degree extension of $f$ (see \Cref{fact:low-deg-extension}). We have also implemented the following algorithms in Python, and the code (with a demo) can be found at \cite{Github-repo-PCP-verifier}.

\newgeometry{top=1in, bottom=1in, left=0.7in, right=0.7in}
\begin{algobox}
	\begin{algorithm}[H]
		\caption{The PCP Verifier}
		\label{algo:pcp-verifier-final}

		\DontPrintSemicolon

		\SetKwProg{Fn}{Function}{}{}

		{\nonl {\bf Constants:}}
		Constants $c,c_1,c_2\in \N$. \;
		{\nonl {\bf Parameters:}} Given $n\in \N$, let $h:=\lceil \log n\rceil$, $m:=\lceil\log n/\log \log n\rceil$,
		$t := 2\lceil c_1 \log(hm)\rceil$, $D:=c_2hm$, and $m_1:=\lfloor D^{1/c}\rfloor+1$. Let $q:=2^t$, $H\subseteq \F_q$ be of size $h$, $\rho:\F_q \to \F_2^t$ be an $\F_2$-linear bijection, $\omega,\zeta\in \F_q^{\times}$ be of order $3$ and a generator respectively, and $\eta_i$'s as in~\Cref{fact:characterize-low-deg}. Let $\mathcal{F}(m')$ be $\{ \F_{q}^{m'} \times \mathcal{P}_{3}(t,\F_{2}) \to \F_{2} \}$ and $\mathcal{F}_{1}(m')$ be $\{ \F_{q}^{m'} \times \F_{q}^{cm_{1}} \times \mathcal{P}_{3}(t,\F_{2}) \to \F_{2} \}$.

		\vspace{1mm}

		\KwIn{Graph $G=(V,E)$ over $n$ vertices, where $V=H^m$.}

		\vspace{1mm}

		{\nonl {\bf Oracles:} $\chi \in \mathcal{F}(m)$, $\lines{\chi} \in \mathcal{F}_{1}(2m)$, $\chi' \in \mathcal{F}(2m)$, $\lines{\chi'} \in \mathcal{F}_{1}(4m)$, $\zero{A} \in \mathcal{F}(2m)$, $\lines{\zero{A}} \in \mathcal{F}_{1}(4m)$, $\zero{B} \in \mathcal{F}(4m)$, $\lines{\zero{B}} \in \mathcal{F}_{1}(8m)$.}

		\vspace{3mm}

		Sample $\rand{\veca},\rand{\vecb}\sim  \F_q^m$, $\rand{\bm{\alpha}} \sim \F_{q}^{2m}$, $\rand{\bm{\beta}} \sim \F_{q}^{4m}$, $\rand{\mathbf{u}}, \rand{\mathbf{v}} \sim \F_{q}^{cm_{1}}$, $\rand{\lambda} \sim \F_{q}^{\times}$, $\rand{P_{i}} \sim \cP_{i}(t,\F_{2})$ for $0 \leq i \leq 3$, $\rand{R} \sim \cP_3(t,\F_2)$, \(\rand{s} \sim \F_{2}\), and $\rand{L}\sim \cP_1(t,\F_2)$ with $L(0) = 0$\;

		For $\gamma \in \F_{q}$, interpolate $\Lambda^1_{\gamma} \in \cP_{1}(t,\F_{2})$ s.t. $\Lambda^1_{\gamma}(z) = \rand{L} \circ \rho \circ (\gamma Y) \circ \rho^{-1}(z)~\forall z \in \F_2^t$ \;

		\SetKwBlock{Subroutines}{Subroutines}{}
		\SetKwFunction{SC}{$\texttt{SC}$}
		\SetKwFunction{LC}{$\texttt{LC}$}

		\Subroutines{
			\SC{$g: \cP_{3}(t,\F_{2}) \to \F_{2}$, \, $P \in \cP_{3}(t,\F_{2})$}{  \KwRet{ $g[P+\rand{R}] - g[\rand{R}]$ }}

			\vspace{1mm}

			\LC{$\theta: \F_{q}^{cm_{1}} \times \cP_{3}(t,\F_{2})$, \, $\mathbf{u_0} \in \F_{q}^{cm_{1}}$}{  \KwRet{ $\sum_{i=1}^{c+1} \theta[\mathbf{u}_0 + \zeta^{i} \rand{\mathbf{v}}, \Lambda^1_{\eta_i}] $ }}
		}

		$Z_{H, \rand{\veca}} \gets ((\prod_{\gamma\in H} (\rand{a_i}-\gamma))_{i=1}^m)$ and $Z_{H, \rand{\veca}, \rand{\vecb}} \gets ((\prod_{\gamma\in H} (\rand{a_i}-\gamma))_{i=1}^{m}, (\prod_{\gamma\in H} (\rand{b_{i}}-\gamma))_{i=1}^{m})$ \;

		$\Phi_{\rand{\lambda}} \gets (1,\rand{\lambda}^{m_1^i},\dots,\rand{\lambda}^{(m_1-1)\cdot m_1^i})_{0\le i<c}$ and $\Phi_{0} \gets (1,0,\dots,0)_{0\le i<c}$ \;

		For $\gamma \in \F_{q}$, interpolate $\Lambda_{\gamma} \in \cP_{3}(t,\F_{2})$ s.t. $\Lambda_{\gamma}(z) = \rand{L} \circ \rho \circ (\gamma Y^{3} - \gamma) \circ \rho^{-1}(z)~\forall z \in \F_2^t$ \;

		\vspace{1mm}

		$\theta_{A,1} \gets A_{01}[(\rand{\mathbf{a}}, \mathbf{0}, \rand{\bm{\alpha}}), \cdot, \cdot]$ and $\theta_{A,2} \gets A_{01}[(\rand{\mathbf{a}}, Z_{H, \rand{\veca}}, \rand{\bm{\alpha}}), \cdot, \cdot]$ \; $\theta_{B,1} \gets B_{01}[(\rand{\mathbf{a}}, \rand{\mathbf{b}}, \mathbf{0}, \rand{\bm{\beta}}), \cdot, \cdot]$ and $\theta_{B,2} \gets B_{01}[(\rand{\mathbf{a}}, \rand{\mathbf{b}}, Z_{H, \rand{\veca}, \rand{\vecb}}, \rand{\bm{\beta}}), \cdot, \cdot]$ \;

		\vspace{1mm}

		\tikz[remember picture,overlay]\node (Astart) {};
		\For{$(\Pi,\lines{\Pi}) \in \{(\chi,\chi_{1}), (\chi', \chi'_{1}), (A_{0}, \lines{A_{0}}), (B_{0}, \lines{B_{0}}) \}$ where $\Pi \in \mathcal{F}(m')$ and $\lines{\Pi}\in \mathcal{F}_{1}(2m') $}{

		Sample $(\randsecond{\veca'},\randsecond{\vecb'}) \sim \F_q^{2m'}$, let $f\gets \Pi[\randsecond{\veca'},\cdot]$ and $f'\gets {\lines{\Pi}}[(\randsecond{\veca'},\randsecond{\vecb'}),\cdot,\cdot]$ \;

		Is \SC{$g, \rand{P_{i}} + \rand{s}$} $\eqq$ $g[\rand{P_{i}}] + \rand{s}$ for $0 \leq i \leq 3$ and $g \in \set{f,f'[\rand{\mathbf{u}}, \cdot]}$ \;

		Is \SC{$g, \rand{L} \cdot \rand{P_{i}}$} $\eqq$ \SC{$g, \rand{L}$} $\cdot$ \SC{$g, \rand{P_{i}}$} for $i \in \set{1,2}$ and $g \in \set{f,f'[\rand{\mathbf{u}}, \cdot]}$ \;

		Is $f'[\rand{\mathbf{u}}, \Lambda^1_{\eta_0}] + \sum_{i=1}^{c+1} f'[\rand{\mathbf{u}} + \zeta^{i} \rand{\mathbf{v}}, \Lambda^1_{\eta_i}] \, \eqq \, 0$   \;

		Is $\sum_{i=1}^{c+1}  f'[\Phi_{\rand{\lambda}} + \zeta^{i} \rand{\mathbf{v}}, \Lambda^1_{\eta_i}] \, \eqq \, \Pi[\randsecond{\mathbf{a}'} + \rand{\lambda} \randsecond{\mathbf{b}'} , \rand{L}]$}
		\tikz[remember picture,overlay]\node (Aend) {};
		\begin{tikzpicture}[remember picture,overlay]
			\coordinate (Atop) at ($(Astart)+(-1.3cm,0.8ex)$);
			\coordinate (Abottom) at ($(Atop |- Aend)+(0,-0.8ex)$);

			\draw[decorate,
				decoration={brace,mirror,amplitude=4pt},
				thick,
				purple] (Atop) -- (Abottom);
			\node[rotate=90,anchor=south]
			at ($(Atop)!0.5!(Abottom)+(-0.2cm,0)$)
			{\small \textcolor{purple}{Low-Degree Tests}};
		\end{tikzpicture}

		\tikz[remember picture,overlay]\node (Bstart) {};
		\For{$\theta \in \set{\theta_{A,1}, \theta_{A,2}, \theta_{B,1}, \theta_{B,2}}$}{
			Is \SC{$\theta[\rand{\mathbf{u}}, \cdot], \rand{P_{i}}+\rand{s}$} $\eqq$ $\theta[\rand{\mathbf{u}}, \rand{P_{i}}] + \rand{s}$ for $0\le i\le 1$ and \LC{$\theta, \rand{\mathbf{u}}$} $\eqq \; \theta[\rand{\mathbf{u}}, \rand{L}]$ \;
		}

		Is \LC{$\theta_{A,1}, \, \Phi_{\rand{\lambda}}$} $\eqq$ $A_{0}[(\rand{\mathbf{a}}, \mathbf{0}) + \rand{\lambda} \rand{\bm{\alpha}}, \rand{L}]$ and \LC{$\theta_{A,1}, \, \Phi_{0}$} $\eqq$ $0$ \;

		Is \LC{$\theta_{A,2}, \, \Phi_{\rand{\lambda}}$} $\eqq$ $A_{0}[(\rand{\mathbf{a}}, Z_{H, \rand{\veca}}) + \rand{\lambda} \rand{\bm{\alpha}}, \rand{L}]$ and \LC{$\theta_{A,2}, \, \Phi_{0}$} $\eqq$ \SC{$\chi[\rand{\mathbf{a}}, \cdot], \Lambda_{1}$}

		Is $\chi'[(\rand{\veca},\rand{\vecb}),\rand{L}] \eqq \chi[\rand{\veca}, \rand{L}] - \chi[\rand{\vecb}, \rand{L}]$ \;

		Is \LC{$\theta_{B,1}, \, \Phi_{\rand{\lambda}}$} $\eqq$ $B_{0}[(\rand{\mathbf{a}}, \rand{\mathbf{b}}, \mathbf{0}) + \rand{\lambda} \rand{\bm{\beta}}, \rand{L}]$ and \LC{$\theta_{B,1}, \, \Phi_{0}$} $\eqq$ $0$ \;

		Is \LC{$\theta_{B,2}, \, \Phi_{\rand{\lambda}}$} $\eqq$ $B_{0}[(\rand{\mathbf{a}}, \rand{\mathbf{b}}, Z_{H, \rand{\veca}, \rand{\vecb}}) + \rand{\lambda} \rand{\bm{\beta}}, \rand{L}]$ and \LC{$\theta_{B,2}, \, \Phi_{0}$} $\eqq$ \SC{$\chi'[(\rand{\mathbf{a}}, \rand{\mathbf{b}}), \cdot], \Lambda_{\LDE(E)(\rand{\mathbf{a}}, \rand{\mathbf{b}})}$}
		\tikz[remember picture,overlay]\node (Bend) {};

		\begin{tikzpicture}[remember picture,overlay]
			\coordinate (Btop) at ($(Bstart)+(-1.3cm,0.8ex)$);
			\coordinate (Bbottom) at ($(Btop |- Bend)+(0,-0.8ex)$);

			\draw[decorate,
				decoration={brace,mirror,amplitude=4pt},
				thick,
				purple] (Btop) -- (Bbottom);
			\node[rotate=90,anchor=south]
			at ($(Btop)!0.5!(Bbottom)+(-0.2cm,0)$)
			{\small \textcolor{purple}{$3$-Coloring Tests}};
		\end{tikzpicture}

		\lIf{any of the above tests return $\reject$}{\Return{$\reject$}}
		\lElse{{\Return{$\accept$}}}
	\end{algorithm}
\end{algobox}

\newgeometry{top=1in, bottom=1in, left=0.7in, right=0.7in}
\begin{algobox}
	\begin{algorithm}[H]
		\caption{Completeness Prover}
		\label{algo:compl-prover}

		\DontPrintSemicolon

		\SetKwProg{Fn}{Function}{}{}

		{\nonl {\bf Parameters:}} Given $n,h,m,c,c_1,c_2,t,D,m_1,q \in \N$, $\omega\in \F_q$, $\rho:\F_q\to \F_2^t$, $H\subseteq \F_q$, $V=H^m$, $\mathcal{F}(\cdot)$, $\mathcal{F}_1(\cdot)$ as in~\Cref{algo:pcp-verifier-final}.\\
		\KwIn{A 3-colorable graph $G=(V,E)$ and its proper 3-coloring $\mathsf{Color}:V\to \{1,\omega,\omega^2\}$.\\}
		{\nonl {\bf Output:} $\chi \in \mathcal{F}(m)$, $\chi_1 \in \mathcal{F}_{1}(2m)$, $\chi'\in \mathcal{F}(2m), \lines{\chi'}\in \mathcal{F}_{1}(4m)$, $\zero{A} \in \mathcal{F}(2m), \lines{\zero{A}} \in \mathcal{F}_{1}(4m)$, $\zero{B} \in \mathcal{F}(4m)$, $\lines{\zero{B}} \in \mathcal{F}_1(8m)$.}\\

		\SetKwBlock{Subroutines}{Subroutines}{}

		\SetKwFunction{routinehad}{$\mathrm{Had}$}
		\SetKwFunction{routinepsi}{$\Psi$}
		\SetKwFunction{psionlines}{$\Psi^\ast$}
		\SetKwFunction{routinelines}{$\mathcal{L}$}
		\SetKwFunction{polydiv}{PolyDivide}
		\SetKwFunction{multdiv}{MultiDivide}

		\vspace{1mm}

		\Subroutines{

			\Fn{\routinelines{$P\in \cP_D(k,\F_q)$}}{
				\KwRet{%
					$\mathcal{L}(P)\in
						(\F_q^{D+1})^{q^{2k}}$
					defined by $\mathcal{L}(P)[\mathbf{a},\mathbf{b}] = P(\mathbf{a}+\mathbf{b} X)$ for $\veca,\vecb\in \F_q^k$
				}\;
			}

			\Fn{\routinepsi{$(\alpha_0,\ldots,\alpha_D)\in\F_q^{D+1}$}}{
				\For{$0\le i\le D$}{
					$\mathrm{rep}(i)\gets(i_0,\ldots,i_{c-1})$
					such that
					$i=\sum_{j=0}^{c-1}i_jm_1^j$ and $0\le i_j<m_1$\;
				}

				\KwRet{%
					$\Psi(\alpha_0,\ldots,\alpha_D)\in
						(\F_q)^{q^{cm_1}}$
					defined by $\Psi(\alpha_{0},\ldots,\alpha_{D})[\mathbf{u}] = \sum_{j=0}^{D} \alpha_{j} \cdot u_{0,\mathrm{rep}(j)[0]} \cdots u_{c-1,\mathrm{rep}(j)[c-1]}$ for $\mathbf{u} \in \F_{q}^{cm_{1}}$
				}\;
			}

			\Fn{\psionlines{$F\in(\F_q^{D+1})^{q^{2k}}$}}{
				\KwRet{
					$\Psi^\ast(F)\in
						(\F_q)^{q^{2k}\times q^{cm_1}}$
					defined by
					$\Psi^\ast(F)[\mathbf a,\mathbf b,\mathbf u]
						=$
					\routinepsi{$F[\mathbf a,\mathbf b]$}$[\mathbf u]$ for $(\mathbf{a}, \mathbf{b}) \in \F_{q}^{2k}$ and $\mathbf{u} \in \F_{q}^{cm_{1}}$
				}
			}

			\Fn{\polydiv{$P=P(\mathbf{X}),Z=Z(X_i)$}}{
				\KwRet{%
					polynomials $Q,R$ such that
					$P=Q\cdot Z+R$ and $\deg_{X_i}(R)<\deg(Z)$
				}\;
			}

			\Fn{\multdiv{$P,Z_1,\ldots,Z_k$}}{
				$\mathrm{List}\gets\emptyset$,  $R\gets P$\;

				\For{$1\le i\le k$}{
					$(Q_i,R) \gets$ \polydiv{$R,Z_i$}\;
					$\mathrm{List}\gets(\mathrm{List}, Q_i)$\;

				}
				\KwRet{$\mathrm{List}$}
			}

			\Fn{\routinehad{$f \in \{\mathcal D \to \F_q\}$}}{
				\KwRet{%
					$\mathrm{Had}(f)\in
						\F_2^{|\mathcal D|\times|\cP_3(t,\F_2)|}$
					defined by
					$\mathrm{Had}(f)[\mathbf a,P]
						=P(\rho(f(\mathbf a)))$
					for $\mathbf a\in\mathcal D$ and $P\in \cP_3(t,\F_2)$
				}\;
			}
		}

		\vspace{1mm}

		$Z_{H}(X) \gets \prod_{\gamma\in H} (X-\gamma)$

		\vspace{1mm}

		$(A^{(1)}, \ldots, A^{(m)}) \gets$ \multdiv{$\LDE${$(\mathsf{Color})$}$^{3}-1$, $Z_H(X_1),\dots, Z_H(X_m)$}{} \;


		\vspace{1mm}

		$\widetilde{\chi} \gets$ $\LDE(\mathsf{Color})(\mathbf{Y_{1}}) - \LDE(\mathsf{Color})(\mathbf{Y_{2}})$

		\vspace{1mm}

		$(B^{(1)}, \ldots, B^{(2m)}) \gets$ \multdiv{$\LDE({E})$ $\cdot (\widetilde{\chi}^{3}-1)$, $Z_{H}(X_{1}),\dots, Z_{H}(X_{2m})$  }{} \;


		\vspace{1mm}

		\Return{ \routinehad{$\LDE${$(\mathsf{Color})$}}, \routinehad{\psionlines{\routinelines{$\LDE${$(\mathsf{Color})$}}}}, \routinehad{$\widetilde{\chi}$}, \routinehad{\psionlines{\routinelines{$\widetilde{\chi}$}}}, \newline \routinehad{$\sum_{i=1}^{m} A^{(i)}(\mathbf{X}) \cdot Y_{i}$}, \routinehad{\psionlines{\routinelines{$\sum_{i=1}^{m} A^{(i)}(\mathbf{X}) \cdot Y_{i}$}}},\newline
			\routinehad{$\sum_{i=1}^{2m} B^{(i)}(\mathbf{X}) \cdot Y_{i}$}, \routinehad{\psionlines{\routinelines{$\sum_{i=1}^{2m} B^{(i)}(\mathbf{X}) \cdot Y_{i}$}} }}

	\end{algorithm}
\end{algobox}

\end{document}